\documentclass[a4paper,12pt]{article}
\usepackage[top=1.25in, bottom=1.25in, left=1.05in, right=1.05in]{geometry}
\usepackage[utf8]{inputenc}
\usepackage[T1]{fontenc}
\usepackage[english]{babel}
\usepackage{comment}

\usepackage{natbib,longtable}
\usepackage[flushleft]{threeparttable}
\usepackage{threeparttablex} 
\usepackage{eurosym}
\usepackage{tikz}

\usepackage[section]{placeins}

\usepackage{setspace}

\usepackage{enumitem}
\setlist{noitemsep}  

\usepackage{amscd,amsmath,amssymb,amsfonts,amsthm,bbm,bm,latexsym,mathrsfs,mathtools,bigints}
\usepackage[subnum]{cases}	
\usepackage{dsfont}
\makeatletter
\newcommand*{\distas}[1]{\mathbin{\overset{#1}{\kern\z@\sim}}}	
\makeatother
\allowdisplaybreaks

\newtheorem{proposition}{Proposition}

\theoremstyle{remark}

\theoremstyle{plain}

\RequirePackage[colorlinks=true,citecolor=teal,urlcolor=teal]{hyperref}

\defcitealias{acer2023demand}{ACER}
\defcitealias{iea2022}{IEA}
\defcitealias{eia2021}{EIA}
\defcitealias{oies2024}{OIES}
\defcitealias{IMF2025}{IMF}
\defcitealias{electricity2024analysis}{IEA}

\usepackage{xcolor,subfigure,epsfig,epstopdf,rotating,graphics,graphicx} 
\graphicspath{{./figures/}}
\usepackage[width=0.9\textwidth, font={footnotesize}]{caption}

\usepackage{rotating,lscape,pdflscape}
\usepackage{booktabs,longtable,float,array,multirow,colortbl,adjustbox}

\newcolumntype{C}[1]{>{\centering\arraybackslash}p{#1}}

\makeatother

\renewtagform{default}{\normalsize(}{)}     
\usetagform{default}                        

\makeatletter

\def \by{\mathbf{y}}

\def \bx{\mathbf{x}}

\def \bz{\mathbf{z}}
\def \be{\mathbf{e}}

\def \bq{\mathbf{q}}

\def \br{\mathbf{r}}

\def \biota{\boldsymbol{\iota}}
\def \btheta{\boldsymbol{\theta}}

\def \bbeta{\boldsymbol{\beta}}

\def \balpha{\boldsymbol{\alpha}}
\def \bgamma{\boldsymbol{\gamma}}

\def \bpsi{\boldsymbol{\psi}}

\def \bzeta{\boldsymbol{\zeta}}

\def \blambda{\boldsymbol{\lambda}}
\def \bchi{\boldsymbol{\chi}}

\DeclareMathOperator*{\diag}{diag}
\DeclareMathOperator*{\Ev}{E}
\DeclareMathOperator*{\Cov}{Cov}

\newcommand\comma{,\allowbreak}
\newcommand\equal{=\allowbreak}
\newcommand\leftbr{(\allowbreak}
\newcommand\rightbr{)\allowbreak}

\newcommand{\legdot}[1]{\tikz[baseline=-0.6ex]\draw[fill=#1,draw=#1] (0,0) circle (0.6ex);}

\definecolor{acquamarina}{RGB}{128, 255, 255}
\definecolor{limegreen}{RGB}{77, 255, 195}
\definecolor{blue1}{RGB}{8,4,252}
\definecolor{blue2}{RGB}{8,84,252}
\definecolor{blue3}{RGB}{8,172,252}

\makeatother

\title{\vspace{-60pt} \textbf{Modeling European Electricity Market Integration during turbulent times\thanks{\footnotesize
The authors gratefully acknowledge Andrea Bastianin, Hilde C. Bjørnland, Isabel Catalina Figuerola Ferretti, and the participants at the Virtual Workshop for Junior Researcher in Time Series (VTTS), Workshop Energy Transition and Climate Change, 4th Dolomiti Macro Meetings, University of Strathclyde, Junior MiTSS, 3rd UEA ECO Time Series Workshop, ICEEE 2025, and 15th RCEA Bayesian Econometrics Workshop for their useful feedback. Francesco Ravazzolo acknowledges financial support from the European Union - Next Generation EU – PNRR Missione 4 - Componente 2 - Investimento 1.1 and Italian Ministry MIUR – Project PRIN 20225J7H4K ``Econometric and Macro-Financial Models of Climate Change: Transition, Policies and Extreme Events'' (EMMCC).
Luca Rossini and Andrea Viselli acknowledge financial support from the Italian Ministry MIUR under the PRIN project 20223CEZSR ``Modelling Non-standard data and Extremes in Multivariate Environmental Time series'' (MNEMET). This research used the Computational resources provided by the Core Facility INDACO, which is a project of High-Performance Computing at the University of Milan.}
}}
\date{\today}
\author{Francesco Ravazzolo\thanks{BI Norwegian Business School, Norway and Free-University of Bozen-Bolzano, Italy and RCEA. \color{blue}\texttt{francesco.ravazzolo@bi.no}}
\and 
Luca Rossini\thanks{University of Milan, Italy and Fondazione Eni Enrico Mattei (FEEM). \color{blue}\texttt{luca.rossini@unimi.it}} \and Andrea Viselli\thanks{University of Milan, Italy.  \color{blue}\texttt{andrea.viselli@unimi.it}} }
\begin{document}
\maketitle


\begin{abstract}
This paper introduces a novel Bayesian reverse unrestricted mixed-frequency model applied to a panel of nine European electricity markets. Our model analyzes the impact of daily fossil fuel prices and hourly renewable energy generation on hourly electricity prices, employing a hierarchical structure to capture cross-country interdependencies and idiosyncratic factors. The inclusion of random effects demonstrates that electricity market integration both mitigates and amplifies shocks. Our results highlight that while renewable energy sources consistently reduce electricity prices across all countries, gas prices remain a dominant driver of cross-country electricity price disparities and instability. This finding underscores the critical importance of energy diversification, above all on renewable energy sources, and coordinated fossil fuel supply strategies for bolstering European energy security.
\end{abstract}

\clearpage
\onehalfspacing

\section{Introduction}
    \label{sect:intro}

On February 24, 2022, Russian troops launched a large-scale military operation against Ukraine, and energy quickly became a geopolitical weapon. 
The disruption of Russian gas supplies plunged Europe, and a bit less dramatically all the World, into an energy and political crisis, exposing deep vulnerabilities in its energy system. 
As electricity cannot be easily stored, the crisis underscored the urgency of market integration as a safeguard for energy security and a cornerstone of European competitiveness and economic growth (\citealp{draghi2024a}; \citealp{draghi2024b}).  

\cite{fabra2022electricity} describes how the market design based on merit order can result in a surge in electricity prices. The design sets a uniform price based on the most expensive plant required to
meet demand. As a result, all electricity – including that one generated from low marginal cost renewables– is paid at the higher marginal price of fossil-fired generation. The disruption of Russian
gas flows forced European countries to rely on neighboring countries to meet their energy demand, often importing even higher prices. Then, national policies aimed at securing domestic energy supply have
re-emerged, slowing progress toward deeper market integration and undermining recent advancements.
\citetalias{acer2023demand} (\citeyear{acer2023demand}) estimates significant welfare gains from cross-border electricity trade in 2021, with further benefits expected by 2030 under the REPowerEU strategy (\citealp{european2022repower}). 
This initiative aims to reduce short-term dependence on (Russian) fossil fuels while increasing the share of renewable energy sources (RES), enhancing resilience to external shocks, and improving market competition, innovation, and efficiency (\citealp{zachmann2024unity}). 
Meanwhile, global electricity demand continues to rise such as economies transition to lower-carbon energy sources, making these goals even more challenging.

Unlike fossil prices, electricity prices display high-frequency fluctuations with strong variation across hours. 
As noted by \cite{Durante2022}, electricity prices and renewable energy sources (RES), such as forecasted wind and solar generation, are strongly connected at the daily but also the hourly level. 
Notably, RES show strong seasonality and variations across periods of the years and hours of the day. 
Therefore, analyzing them at an hourly level is essential to model their patterns. Fossil prices do not show the same variability at hourly level and they are usually treated as constant during a given day, meaning the same value is used for the 24 hours, see e.g. \cite{gianfreda2020comparing}. Nevertheless, fossil-based producers leverage during the day their dispatchability, the ability to generate power when demand peaks, unlike intermittent renewable energy sources. Consequently, despite the relatively stable hourly pricing of fossil fuels within a day, the dynamic interaction among fossil fuel prices, RES output, and electricity prices is intricate and extends beyond a daily (fixed-value) regressor approach.

Then, to study the effects of daily surging fossil fuel prices and hourly RES generation on hourly electricity prices, we introduce a panel reverse unrestricted mixed data sampling (PRUMIDAS)  across nine European countries -- Denmark, Finland, France, Germany, Italy, Norway, Portugal, Spain, and Sweden -- covering the period from January 2019 to October 2023. 
The proposed dynamic panel model allows the inclusion of different time frequencies among covariates and prices while underlying the influence of a specific covariate on electricity prices.

Our approach builds on a mixed-frequency framework to allow for different sampling frequencies among the dependent (electricity prices) and independent variables -- forecasted renewable sources generation (wind and solar), and forecasted demand,  along with daily fossil fuel prices (Co$_2$, coal, and gas) -- and model their effect on hourly electricity prices.  
Since electricity prices have a higher sampling frequency than some covariates, we adopt the reverse unrestricted MIDAS approach of \cite{foroni2018using}, applied in \cite{foroni2023low} to analyze the impact of macroeconomic indicators on electricity prices in Italy and Germany. 
Extending it to a panel setting, we build on the hierarchical multi-country VAR framework of \cite{canova2009estimating} and \cite{casarin2018uncertainty}, replacing the unrestricted parametrization with a reverse mixed-frequency model to account for both hourly and daily covariates.  

There is a growing literature interested mainly in forecasting hourly electricity prices using both simple and complex methods. \cite{Weron2014} provides an comprehensive review of the different methods used for forecasting electricity prices, while \cite{Ziel2018} explore multivariate specification. 
The roles of RES and fossil fuel prices in these forecasts are well-documented as in \citealp{gianfreda2020comparing,gianfreda2023large}, and \citealp{lago2021forecasting}. 
However, previous approaches typically focus on single-country models, neglecting interdependencies among countries. 
We contribute to this literature in two key aspects. 
First, we integrate both hourly and daily explanatory variables as drivers of hourly electricity prices. 
Second, we propose a hierarchical panel structure that accounts for cross-country dependencies through both common and country-specific effects. 
This framework allows us to assess the electricity market integration by quantifying the impact of forecasted RES generation, demand, and fossil fuel prices across countries and hours. 
This aspect is crucial since as noted in different studies \citep[e.g.][]{Roth2023, zachmann2024unity}, greater integration among a subset of European countries could reduce the required dispatchable generation capacity by nearly 20\% and storage capacity by 30\% compared to a baseline scenario without increased integration.
Furthermore, the proposed methodology allows us to estimate the time-varying volatility of the electricity prices across countries without the need for explicit modeling it, thereby avoiding increased computational complexity.

Our findings confirm the central role of RES in reducing electricity prices. Wind and solar generation reduce costs in most countries, except Finland, where increased solar generation has been associated with almost no impact on electricity prices.
Germany, Denmark, and Italy benefit most from RES expansion, reinforcing their role in fostering market convergence by reducing price disparities and enhancing cross-border trade. 
These results corroborate the findings of the \citetalias{IMF2025} (\citeyear{IMF2025}), which highlights the importance of renewable-based power for deeper market integration. 
In contrast, gas price shocks have significantly increased electricity prices, especially in Italy, Germany, France, and Denmark.
Portugal and Spain, however, were largely shielded by the ``Iberian exception'', a 2022 policy that capped gas prices and reduced market uncertainty.

The crisis has demonstrated that energy integration can both mitigate and amplify shocks, depending on the level of coordination among European countries. 
France, for example, became reliant on electricity imports due to nuclear power outages, amplifying its exposure to external price fluctuations. 
Moreover, Denmark, France, Germany, and Italy exhibit substantial changes in the magnitude of the estimated coefficients during the current invasion in particular for renewable energy sources and gas prices.
Lastly, dependence on Russian natural gas emerges as a key driver of cross-country electricity price disparities, underscoring the importance of energy diversification and well-coordinated fossil fuel supply strategies at the European level.  

These findings open a debate on the importance of the integration of the European electricity markets and the persistent national-level constraints that hinder it.
Despite energy's central role in European integration, policy in this domain remains largely under national control, with the European Commission holding limited formal authority. The European Union should emphasize the need for a more diversified energy infrastructure to reduce reliance on single sources and mitigate price volatility, including implementing strategies that coordinate fossil fuel supply across countries to manage price disparities and instability given the dominant role of the gas prices as a driver of cross-country electricity price variations. At the same time, policies should continue to support the deployment of renewable energy sources to further stabilize markets and reduce dependence on fossil fuels.

The paper is structured as follows. 
Section~\ref{sec:data} introduces the dataset and discusses key aspects of the European energy market. 
Section~\ref{sec:method} presents the panel reverse unrestricted mixed data sampling (PRUMIDAS) model, including its specification and estimation. 
Section \ref{sec:analysis-application} applies the PRUMIDAS model to analyse cross-country differences in the drivers of electricity prices and discusses the results. 
Section~\ref{sec:postwarresults} focuses on the post-invasion period. 
Section~\ref{sec:concl} concludes with future policy implications.

\section{Electricity, renewable, and fossil data}
\label{sec:data}

We concentrate our analysis on nine European countries -- Denmark, Finland, France, Germany, Italy, Norway, Portugal, Spain, and Sweden -- which are characterized by significant heterogeneity in power market structures, regulatory frameworks, and energy mixes.
These countries differ in their reliance on fossil fuels versus renewables, their interconnections with neighboring countries, the degree of market liberalization, and the extent of government intervention in electricity pricing (see e.g. \citetalias{electricity2024analysis}, \citeyear{electricity2024analysis}). 

Moreover, their exposure to geopolitical risks varies, particularly regarding their dependence on Russian gas. Some countries have historically relied on long-term pipeline contracts, while others have diversified their energy imports through liquefied natural gas (LNG) terminals and cross-border interconnections (\citetalias{oies2024}, \citeyear{oies2024}; \citealp{sziklai2020impact}).

Unlike other commodities (such as oil), electricity is a non- or partially storable commodity; as a result, its price exhibits substantial variability across the hours of the day and throughout the year as shown in the Supplement. It is therefore crucial to work with prices at an hourly level to better understand the behaviour of the European electricity markets.

We use hourly data for electricity prices, forecasted demand, wind and solar generation, alongside daily fossil fuel prices (Co$_2$, coal, gas), from 1 January 2019 to 31 October 2023 -- covering more than $40$ thousand hours per country. 
Hourly day-ahead auction prices are collected from power exchanges: the European Energy Exchange (EEX) for Germany and France, the Gestore dei Mercati Energetici (GME) for Italy (using the single national price, PUN), NordPool for Denmark, Norway, Sweden, and Finland; and the Operador do Mercado Iberico (OMI) for Spain and Portugal. 
Prices are quoted in \euro/MWh and pre-processed to account for time-clock changes, excluding the 25th hour in October and interpolating the missing 24th hour in March, thus ensuring no missing observations (\citealp{gianfreda2020comparing}).

Timely forecasts for demand, wind, and solar generation are essential for market operators: on day $t$, they run forecasting models to generate $24$ prices and quantities for delivery on day $t+1$, submitting them before market closure (around noon on day $t$). Forecasted quantities reflect conditions on day $t$, while fossil fuel settlement prices and plant capacities are determined on day $t-1$.

Hourly forecasted quantities are sourced from London Stock Exchange Group (LSEG), using the operational weather model from the European Centre for Medium-Range Weather Forecasts (EC). This model updates between 05:40 and 06:55 a.m., providing the latest data for market operators. Missing forecasts are replaced hierarchically: first, using the ECop00 model at midnight; if unavailable, the Global Forecast System (GFS) ensemble model at midnight (GFSen00); then the GFS operational model (GFSop00); and, if necessary, GFSen18, GFSop18, or ECen12, depending on publication time. 
Residual gaps are interpolated. The EC operational model is deterministic with high resolution, while the GFS ensemble model is probabilistic, simulating weather instability.

As noted in \cite{gianfreda2020comparing}, solar power forecasts exhibit null values in the early morning and late evening, therefore, we pre-process them using a linear transformation. 
In particular, as shown in the Supplement, electricity prices and forecasted renewable energy sources exhibit distinct intraday path, with notable differences between peak and off-peak hours. Thus, it is preferable to employ hourly data instead of merging them by using averages and working with daily data.

Fossil fuel prices include coal (Intercontinental Exchange API2, ticker LMCYSPT), natural gas (one-month forward ICE UK, ticker NATBGAS), and Co$_2$ (EEX-EU Co$_2$ Emissions, ticker EEXEUAS), all converted to \euro/MWh using WMR\&DS exchange rates (USEURSP for USD-EUR, UKEURSP for GBP-EUR). 
These daily prices are constant over 24 hours, with interpolated values for weekends and holidays.
Following \cite{paraschiv2014impact} and \cite{gianfreda2023large}, we use price levels -- not logarithmic prices -- to preserve statistical properties and volatility dynamics, while exogenous variables are standardized to preserve the scale of electricity prices.


\begin{figure}[h!]
    \begin{tabular}{ccc}
    Denmark & Finland & France \\
    \includegraphics[width = 5cm]{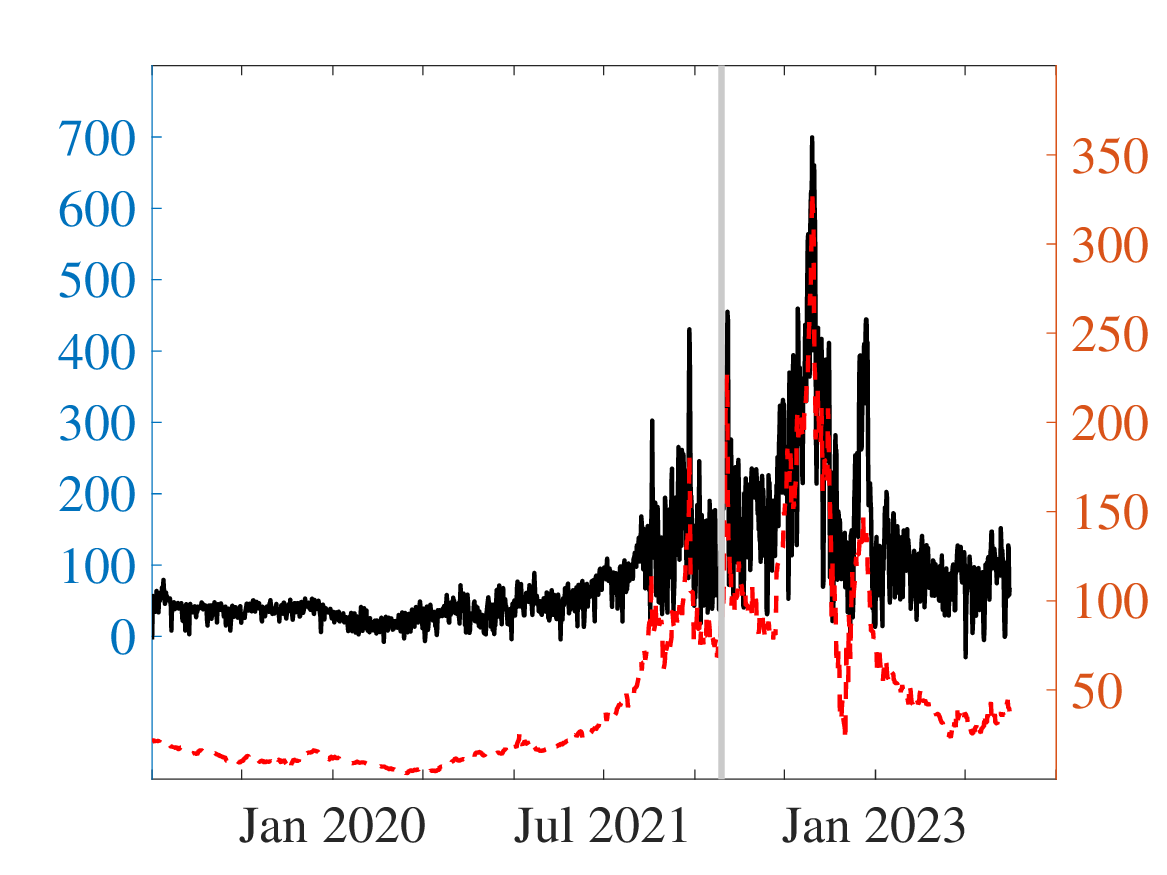} &
    \includegraphics[width = 5cm]{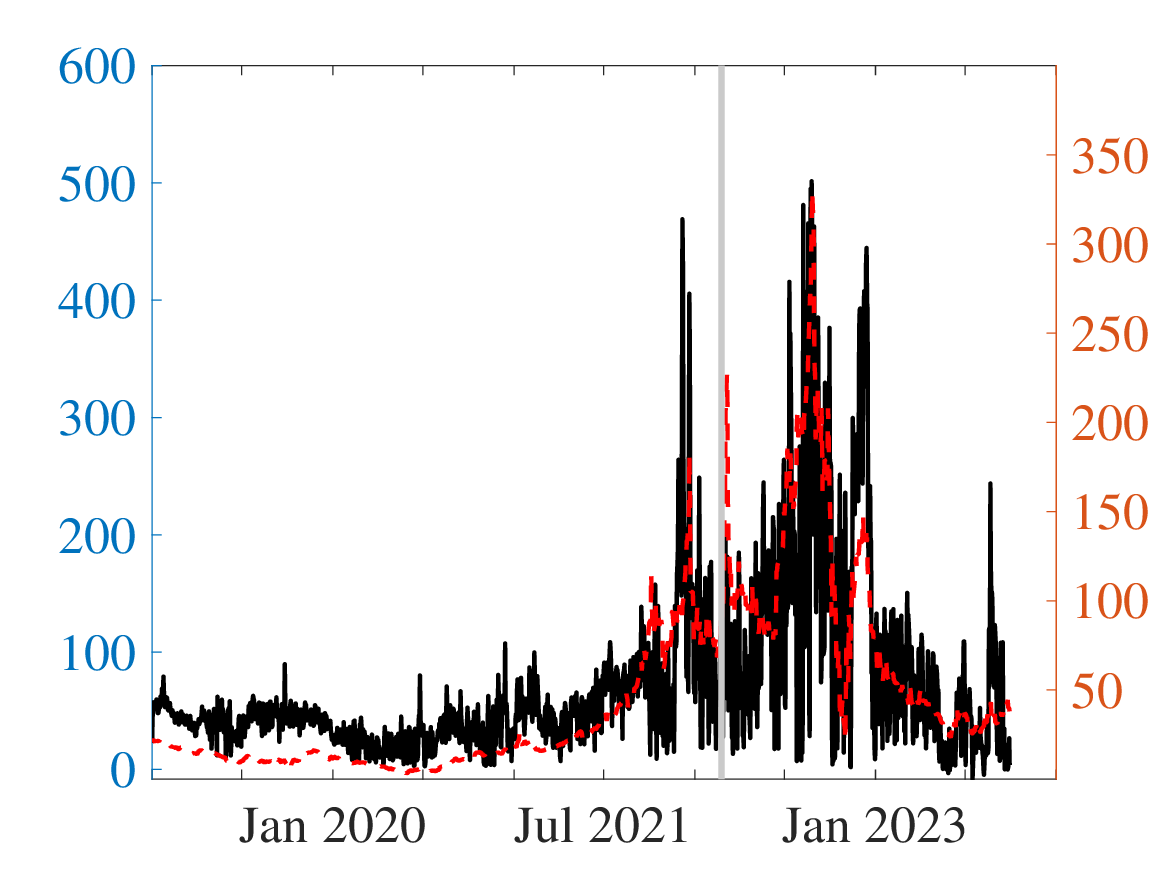} &
    \includegraphics[width = 5cm]{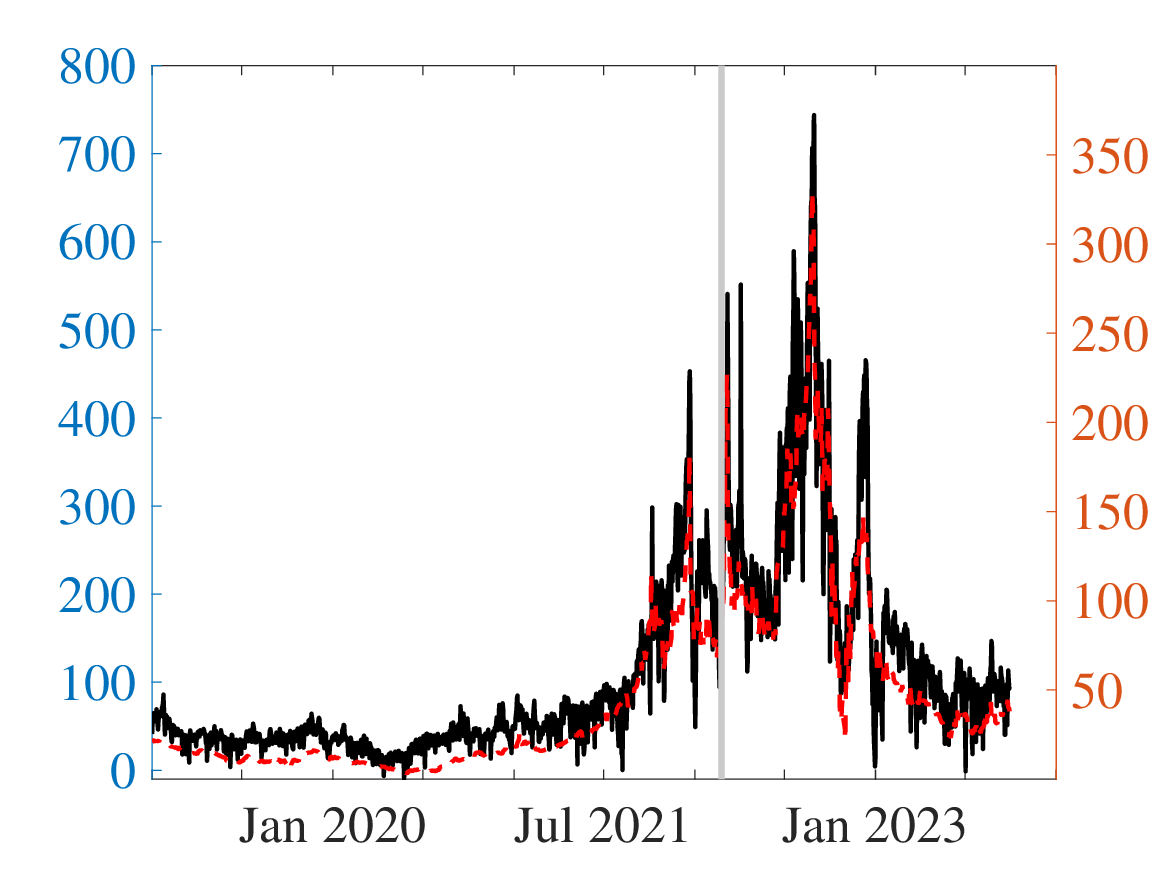} \\
    Germany & Italy & Norway \\
    \includegraphics[width = 5cm]{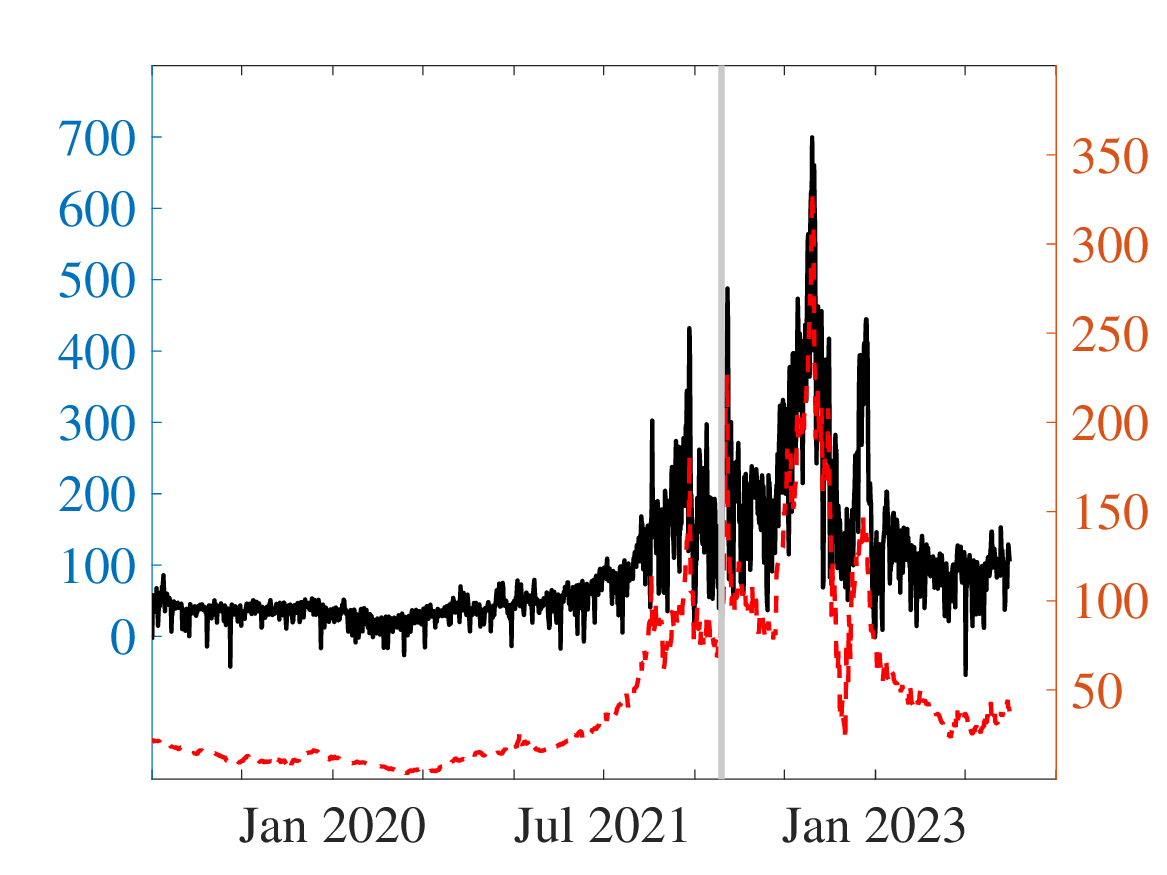}  &
    \includegraphics[width = 5cm]{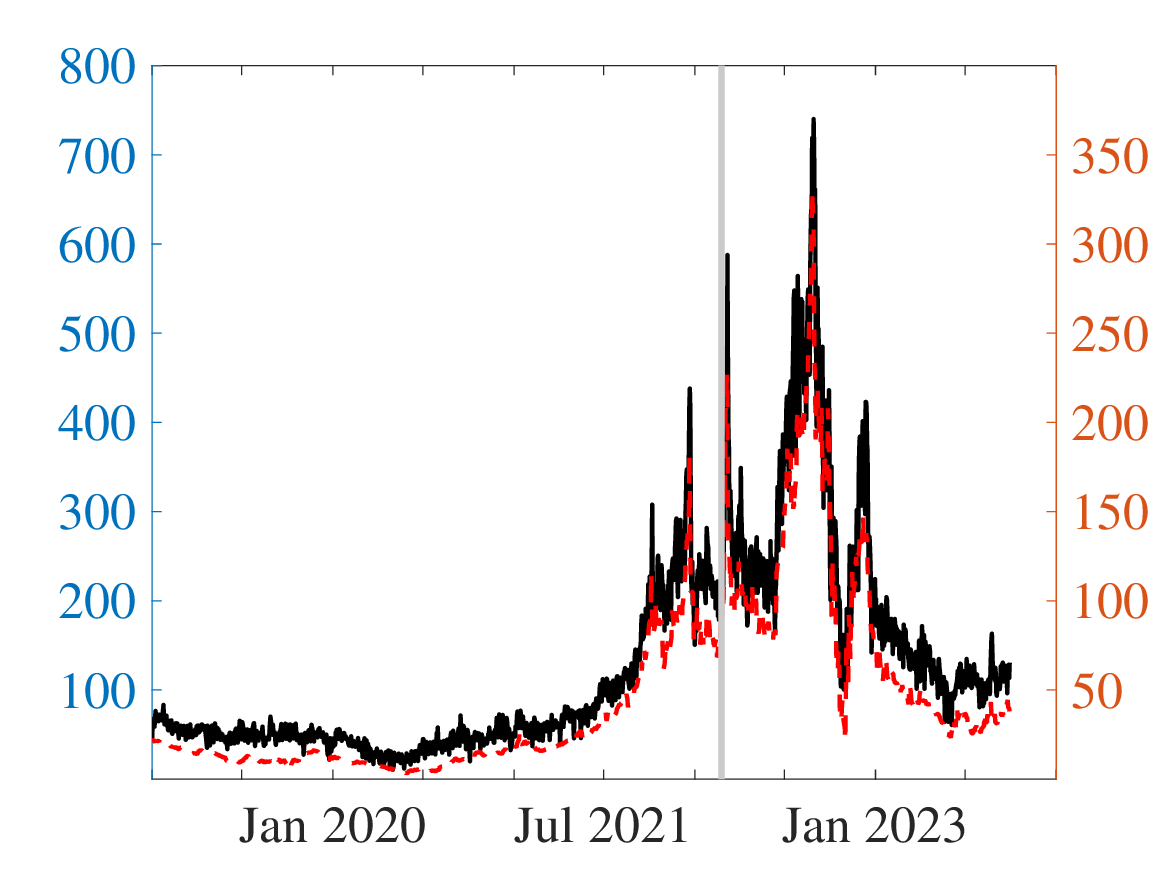} &
    \includegraphics[width = 5cm]{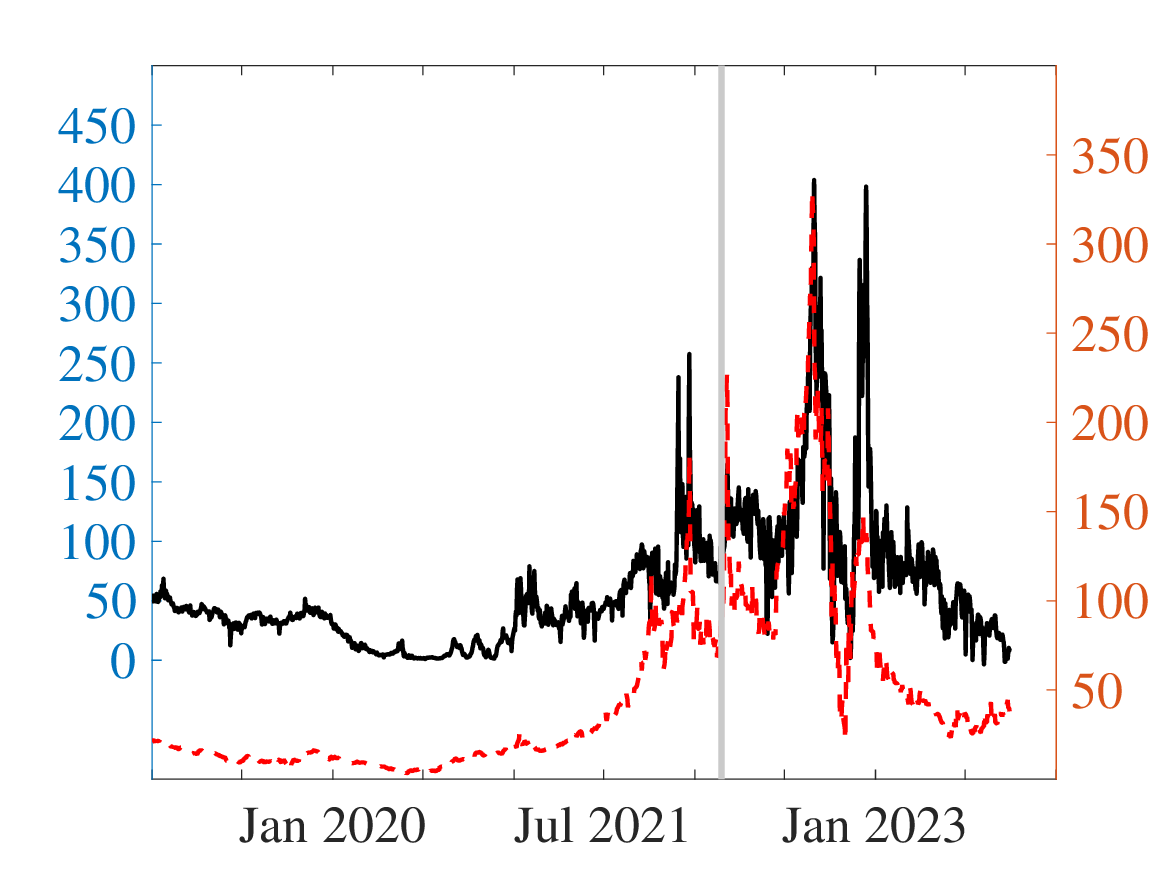} \\
    Portugal & Spain & Sweden \\
    \includegraphics[width = 5cm]{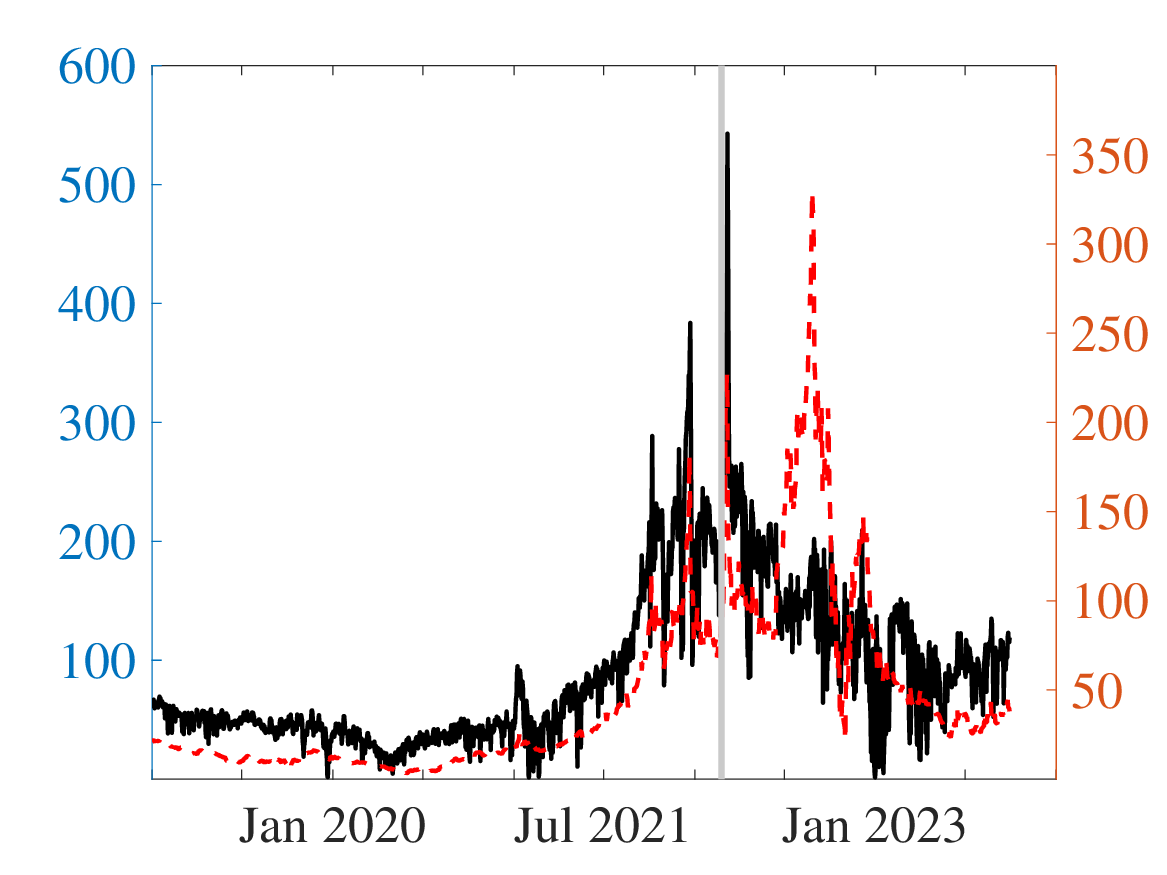} &
    \includegraphics[width = 5cm]{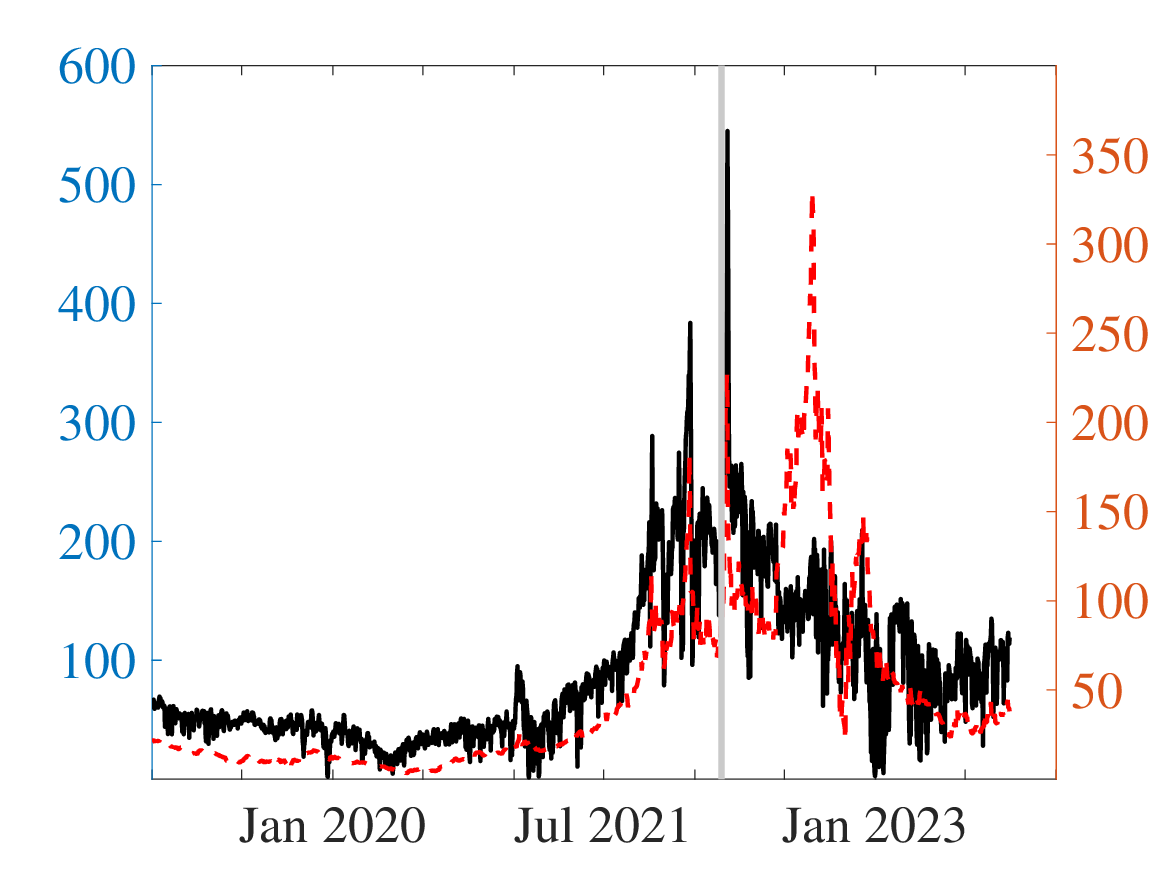} &
    \includegraphics[width = 5cm]{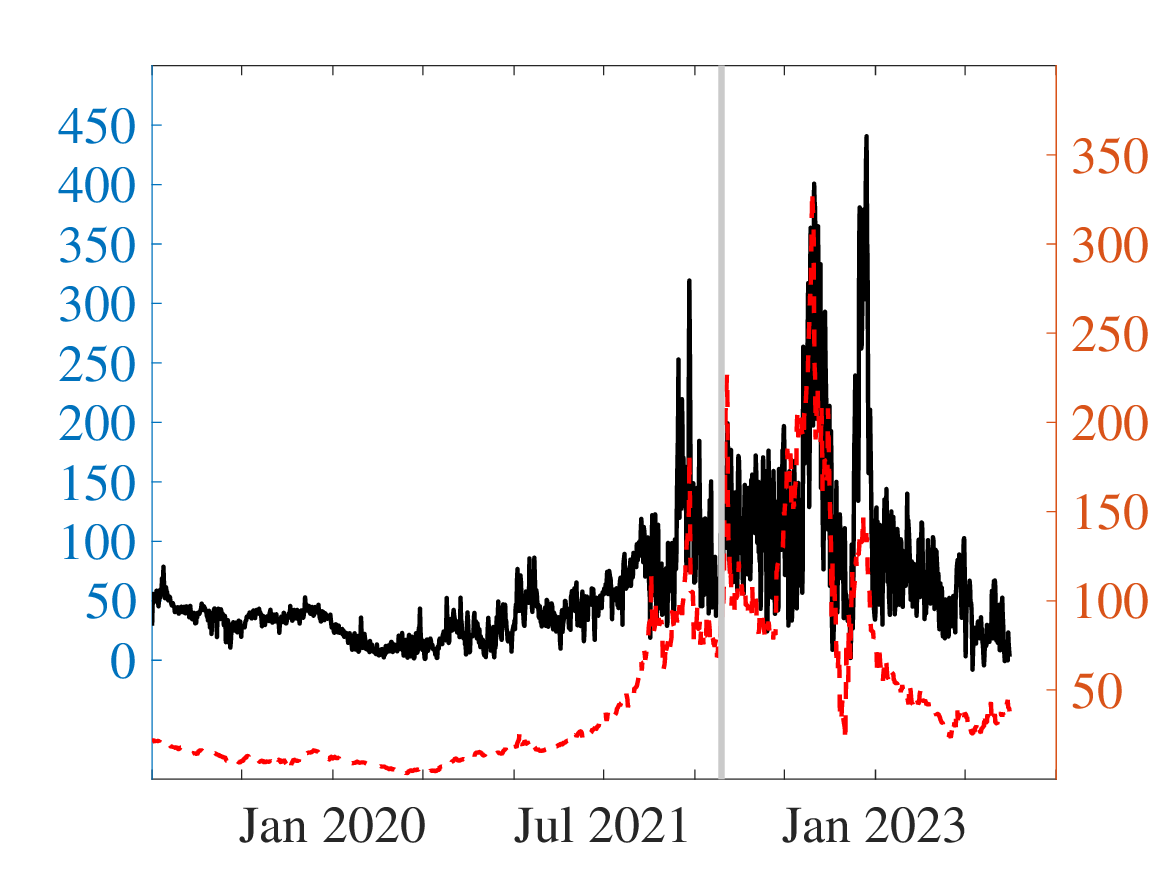}
    \end{tabular}
\caption{Daily electricity prices across countries -- Denmark, Finland, France, Germany, Italy, Norway, Portugal, Spain, and Sweden -- against the natural gas price, from January 2019 to October 2023. The gray vertical line corresponds to the Russia's invasion of Ukraine.} 
\label{fig:country-electricity-gas} 
\end{figure}

As discussed in \cite{ravazzolo2023price}, electricity prices in some countries are closely linked to the gas prices movements. 
Figure~\ref{fig:country-electricity-gas} plots these relations with the left axis representing electricity prices and the right axis gas prices. 
Both prices exhibit co-movement, with electricity prices rising sharply from mid-2021 (\citetalias{eia2021}, \citeyear{eia2021}; \citetalias{iea2022}, \citeyear{iea2022}; \citealp{enerdata2023}). 
As economies rebounded from the COVID-19 pandemic, increased industrial activity and energy demand tightened supply. 
Europe entered 2021 with historically low gas storage levels due to the cold 2020–2021 winter and storage refilling was slower than usual, partly due to high demand and limited supply. 
Russian pipeline deliveries also declined during 2021, owing to Gazprom's domestic storage needs, geopolitical tensions surrounding Nord Stream 2, and infrastructure maintenance.

In the aftermath of Ukraine's invasion, European gas prices reached record peaks, which consequently affected electricity prices. 
Russian pipeline supply, including Nord Stream 1, was gradually reduced, forcing Europe to seek alternative, often more expensive, sources. 
Increased global competition, especially from Asia for liquefied natural gas (LNG) further constrained European supply. 
In response, governments rushed to fill gas storage ahead of winter, driving prices high. 
During this period, electricity prices exceeded 600 \euro/MWh in Denmark, France, Germany, and Italy. 
Portugal and Spain were less affected due to the ``Iberian exception'', a measure approved by the European Commission allowing these countries to decouple gas and electricity prices for 12 months (\citealp{hidalgo2024iberian}).

Starting in 2023, electricity prices declined due to several factors (\citetalias{oies2024}, \citeyear{oies2024}; \citealp{mckinsey2022}; \citealp{columbia2021}). 
After the launch of the REPowerEU strategy by the European Commission, Europe increased LNG imports, particularly from the United States, Qatar, and other suppliers, reducing its reliance on Russian pipeline gas by diversifying sources. 
While imports from Russia fell sharply, they were largely offset by increased supply from Norway, Algeria, and Azerbaijan. 
EU governments also implemented market-stabilizing measures such as price caps, subsidies, and windfall taxes. Notably, the EU introduced a gas price cap of \euro180/MWh in 2023, helping to stabilize the  market (\citealp{ravazzolo2023price}; \citetalias{acer2023demand}, \citeyear{acer2023demand}).

In the Supplement, we plot the time series of electricity, Co$_2$, coal, and gas prices, as well as forecasts of demand, solar and wind generation, highlighting the role of fat-tailed and asymmetric distributions.

\section{A panel reverse unrestricted mixed data sampling model}
\label{sec:method}
    
We propose a new Bayesian mixed-frequency multi-country dynamic panel model, where interdependence among countries and across high-frequency time periods is introduced through a parsimonious hierarchical structure on the coefficients. 
In this framework, the dependent ``target'' variable is sampled at a lower frequency than some of the covariates.

This class of models offers flexibility in reducing the number of parameters when this grows faster than the sample size, thereby mitigating potential overfitting issues.
\cite{chib1995hierarchical} originally considered a Bayesian seemingly unrelated regression (SUR) model with time-varying parameters, which can be easily cast within a hierarchical structure. 
Subsequently, \cite{canova2009estimating} proposed a multi-country panel vector autoregressive (VAR) model, where interdependence among countries is introduced via common factors that substantially reduce the number of parameters. 
More recently, \cite{koop2019forecasting} extended the multi-country approach of \cite{canova2009estimating} by allowing for structural breaks and time-varying volatility.
Similarly, \cite{casarin2018uncertainty} introduced a dynamic panel with Markov-switching coefficients, where mixed frequency data are handled through an unrestricted mixed-data sampling (MIDAS) specification (\citealp{foroni2015unrestricted}).

Building on \cite{canova2009estimating} and \cite{casarin2018uncertainty}, we propose a new panel reverse unrestricted mixed data sampling (PRUMIDAS) model. 
Then, we consider a convenient specification to analyze the impact of forecasted renewable energy sources (RES) generation, forecasted demand, and fossil fuels prices on electricity prices.  
Our analysis focuses on disentangling country-specific deviations (i.e. country-specific effects) from the ``common'' European effect.

\subsection{Notation} \label{sec:Notation}

Before presenting our model, we briefly introduce the basic notation used throughout the paper. 
We denote the dependent variable for country $g$ as $y_{g,t+h}$, the set of high-frequency covariates as $x^{(1)}_{gj,t+h}$, for $j=1,\ldots,\Tilde{N}$ and the set of low-frequency covariates as $x^{(2)}_{gj,t}$, for $j=\Tilde{N}+1,\ldots,N$, where $t=1,\ldots,T$, $g=1,\ldots, G$, and $h=0,\ldots, H-1$. 
In detail, $G$ denotes the number of countries, $\tilde{N}$ the number of high-frequency covariates, and $N$ the total number of covariates.

Specifically, the target variable $y_{g,t+h}$ and high-frequency covariates $x^{(1)}_{gj,t+h}$ are sampled over the high-frequency time periods $t=1,\ldots, T$, whereas the low-frequency covariates $x^{(2)}_{gj,t}$ are observed over $t=H,2H,\ldots,\Tilde{T}H$, where $H$ denotes the frequency mismatch between the high- and low-frequency variables. 
Consequently, $\Tilde{T} = \lfloor T/H \rfloor$ represents the number of low-frequency observations such that $\Tilde{T}H \leq T$, where $\lfloor\cdot\rfloor$ indicates the smallest integer part.

Note that we may re-write $x^{(2)}_{gj,t}$ as $x^{(2)}_{gj,t+0}$, allowing us to use the same subscript $j$ for both the low- and high-frequency types of variables. We denote $N$ as the number of covariates.
Then, we define
\begin{equation*}
\label{sec:method:Hj}
H_j = \begin{cases}
    1, &\text{ if } j\in[1,\tilde{N}], \\
    H, &\text{ if } j\in[\tilde{N}+1,N],
    \end{cases}        
\end{equation*}
as the ``lag multiplier'', and
\begin{equation*}
\label{sec:method:hj}
h_j = \begin{cases}
    h, &\text{ if } j\in[1,\tilde{N}], \\
    0, &\text{ if } j\in[\tilde{N}+1,N],
    \end{cases}
\end{equation*}
as the ``frequency mismatch''. Both quantities vary depending on the sampling frequency.

\subsection{The mixed-frequency model}

As a starting point, we consider the mixed-frequency dynamic panel model given by
\begin{equation}
\label{sec:method:first-eq}
    y_{g,t+h} = \mu_{gh} + \sum_{a=1}^{A} \alpha_{gha}y_{g,t+h-a} + \sum_{j=1}^{N} \sum_{b=0}^{B_j} \beta_{gjhb} x_{gj,t+h_j-bH_j} + e_{g,t+h},
\end{equation}
where $\mu_{gh}$ denotes the intercept, $\alpha_{gha}$ are the coefficients of the lagged high-frequency dependent variable, and $\beta_{gjhb}$ the coefficients associated with the $j$-th covariate $x_{gj,t+h_j-bH_j}$.
$A$ and $B_j$ are the maximum lag orders for the dependent variable and the $j$-th covariate, respectively. 
The error term $e_{g,t+h}$ is assumed to be independently and identically normally distributed with zero mean and variance $\sigma^2_{gh}$, varying by country and high-frequency time periods.

To simplify the notation, we rewrite the model in Eq.~\eqref{sec:method:first-eq} using lag polynomials.
Specifically, let $\alpha_{gh}(L)$ and $\beta_{gjh}(L^{H_j})$ denote the lag polynomials for the autoregressive and the $j$-th explanatory term, respectively:
\begin{equation*}
    \alpha_{gh}(L) = 1-\sum_{a=1}^{A}\alpha_{gha}L^{a}, 
    \qquad 
    \beta_{gjh}\left(L^{H_j}\right) = \sum_{b=0}^{B_j}\beta_{gjhb}L^{bH_j},
\end{equation*}
where $L$ is the lag operator and $H_j$ is the ``lag multiplier'' associated with the sampling frequency.
In detail, the lag polynomial for low-frequency covariates only includes powers of $L^{H}$. 
Moreover, we allow for contemporaneous effects of the covariates by initiating the summation at $b=0$, although this might not be applicable if the covariate is observed at a lower frequency.

Using this notation, the model in Eq.~\eqref{sec:method:first-eq} can be rewritten as
\begin{align}
\label{sec:method:second-eq}
\small
 \alpha_{gh}(L)y_{g,t+h} &= \mu_{gh}+\sum_{j=1}^{N} \beta_{gjh}(L^{H_j})x_{gj,t+h_j}+e_{g,t+h}, 
\end{align}
where $y_{g,t+h}$ is a linear function of its own distributed lags and those of the covariates.
Thus, our model resembles an autoregressive distributed lag model, with the key difference that we consider covariates over two different sampling frequencies.  
Differently from our approach, \cite{casarin2018uncertainty} consider a mixed-frequency panel in which the target variable is sampled at a lower frequency than the covariates. 
Accordingly, our model is related to the reverse mixed data sampling (RMIDAS) model introduced by \cite{foroni2018using}, and employs the ``unrestricted'' linear parametrization proposed by \cite{foroni2015unrestricted}. 

However, while in the RMIDAS model the coefficients depend on the horizon $h$ for which a forecast is desired, in our model the dependence of the coefficients on the horizon $h$ is introduced through a hierarchical structure, allowing for an interpretation as ``high-frequency'' time effects. 
Similarly, the coefficients depend on $g$ to introduce ``country'' effects. 
Importantly, we assume that, for a given horizon $h$, the high-frequency effect is the same across all countries and, for a given country $g$, the country effect is the same across all high-frequency periods.

Our framework can also be compared to the multi-country vector autoregressive (VAR) specification of \cite{canova2009estimating} (see e.g. Example 2.1.1, p. 933). 
While our model can be interpreted as a panel VAR model with $H$ series, it does not capture cross-country or cross-hourly feedback effects from the past to the present. This simplification is justified by the empirical context of electricity pricing, where such an effect is not relevant.
In contrast, our model enables the incorporation of mixed-frequency data, extending the analysis to jointly include high-frequency electricity prices and RES forecasted generation alongside low frequency fossil prices.

\subsection{Model specification}

In our analysis, we consider $\tilde{N} = 3$ high-frequency (hourly forecasted demand, forecasted wind power generation, and forecasted solar power generation) and three low-frequency covariates (daily coal, gas, and Co$_2$ prices) resulting in $N=6$ variables in total. 
The frequency mismatch parameter $H$ is equal to $24$, corresponding to the number of hours in a day.
The panel comprises $G=9$ countries, representing the major European electricity markets.

Let $y_{g,t+h}$ denote the electricity prices in the day-ahead (spot) market for country $g$, for delivering on day $t$ at hour $h$. Following \cite{Ravivetal2015} and \cite{gianfreda2020comparing}, the model in Eq.~\eqref{sec:method:second-eq} becomes
\begin{equation}
\label{sec:method:model-specification-eq}
    y_{g,t+h} = \mu_{gh}+\sum_{a\in\Tilde{A}}\alpha_{gha}y_{g,t+h-a}+\sum_{j=1}^{6}\beta_{gjh}x_{gj,t+h_j}+e_{g,t+h},
\end{equation}
where $\Tilde{A} = \{H,2H, 7H\}$ includes yesterday, two-days ago and 1-week ago hourly electricity prices on hour $h$ of the day. Based on the strong co-movement between fossil fuel and electricity prices, we only consider the contemporaneous effects of the covariates ($B_j = 0$ for any $j$). 

It is worth noting that while MIDAS approaches are primarily used for forecasting, their application in econometric analysis remains limited.
Nonetheless, \cite{ghysels2004midas} argue that MIDAS models, in general, can improve the precision of causal effect estimation by mitigating aggregation bias inherent in distributed lag models using same-frequency variables.
In this context, the high granularity of our hourly dataset enables us to reduce parameter estimation uncertainty.

The temporal alignment of the data must be carefully considered, depending on the type of variable and its release. Since day-ahead electricity prices are determined the day before physical delivery, we use the first lag (i.e. the previous day's value) for fossil fuel prices to reflect the information available at the time of the market closure.
Similarly, forecasted RES generation variables reflect the information available when the market is set. We do not include additional past lags of the covariates as, economically, electricity prices are primarily determined by contemporaneous RES generation and fossil fuel prices. In addition, including past lags would substantially increase computational complexity.


To keep the discussion as general as possible, the subsequent sections will be based on the general formulation of the model presented in Eqs.~\eqref{sec:method:first-eq} and \eqref{sec:method:second-eq}.

\subsection{Common and random effects}
\label{sec:model:common-random-effects}


Following the approach of \cite{canova2009estimating}, we consider a hierarchical structure that shrinks the coefficients in Eq.~\eqref{sec:method:first-eq} towards group-specific common effects. 
This hierarchical shrinkage structure reduces model complexity and mitigates the risk of overfitting that could arise in alternative specifications without cross-equation restrictions.

In particular, we introduce country- and hourly-specific common and random effects in both the drift and slope coefficients to capture interaction effects across different countries and hours, reflecting the structure of the European electricity markets. 
These coefficients are decomposed as follows:
\begin{align}
    \label{sec:method:cere}
    \mu_{gh} &= \mu + \psi_{\mu,h} + \zeta_{\mu,g},  \nonumber\\
    \alpha_{gha} &= \alpha_{a} + \psi_{\alpha,ha} + \zeta_{\alpha,ga}, \\
    \beta_{gjhb} &= \beta_{jb} + \psi_{\beta,jhb} + \zeta_{\beta,gjb},  \nonumber
\end{align}
where the coefficients vary by country $g$, hour $h$ of the day, and lags $a$ and $b$.
Moreover, we assume that $\Cov(\psi_{\mu,h}\comma\psi_{\alpha,h^{\prime}a^{\prime}})\equal0$, $\Cov(\psi_{\mu,h}\comma\psi_{\beta,jh^{\prime}b^{\prime}})\equal0$, $\Cov(\psi_{\alpha,ha}\comma\psi_{\beta,jh^{\prime}b^{\prime}})\equal0$ and $\Cov(\zeta_{\mu,g}\comma\zeta_{\alpha,g^{\prime}a^{\prime}})\equal0$, $\Cov(\zeta_{\mu,g}\comma\zeta_{\beta,g^{\prime}jb^{\prime}})\equal0$, $\Cov(\zeta_{\alpha,ga}\comma\zeta_{\beta,g^{\prime}jb^{\prime}})\equal0$ for all $g,g^{\prime},h,h^{\prime}\comma a \comma a^{\prime}$ and $b,b^{\prime}$. 

In particular, $\psi_{\mu,h}$, $\psi_{\alpha, ha}$, and $\psi_{\beta, jhb}$ denote the hourly random effects, which are homogeneous across countries but vary across hours.
They can be interpreted as ``time dummies'', that are, as deviations from the common effect $\beta_j$, capturing systematic patterns in electricity consumption that vary across hours of the day but are assumed to be homogeneous across countries. 
Importantly, all hourly time series in the dataset are referenced to the local time zone of each respective country. As a result, issues of temporal misalignment across countries do not arise. Since these effects do not change substantially across countries, we do not report their results.
Conversely, $\zeta_{\mu,g}$, $\zeta_{\alpha,ga}$, and $\zeta_{\beta,gjb}$ represent the country random effects, which are homogeneous across hours but vary across countries. 

The hierarchical structure enables us to model hour-specific dependence while maintaining a univariate specification for each hour, capturing intraday dynamics while avoiding the complexity of a full multivariate model.
Alternatively, one could allow the hourly effect to vary across countries -- i.e., by setting $\psi_{\mu,gh}$ as country-specific -- such an extension would account for differences in daily consumption patterns and peak load times across countries. 
However, to reduce model complexity, we opt for a parsimonious specification in which the hourly random effects are assumed to be homogeneous across countries.

\subsection{Prior specifications}
\label{sec:model:prior-spec}

Having introduced the model and its hierarchical structure, we now specify the prior representation for the parameters of the PRUMIDAS model.

We assume independent normal priors for the common intercept $\mu$, the vector of common autoregressive coefficients $\balpha = \leftbr\alpha_{1}\comma\ldots\comma\alpha_{A}\rightbr^{\prime}$, and, for each covariate $j$, the vector of slope coefficients $\bbeta_{j} = \leftbr\beta_{j0},\beta_{j1}\comma\ldots\comma\beta_{jB_j}\rightbr^{\prime}$. Formally,
\begin{equation*}
    \mu\sim\mathcal{N}(0,s_{0}^2), \qquad
    \balpha\sim\mathcal{N}_A(\bm{0}_{A},r_{0}^2 I_{A}), \qquad
    \bbeta_j\sim\mathcal{N}_{(1+B_j)}(\bm{0}_{(1+B_j)},r_{0}^2 I_{(1+B_j)}),
\end{equation*}
where $s_0$ and $r_0$ are hyperparameters set equal to 10, and $(1+B_j)$ denotes the total number of lags for covariate $j$, including its contemporaneous effect.

For the random effects, we assume hierarchical priors by assuming that both the hourly and country random effects follow independent normal distributions \begin{equation}
\label{sec:model:prior-spec:random-effects}
    \begin{aligned}
        \psi_{\mu,h}&\sim\mathcal{N}(0,q_{\mu}),\qquad&\zeta_{\mu,g}\sim\mathcal{N}(0,r_{\mu}), \\
        \psi_{\alpha,ha}&\sim\mathcal{N}(0,q_{\alpha}),\qquad&\zeta_{\alpha,ga}\sim\mathcal{N}(0,r_{\alpha}), \\
        \psi_{\beta,jhb}&\sim\mathcal{N}(0,q_{\beta}),\qquad&\zeta_{\beta,gjb}\sim\mathcal{N}(0,r_{\beta}),
    \end{aligned}
\end{equation}
where the scale parameters are assumed to follow independent inverse gamma priors
\begin{equation*}
    \begin{aligned}
        &q_{\mu},q_{\alpha},q_{\beta}\sim\mathcal{IG}(n_0,m_0), \\
        &r_{\mu},r_{\alpha},r_{\beta}\sim\mathcal{IG}(n_0,m_0),
    \end{aligned}
\end{equation*}
with the hyperparameters $n_0$ and $m_0$ set equal to 0.1. Here, $\mathcal{IG}(n_0,m_0)$ denotes the inverse gamma distribution with shape and rate parameters denoted as $n_0$ and $m_0$, respectively.

We further specify a flexible structure for the variance of the error term, $e_{g,t+h}$, by allowing for country- and hourly-specific heteroskedasticity, that is
\begin{equation}
    \label{variance}
    \sigma_{gh}^2 = \sigma^2\lambda_{h}^{-1}\chi_{g}^{-1},
\end{equation}
where $\lambda_{h}$ and $\chi_{g}$ capture hourly- and country-heterogeneity respectively, while $\sigma^2$ represents a common variance component.
We further assume hierarchical prior specifications for the scale parameters as follows
\begin{equation*}
    \sigma^2\sim\mathcal{IG}(v_1,w_1),
    \qquad
    \lambda_{h}\sim\mathcal{IG}(v_2,w_2),
    \qquad
    \chi_{g}\sim\mathcal{IG}(v_3,w_3),
\end{equation*}
where, for the hyperparameters, we set $v_i=w_i=0.1$ for $i=1,2,3$.


\subsection{Posterior approximation}
\label{sec:model:posterior-approx}

We consider additional representations of our model that are derived under a hierarchical structure incorporating both common and random effects. 
In particular, we can rewrite the model in a stacked form, linking it to hierarchical panel vector autoregression models which have become extremely popular in the empirical macroeconomic literature (\citealp{canova2009estimating}; \citealp{koop2019forecasting}). 
This formulation facilitates the approximation of the joint posterior and increases the computational efficiency of the algorithm. We consider an efficient MCMC algorithm where the random effects are integrated out of the likelihood (Rao-Blackwellization) and obtain an analytical expression of the model with a time-varying variance, avoiding the need of explicitly adopting a time-varying variance specification.

To present the model formally, we introduce stacked and vectorized notation for the variables used in the PRUMIDAS model. 
We denote $\by_{g,t+h-1}\equal(y_{g,t+h-1}\comma\ldots\comma y_{g,t+h-A})^{\prime}$ as the vector of lagged values of the high-frequency dependent variable, $\bx_{gj,t+h_j} = (x_{gj,t+h_j},\ldots,x_{gj,t+h_j-B_jH_j}\rightbr^{\prime}$ as the vector of lagged values of either high- or low-frequency covariates, depending on the value of $H_j$ and $h_j$, as defined in Section~\ref{sec:Notation}.
We define $\bz_{g,t+h} = \leftbr1,\by_{g,t+h-1}^{\prime}\comma\bx_{g1,t+h_1}^{\prime}\comma\ldots\comma\bx_{gN,t+h_N}^{\prime}\rightbr^{\prime}$ as the vector including both the autoregressive high-frequency and lagged covariates at both frequency. 

Moreover, denote $L = \bigl(1+A+\sum_{j=1}^{N}(1+B_j)\bigr)$ as the dimension of $\bz_{g,t+h}$. Denote as $\bgamma = \leftbr \mu,\balpha^{\prime},\bbeta_{1}^{\prime}\comma\ldots\comma\bbeta_{N}^{\prime})^{\prime}$ the $L-$vector of common effects, including the common intercept and common coefficients, $\balpha = (\alpha_{1}\comma\ldots\comma\alpha_{A})^{\prime}$ and $\bbeta_{j} = \leftbr \beta_{j0}, \beta_{j1}\comma\ldots\comma\beta_{jB_j}\rightbr^{\prime}$. 
Let $\bpsi_{h} = (\psi_{\mu,h}\comma\bpsi_{\alpha,h}^{\prime}\comma\bpsi_{\beta,1h}^{\prime}\comma\ldots\comma\bpsi_{\beta,Nh}^{\prime}\rightbr^{\prime}$ be the $L-$vector of hourly random effects, with $\bpsi_{\alpha,h} = \leftbr\psi_{\alpha,h1}\comma\ldots\comma\psi_{\alpha,hA}\rightbr^{\prime}$ be a $A$-vector and $\bpsi_{\beta,jh} = \leftbr\psi_{\beta,jh0}\comma\ldots\comma\psi_{\beta,jhB_j}\rightbr^{\prime}$ a $(1+B_j)$-vector for $j=1,\ldots,N$. 
Then, let $\bzeta_{g} = (\zeta_{\mu,g},\bzeta_{\alpha,g}^{\prime},\bzeta_{\beta,g1}^{\prime},\ldots,\bzeta_{\beta,gN}^{\prime}\rightbr^{\prime}$ the $L-$vector of country random effects, where $\bzeta_{\alpha,g} = \leftbr\zeta_{\alpha,g1},\ldots,\zeta_{\alpha,gA}\rightbr^{\prime}$ is $A$--vector and $\bzeta_{\beta,gj} = \leftbr\zeta_{\beta,gj0}\comma\ldots\comma\zeta_{\beta,gjB_j}\rightbr^{\prime}$ is $(1+B_j)$--vector for $j=1,\ldots,N$.

\begin{proposition}
\label{sec:method:prop1}
Based on the previous definitions, the model in Eq.~\eqref{sec:method:second-eq} can be expressed in stacked form as
\begin{equation}
    \label{sec:method:eq-P1}
    y_{g,t+h} = \bz_{g,t+h}^{\prime}(\bgamma+\bzeta_{g}+\bpsi_{h}) + e_{g,t+h}.
\end{equation}
where $\bgamma$, $\bzeta_{g}$, and $\bpsi_h$ are the vectors of common, country, and hourly random effects, respectively.
\end{proposition}
\begin{proof}
    See Appendix \ref{sec:appendix}.
\end{proof}

\begin{proposition}
\label{sec:method:prop2}
The stacked-form model in Eq.~\eqref{sec:method:eq-P1} admits a ``low-frequency'' panel vector autoregressive representation as
\begin{equation}
    \label{sec:method:eq-P2}
    \by_{g,t} = Z_{g,t}\bigl(\bgamma+\bzeta_{g}\bigl)+\diag\bigl(Z_{g,t}\bpsi^{\prime}\bigl)+\be_{g,t},
\end{equation}
where $Z_{g,t} = \leftbr\bz_{g,t+0}\comma,\ldots,\comma\bz_{g,t+H-1}\rightbr^{\prime}$ and $\bpsi = \leftbr\bpsi_0\comma,\ldots\comma\bpsi_{H-1}\rightbr^{\prime}$ are matrices of size $H\times L$, $\by_{g,t} = \leftbr y_{g,t+0},\ldots,y_{g,t+H-1}\rightbr$, and $\be_{gt}\sim\mathcal{N}_{H}\leftbr \mathbf{0},\Sigma_{g})$ with $\Sigma_{g} = \diag(\sigma_{g0}^2\comma\ldots,\sigma_{g(H-1)}^2)$.
\end{proposition}
\begin{proof}
    See Appendix \ref{sec:appendix}.
\end{proof}

Now, we consider the likelihood function of the PRUMIDAS model based on the vectorized representation in Eq. \eqref{sec:method:eq-P2}. 
Let $\btheta = (\bgamma,\bpsi,\bzeta,\sigma^2,\blambda,\bchi,\bq,\br)$ denote the full vector of model parameters, where $\bpsi = (\bpsi_0^{\prime},\ldots,\bpsi_{H-1}^{\prime})^{\prime}$, $\bzeta = (\bzeta_1^{\prime},\ldots,\bzeta_G^{\prime})^{\prime}$ are defined in the preamble of Proposition \ref{sec:method:prop1}, while $\blambda = (\lambda_{0},\ldots,\lambda_{H-1})^{\prime}$ and $\bchi = (\chi_{1},\ldots,\chi_{G})^{\prime}$ are respectively the $H-$vector and $G-$vector of hourly- and country-specific variances.

Following Section~\ref{sec:model:common-random-effects}, we assume that $\bpsi_{h}\sim\mathcal{N}_{L}\leftbr \mathbf{0}_{L},Q)$ and $\bzeta_{g}\sim\mathcal{N}_{L}\leftbr \mathbf{0}_{L},R)$, where $Q = \diag\{\bigl\leftbr q_{\mu},\biota_{A}^{\prime}\otimes q_{\alpha},\biota_{N\sum_j(1+B_j)}^{\prime}\otimes q_{\beta}\bigr\rightbr^{\prime}\}$, 
$R = \diag\{\bigl\leftbr r_{\mu},\biota_{A}^{\prime}\otimes r_{\alpha},\biota_{N\sum_j(1+B_j)}^{\prime}\otimes r_{\beta}\bigr\rightbr^{\prime}\}$, and $\biota_A^{\prime}$ and $\biota_N^{\prime}$ are the vector of ones of size $A$ and $N$, respectively.
The likelihood function of the model in Eq. \eqref{sec:method:eq-P2} is
\begin{equation*}
L(\by|\btheta)\propto\prod_{t=1}^{T}\prod_{g=1}^{G}|\Sigma_{g}|^{-\frac{1}{2}}\exp\left\{-\frac{1}{2}\be_{g,t}^{\prime}\Sigma_{g}^{-1}\be_{g,t}\right\},
\end{equation*}
where $\be_{g,t} = (e_{g,t+0},\ldots,e_{g,t+H-1})^{\prime}$. 
Accordingly, the posterior distributions of the parameter vector is
\begin{equation*}    \pi(\btheta|\by)\propto\,L(\by|\btheta)\pi(\sigma^2)\pi(\mu)\pi(\balpha)\pi(\bbeta)\pi(R)\pi(Q)\prod_{g=1}^{G}\pi(\chi_g)\pi(\zeta_{g})\prod_{h=0}^{H-1}\pi(\bpsi_{h})\pi(\lambda_h),
\end{equation*}
which is analytically intractable. Therefore, we rely on MCMC methods to approximate the posterior distribution. 

Following \cite{casarin2018uncertainty}, we implement a multi-move sampling scheme, drawing blocks of parameters jointly rather than individually. 
To improve sampling efficiency and convergence, we apply Rao–Blackwellization by analytically integrating out the random effects $\bzeta_{g}$ and $\bpsi_{h}$ (see \citealp{roberts1997updating}).

The next proposition allows us to achieve this dual goal, enhancing the efficiency of the sampling process and leading to more accurate posterior estimates. As already mentioned, this will allow us to obtain a time-varying representation of the error's variance.

\begin{proposition}
\label{sec:method:prop3}
By marginalizing the random effects in Eq. \eqref{sec:method:eq-P1}, the model simplifies to 
\begin{equation}
    \label{sec:method:eq-P3}
    y_{g,t+h}= \bz_{g,t+h}^{\prime}\bgamma+\Tilde{e}_{g,t+h},
\end{equation}
where $\Tilde{e}_{g,t+h}\sim\mathcal{N}(0,\sigma_{gh,t}^{2})$ with $\sigma_{gh,t}^{2} = \sigma_{gh}^{2}+ \bz_{g,t+h}^{\prime}(R+Q)\bz_{g,t+h}$.
\end{proposition}
\begin{proof}
    See Appendix \ref{sec:appendix}.
\end{proof}

Additionally, the result obtained in Eq.~\eqref{sec:method:eq-P3} can be extended to Eq.~\eqref{sec:method:eq-P2} on a row-by-row basis following the stacked representation in Eq.~\eqref{sec:method:eq-P1}.
Thus, the conditional posterior distribution of $\by_{g,t}$ given $\bgamma$ is normally distributed with mean $\bz_{g,t+h}^{\prime}\bgamma$ and covariance matrix $\Sigma_{gt} = \diag(\sigma_{g0,t}^2\comma\ldots,\sigma_{g(H-1),t}^2)$.


We implement the following sampling scheme (\citealp{casarin2018uncertainty}):
\begin{enumerate}
    \item[(i)] Draw $\bgamma$, $\bpsi$, and $\bzeta$ from $p(\bgamma,\bpsi, \bzeta|\sigma^2,\blambda,\bchi,\bq,\br,\by)$.
        \begin{enumerate}
        \setlength{\itemsep}{2pt}
            \item[(i.1)] Draw $\bgamma$ from $p(\bgamma|\sigma^2,\blambda,\bchi,\bq,\br,\by)$.
            \item[(i.2)] Draw $\bpsi$, $\bzeta$ from $p(\bpsi, \bzeta|\bgamma,\sigma^2,\blambda,\bchi,\bq,\br,\by)$.
        \end{enumerate}
    \item[(ii)] Draw $\sigma^2$, $\blambda$, and $\bchi$ from $p(\sigma^2,\blambda,\bchi|\bgamma,\bpsi,\bzeta,\bq,\br,\by)$
\end{enumerate}

The MCMC algorithm iteratively generates random samples for the model parameters from their full conditional distributions. In practice, we discard the first $3.000$ iterations as burn-in and retain the remaining $10.000$ MCMC draws for posterior inference.


\section{Analysis of the European electricity markets}
\label{sec:analysis-application}

We apply our PRUMIDAS model to study the effects of forecasted RES generation, forecasted demand, and fossil fuel prices on the electricity prices at the European level. 
Our setup allows us to disentangle the country- and hourly-specific relationships between either of these covariates and the electricity prices, shedding light on the integration of European electricity markets.



We focus our analysis on the distribution of the random effects $\zeta_{\beta,gj}$ for each covariate $j$ and country $g$; notice that we have dropped the subscript $b$ relative to the lag, as we only consider the contemporaneous effects. 
In particular, we are interested in comparing the estimates of the overall country effects $\beta_{j}+\zeta_{\beta,gj}$, i.e. the sum of the common and country-specific random effects, against the average group effect $\beta_{j}$. This analysis highlights how countries deviate from average estimates, thereby providing indications for further policy interventions aimed at strengthening market integration. Figure~\ref{fig:country-random-effects-boxplot} exhibits the boxplots associated with each covariate, where each country effect is signaled by a specific color, and the common effect is colored in black.

\begin{figure}[t!]
\centering
\setlength{\tabcolsep}{0.001pt}
\begin{tabular}{ccc}     
    \includegraphics[width=5.2cm]{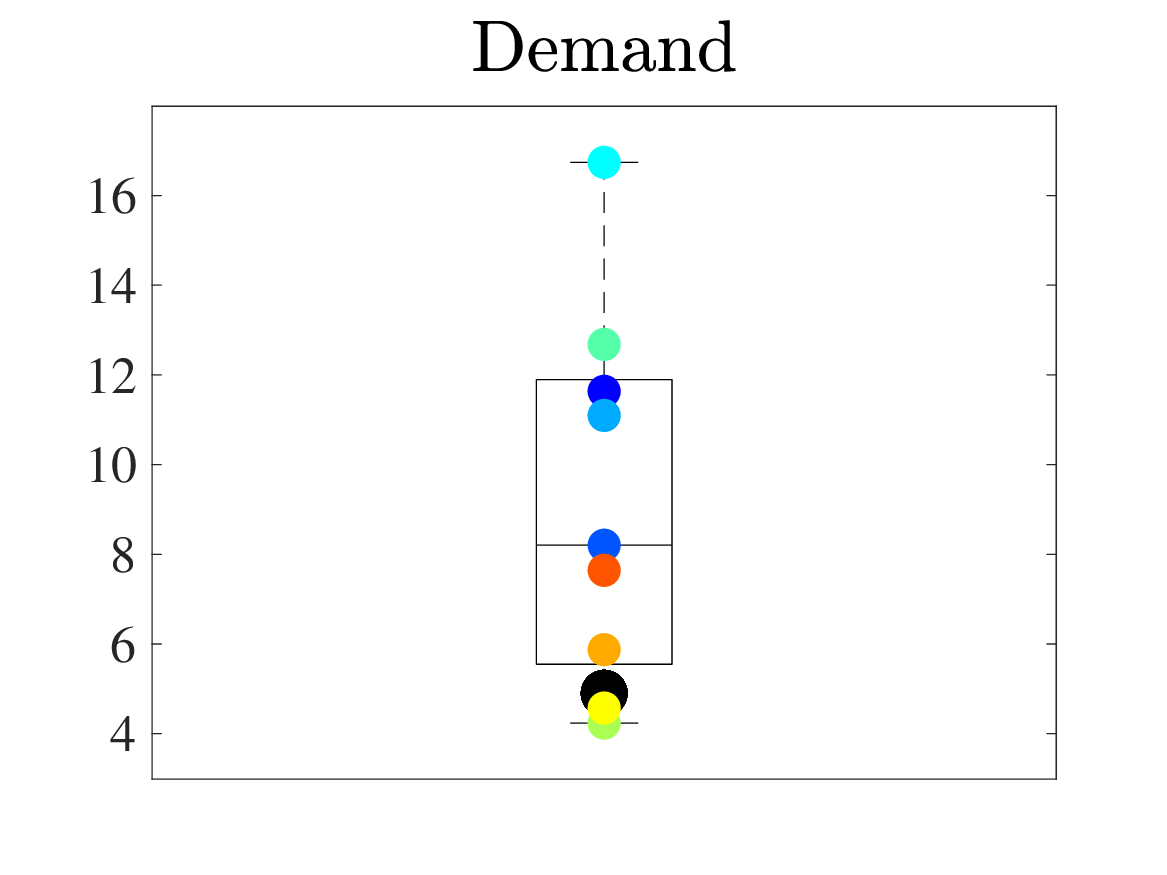} & 
    \includegraphics[width=5.2cm]{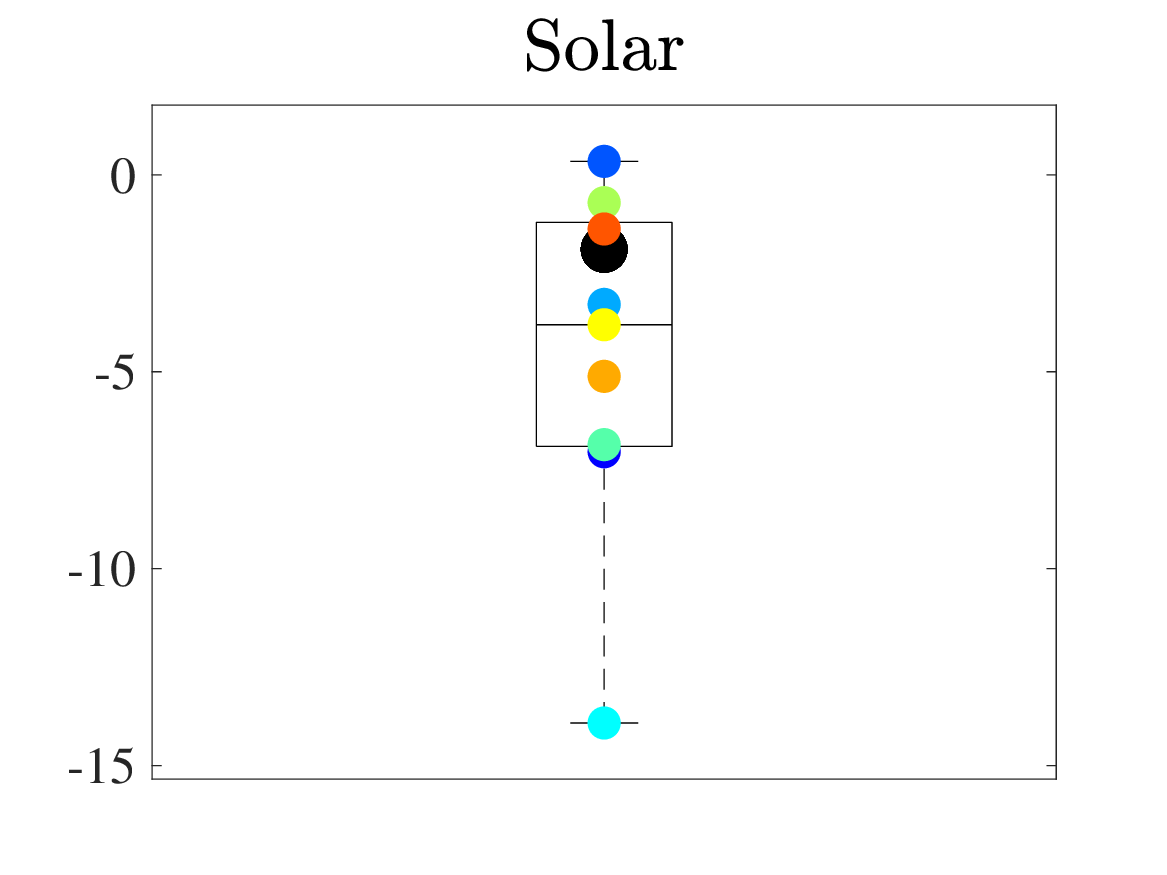} &
    \includegraphics[width=5.2cm]{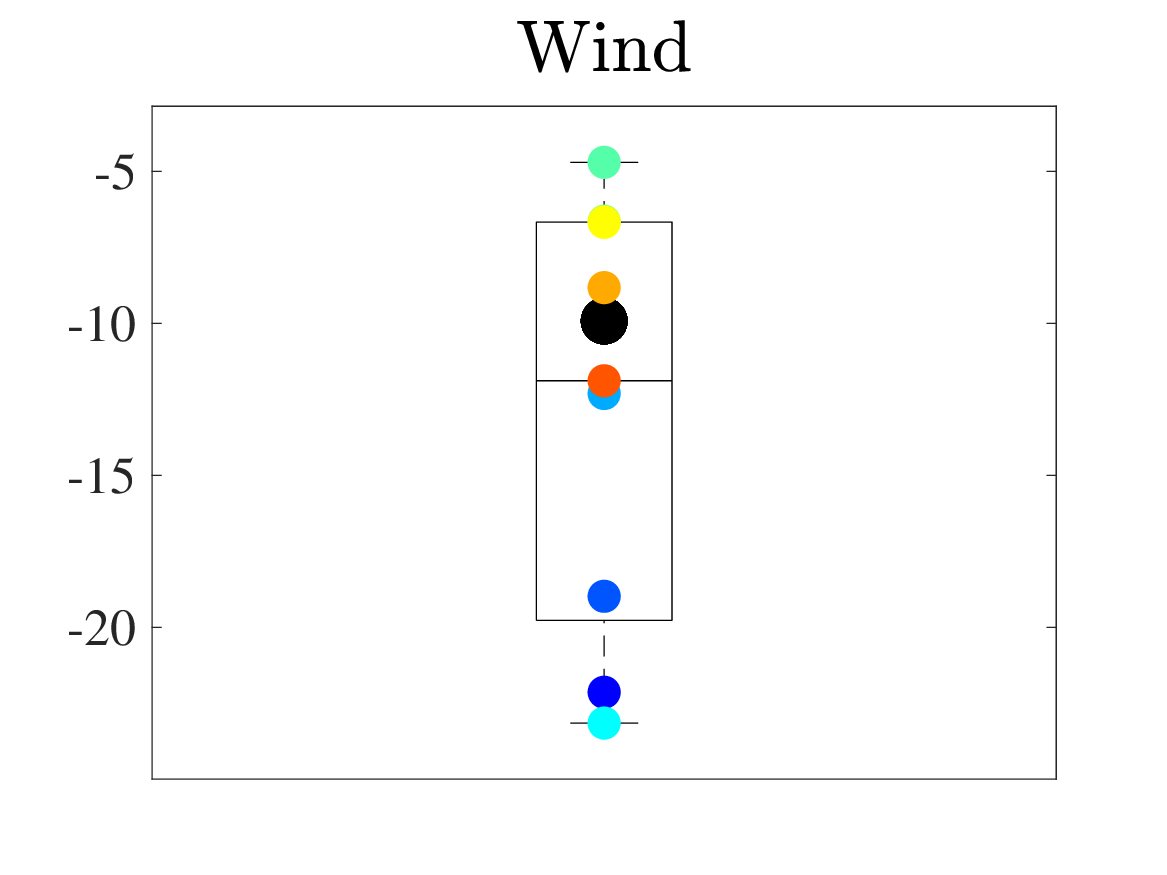} \\
    \includegraphics[width=5.2cm]{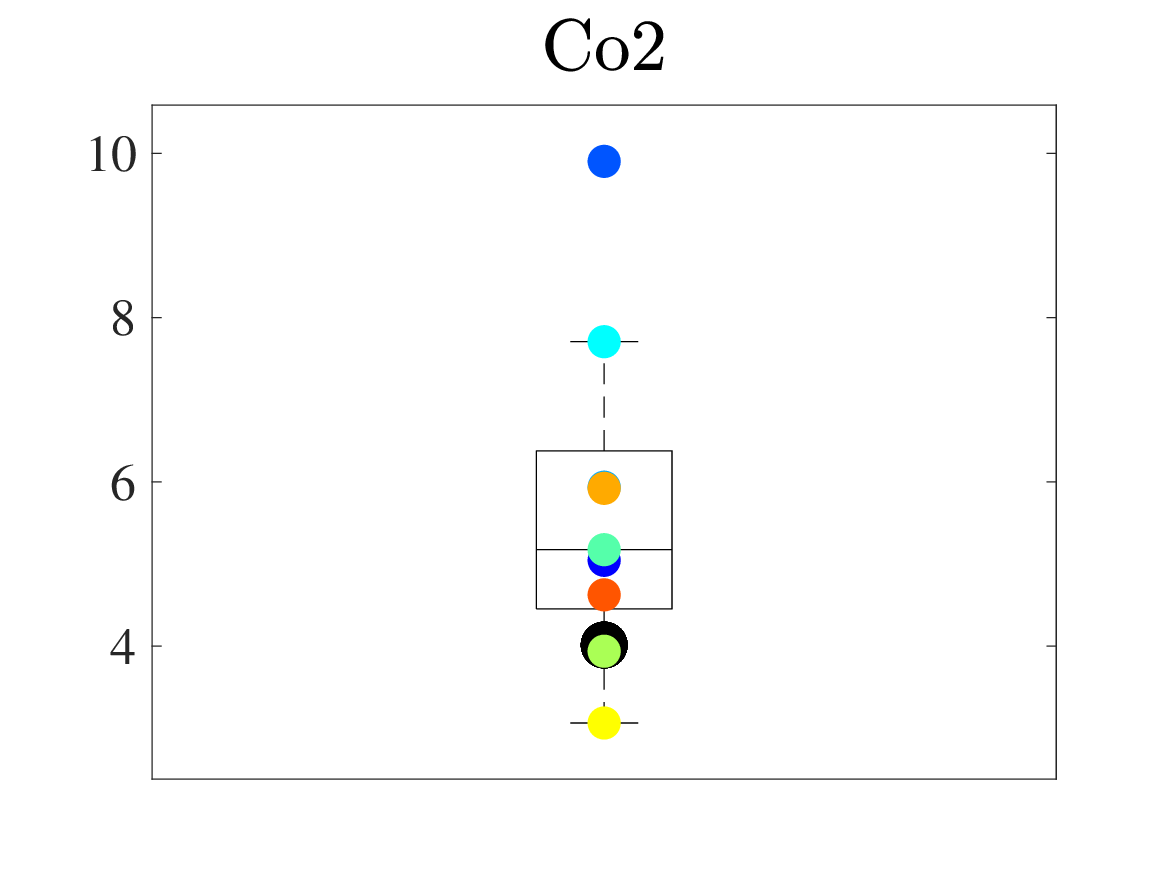} & 
    \includegraphics[width=5.2cm]{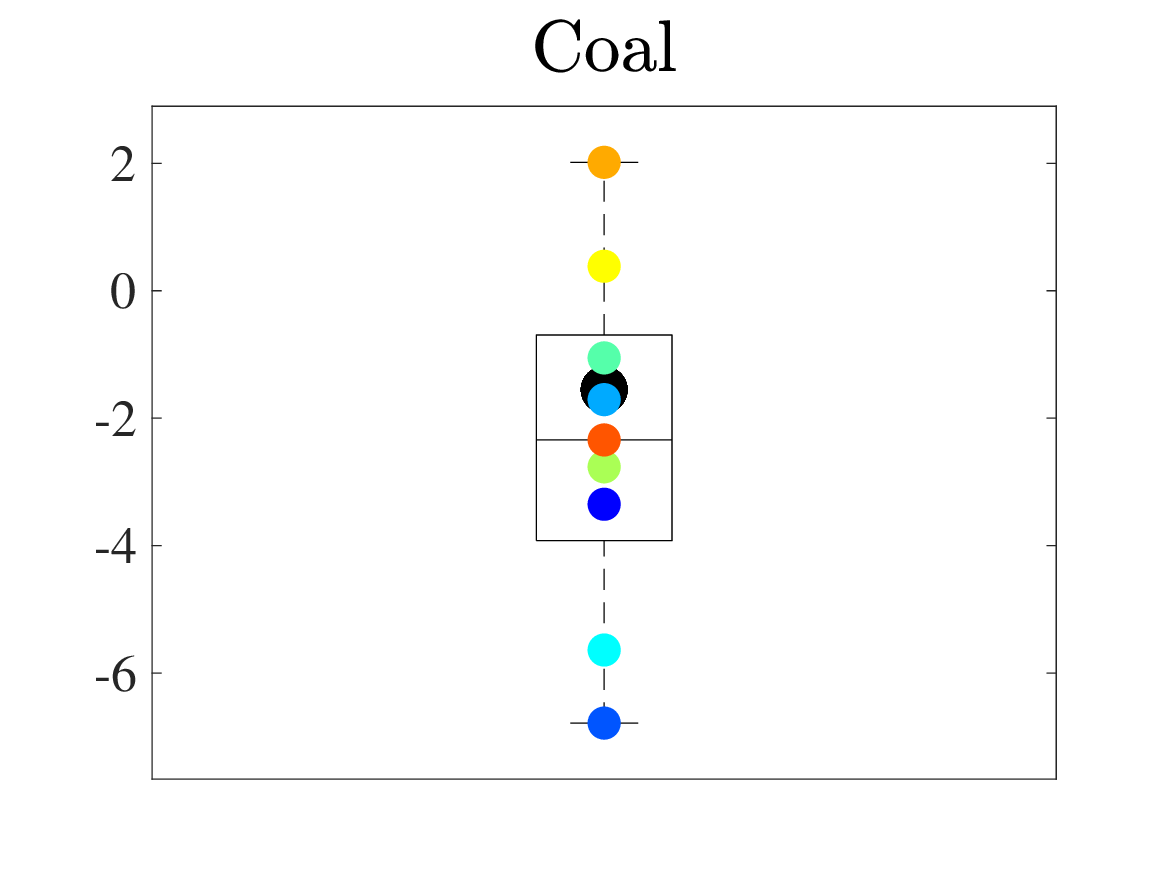} &
    \includegraphics[width=5.2cm]{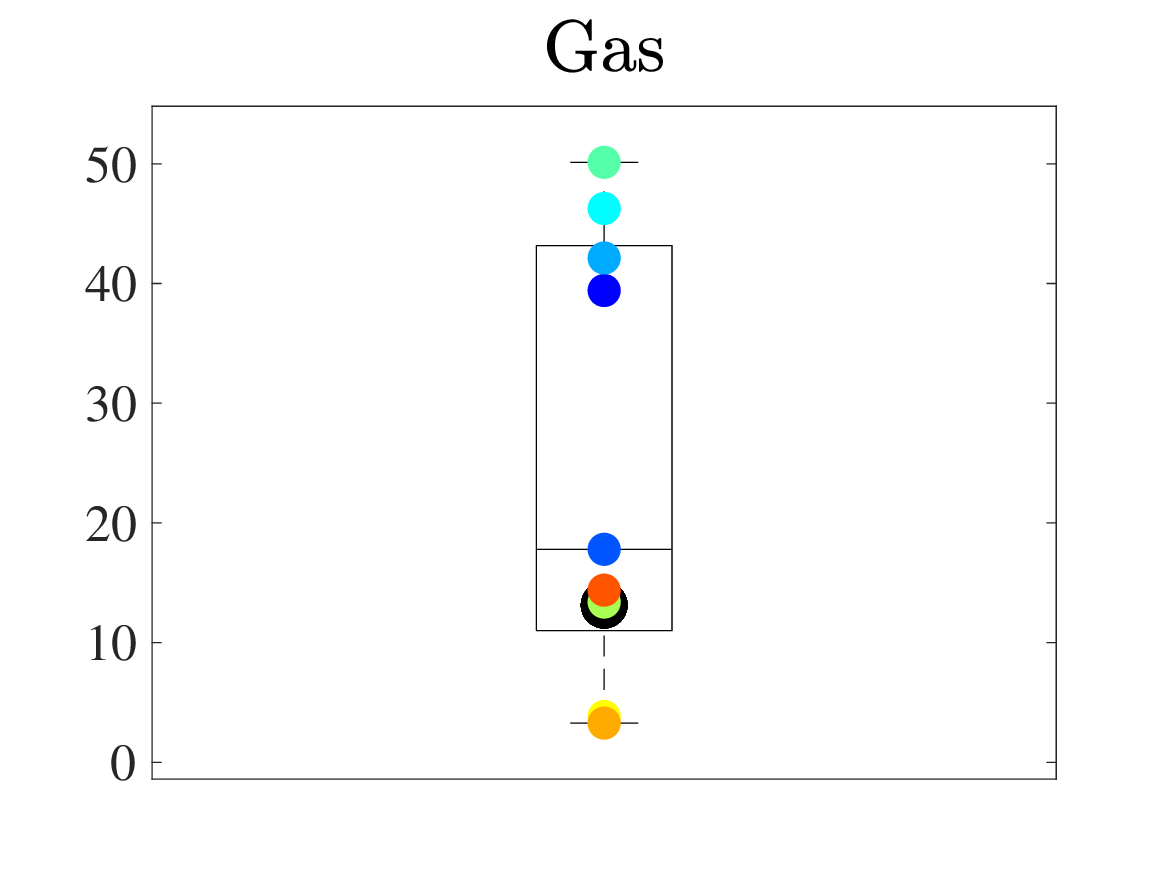} \\
    \end{tabular}
        \vspace{-1ex}
    {\scriptsize
    \legdot{black} Common effect \quad
    \legdot{blue1} Denmark \quad
    \legdot{blue2} Finland \quad
    \legdot{blue3} France \quad
    \legdot{acquamarina} Germany \\
    \legdot{limegreen} Italy \quad
    \legdot{lime} Norway \quad
    \legdot{yellow} Portugal \quad
    \legdot{orange} Spain \quad
    \legdot{red} Sweden
    }
    \caption{Boxplots for the estimates of the common and country-specific effects -- sum of the common and country-specific random effects -- by covariate -- demand, solar, wind,  Co$_2$, coal, and gas. The size of the effect is displayed on the vertical axis. The common effect, for each covariate, is represented by a black point, whereas each country is given a specific color.}
    \label{fig:country-random-effects-boxplot}
\end{figure}

Despite the complexity of the model, the empirical findings carry important implications for policymakers and energy market participant across Europe. 
The estimated country effects for renewable energy sources (RES) -- wind and solar generation -- are predominantly negative, implying a widespread price-reducing effect on electricity, except in Finland, where increased solar generation has been associated with almost no impact on electricity prices. 
In detail, countries such as Germany, Denmark, and Italy exhibit particularly strong negative effects from solar expansion (top center panel), which aligns with their substantial investments in photovoltaic capacity.
In terms of wind power generation (top right panel), Germany and Denmark again experience pronounced price reductions, likely reflecting their geographical exposure to the North Sea and significant power capacity.

Moreover, the distribution of country effects for both solar and wind generation is negatively skewed. 
This indicates that while nearly all countries have benefited from increased RES generation, a subset have seen disproportionately large gains - likely due to more aggressive investment strategies aimed at reducing dependence on fossil fuels.
These findings highlight the crucial role of RES in reducing electricity costs across most of Europe. 
While natural resources endowments (e.g. wind and solar potential) differ across regions, all nine countries in our study pursued energy transition policies, reinforcing the EU’s commitment to decarbonization and green energy as stated in Draghi's report.

From the perspective of market integration, our findings suggest that RES expansion supports market convergence.
Lower electricity prices in RES-intensive countries may encourage cross-border trade and more efficient resource allocation, enhancing integration. But given the intermittent nature of the RES and the limited storage of electricity, fossil fuel dependence is still necessary for energy security. Therefore, given the recent strong correlation between gas and electricity prices, we also examine how rising gas prices have affected electricity prices across European countries, beyond the frequently cited cases of Germany and Italy.

Analysing the country-specific effects of gas prices (bottom right panel), we find strong support that surging gas prices have significantly driven up electricity prices. 
We quantify this impact on a country basis and distinguish three distinct groups of countries.
The first group -- Italy, Germany, France, and Denmark -- exhibits the largest increases and lies in the right tail of the country effects distribution. 
As stated in \cite{ravazzolo2023price} and illustrated in Figure~\ref{fig:country-electricity-gas}, gas and electricity prices comove strongly in Central Europe, particularly in Germany and Italy, both major importers of Russian gas. 
The second group -- Finland, Sweden, and Norway -- shows moderate sensitivity, with effects clustering around the average common effect. 
In contrast, the third group -- Portugal and Spain -- shows minimal impact from gas price fluctuations. 
In this regard, the skewness of the distribution reveals that, except for Portugal and Spain, all European countries have experienced substantial price increases. 
This highlights a further implication of European integration based on a marginal pricing system: since the last power unit in Europe is typically a gas-fired power plant, concentrated dependence on a single energy source (e.g. Russian gas) can generate contagion effects, spreading shocks to countries not directly reliant on that source. Even nations with diversified or regionally distinct energy mixed feel the indirect effects of market interdependence. 

Figure \ref{fig:country-random-effects-boxplot} further illustrates that electricity prices are positively associated with demand (top left panel), with relatively smaller effects observed in Portugal and Spain. 
The distribution of the country-specific effects exhibits pronounced right skewness, indicating large idiosyncratic effects at the country level while suggesting that the common effect across countries remains small. 
In detail, Italy, Germany, France, and Denmark show the largest increases in electricity prices associated with rising energy demand.
These are also the countries most affected by recent surges in gas prices,. This pattern suggests that these countries may have responded to rising demand by relying more heavily on natural gas, thereby increasing their exposure to gas price volatility -- especially following the Russia's invasion of Ukraine.

\begin{figure}[t!]
\centering
\setlength{\tabcolsep}{0.001pt}
\begin{tabular}{lll}     
    \includegraphics[width=5.2cm]{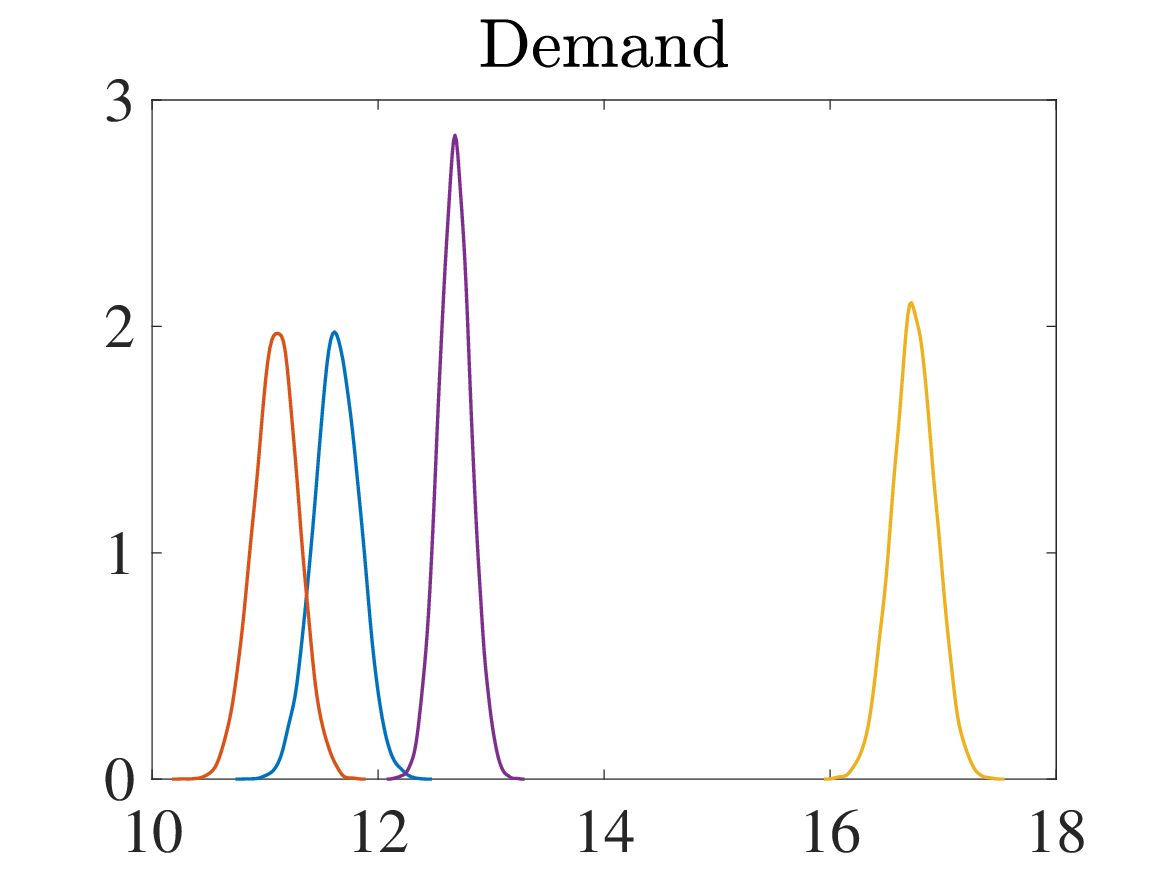} & 
    \includegraphics[width=5.2cm]{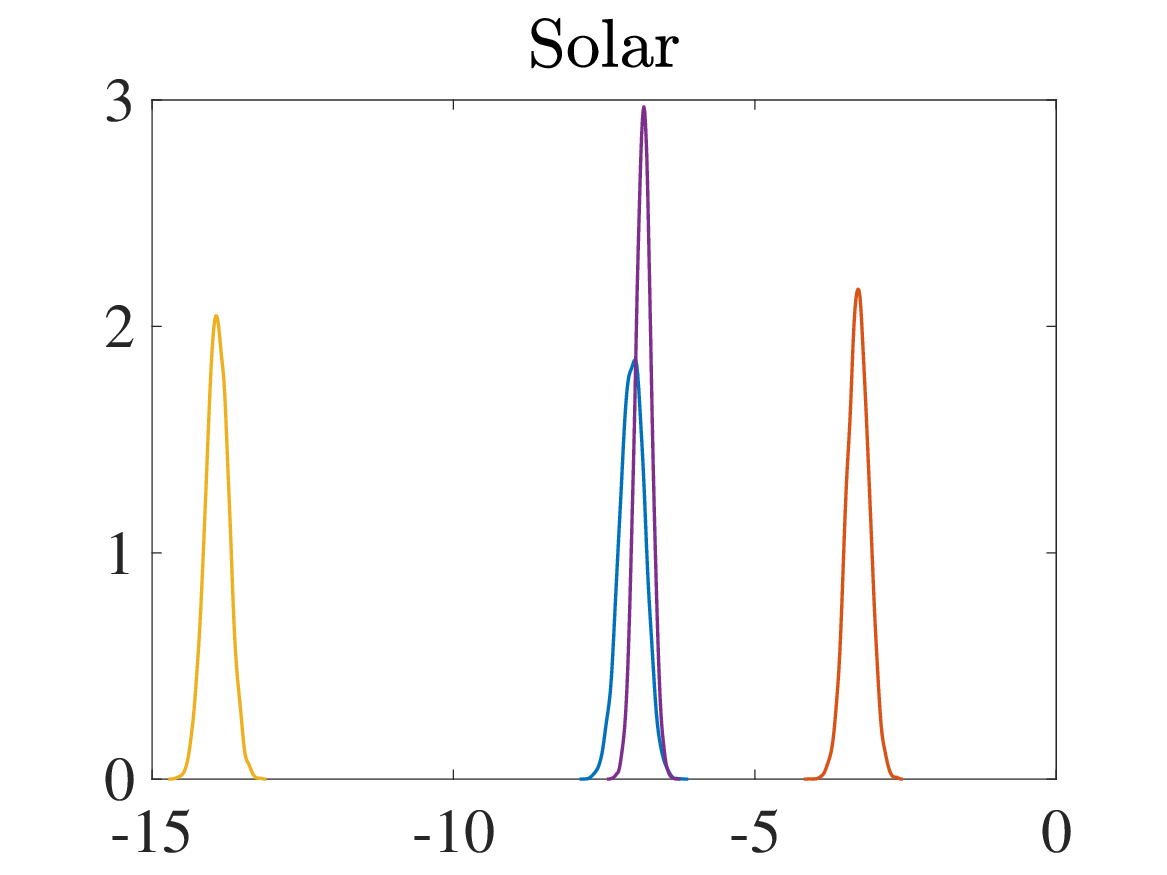} &
    \includegraphics[width=5.2cm]{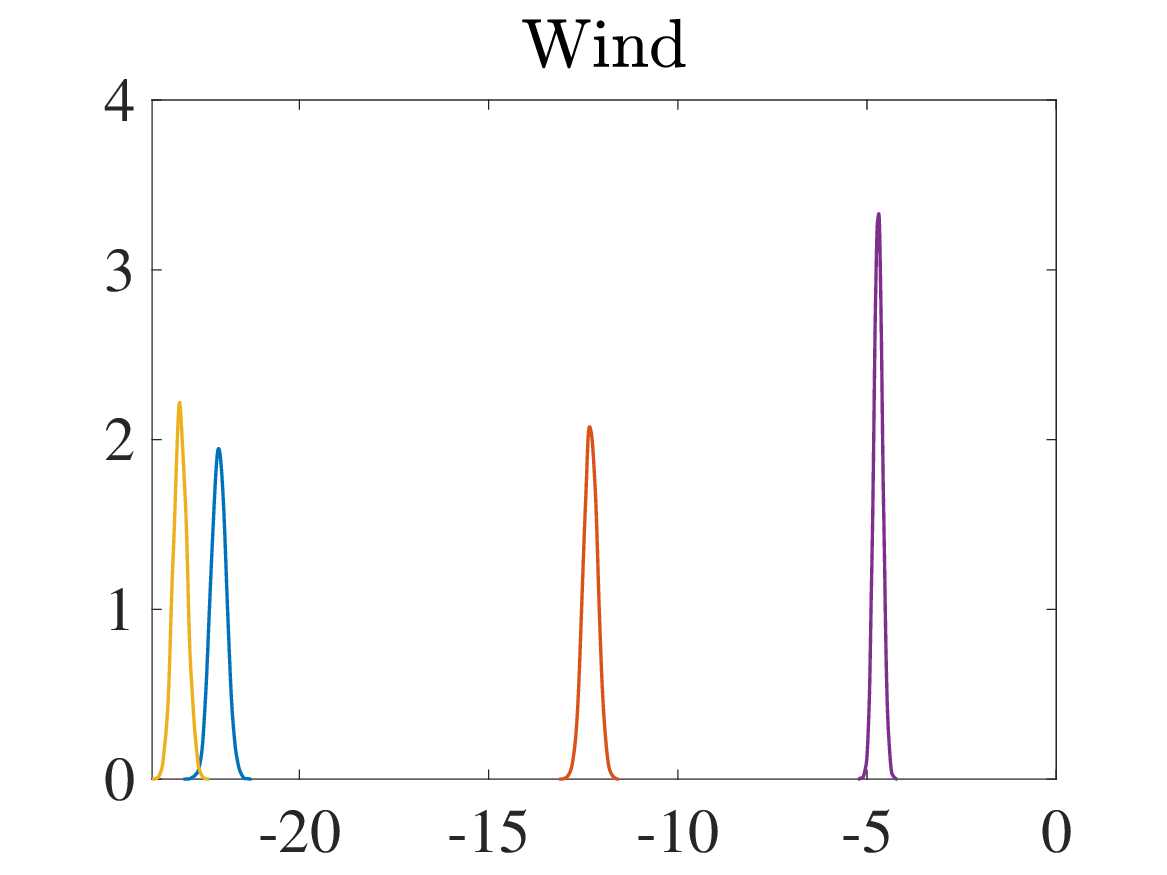} \\
    \includegraphics[width=5.2cm]{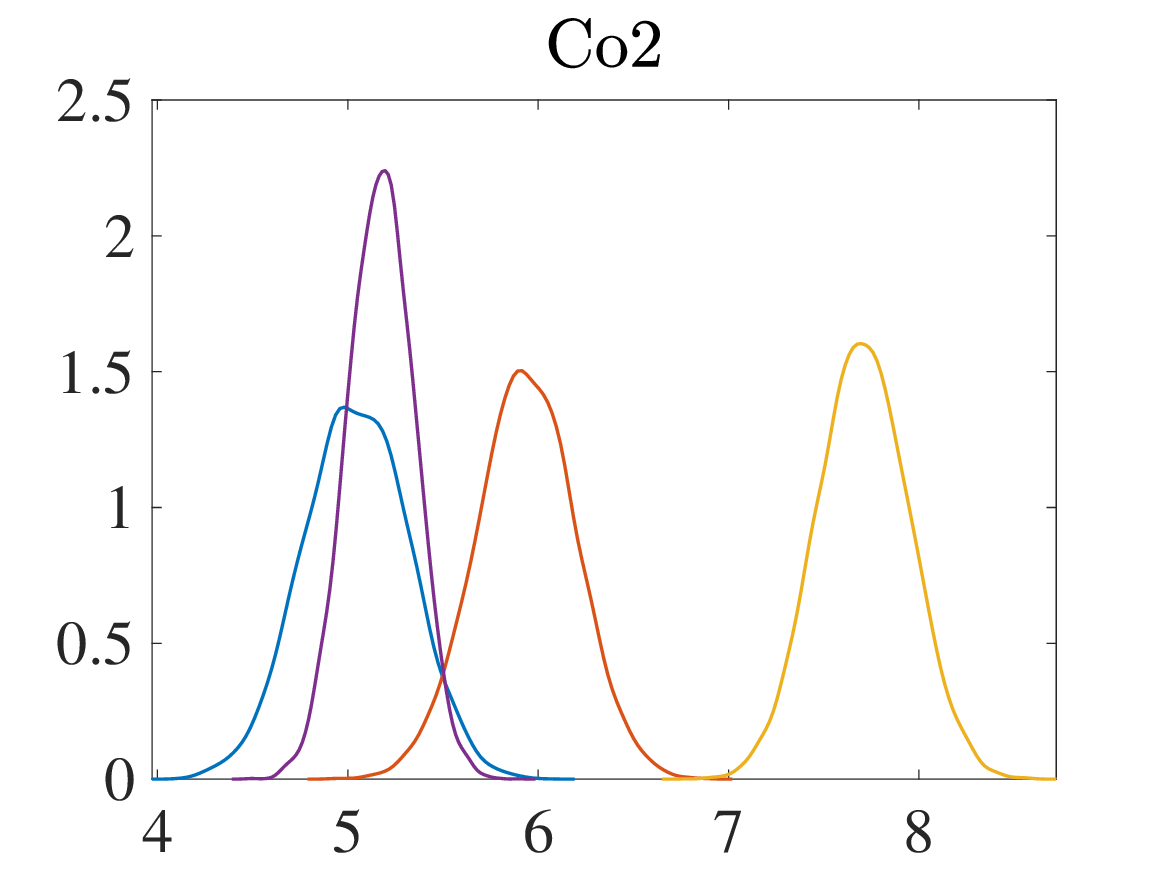} &
    \includegraphics[width=5.2cm]{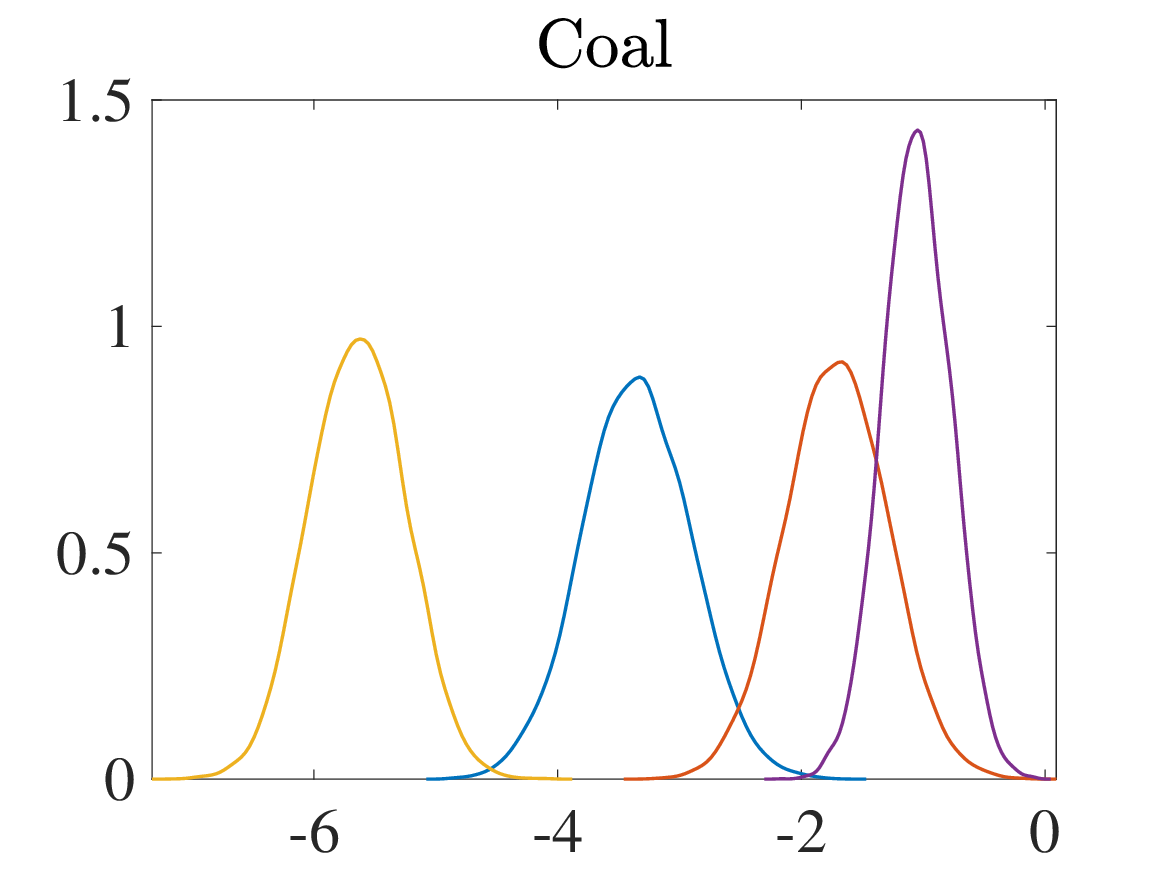} &
    \includegraphics[width=5.2cm]{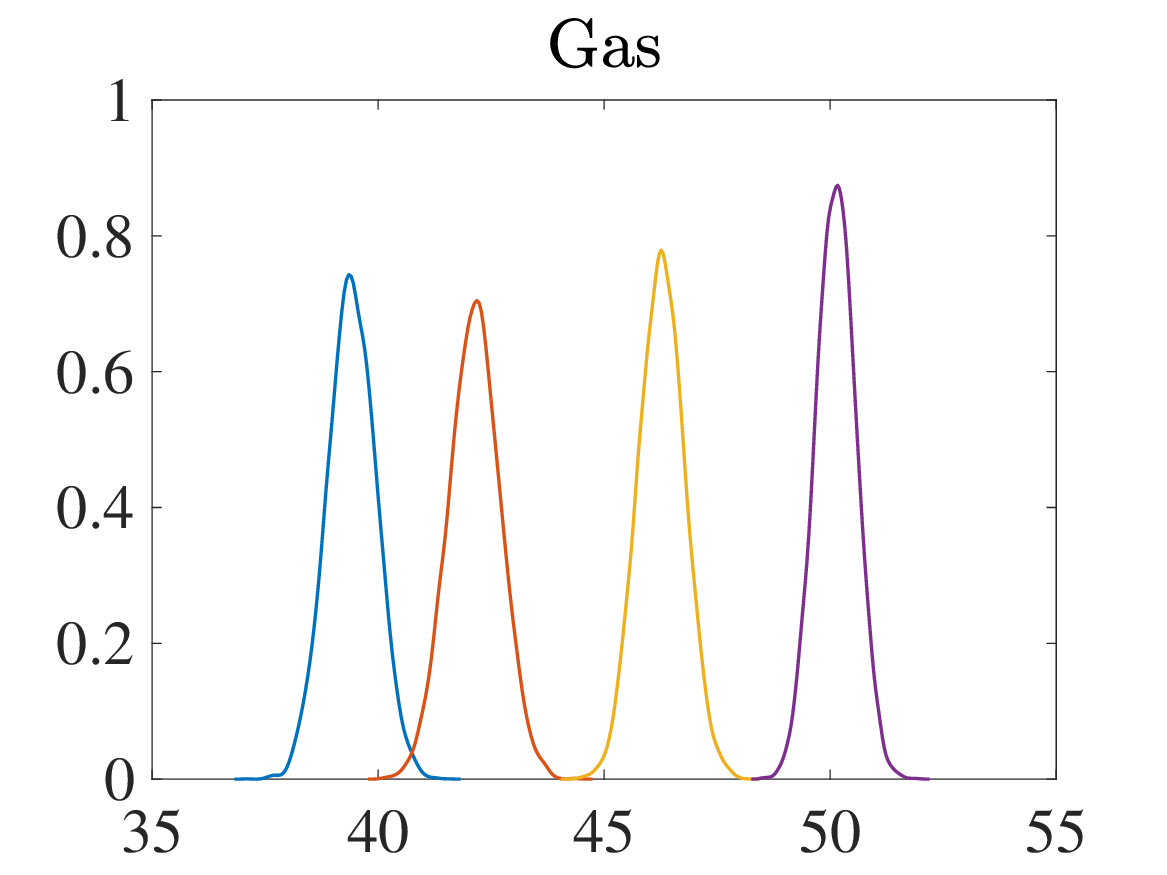} 
    \end{tabular}
    \caption{Full conditional posteriors for the country-effects -- sum of the common and country-specific random effects -- by covariate -- demand, solar, wind,  Co$_2$, coal, and gas -- for Denmark (blue), France (red), Germany (yellow), and Italy (violet). The size of the effect and the density are displayed on the horizontal and vertical axis, respectively.}
    \label{fig:posterior-DE-FE-GE-IT}
\end{figure}

Turning to the fossil fuel prices, we find evidence that the effect of coal (bottom center panel) is generally small and, for most countries, negative, except in Portugal and Spain, where coal price increases appear to be positively associated with electricity prices.
This limited impact is likely due to the relatively small share of coal in the contemporary European energy mix. 
Interestingly, in Germany and Finland, coal prices are negatively associated with electricity prices indicating a possible stabilizing role in these market. Finally, the price of EU ETS allowances for Co$_2$ (bottom left panel) is positively associated with electricity prices, reflecting the broader increase in the cost of emission permits over time.



\begin{figure}[t!]
\centering
\setlength{\tabcolsep}{0.001pt}
\begin{tabular}{lll}
    \includegraphics[width=5.2cm]{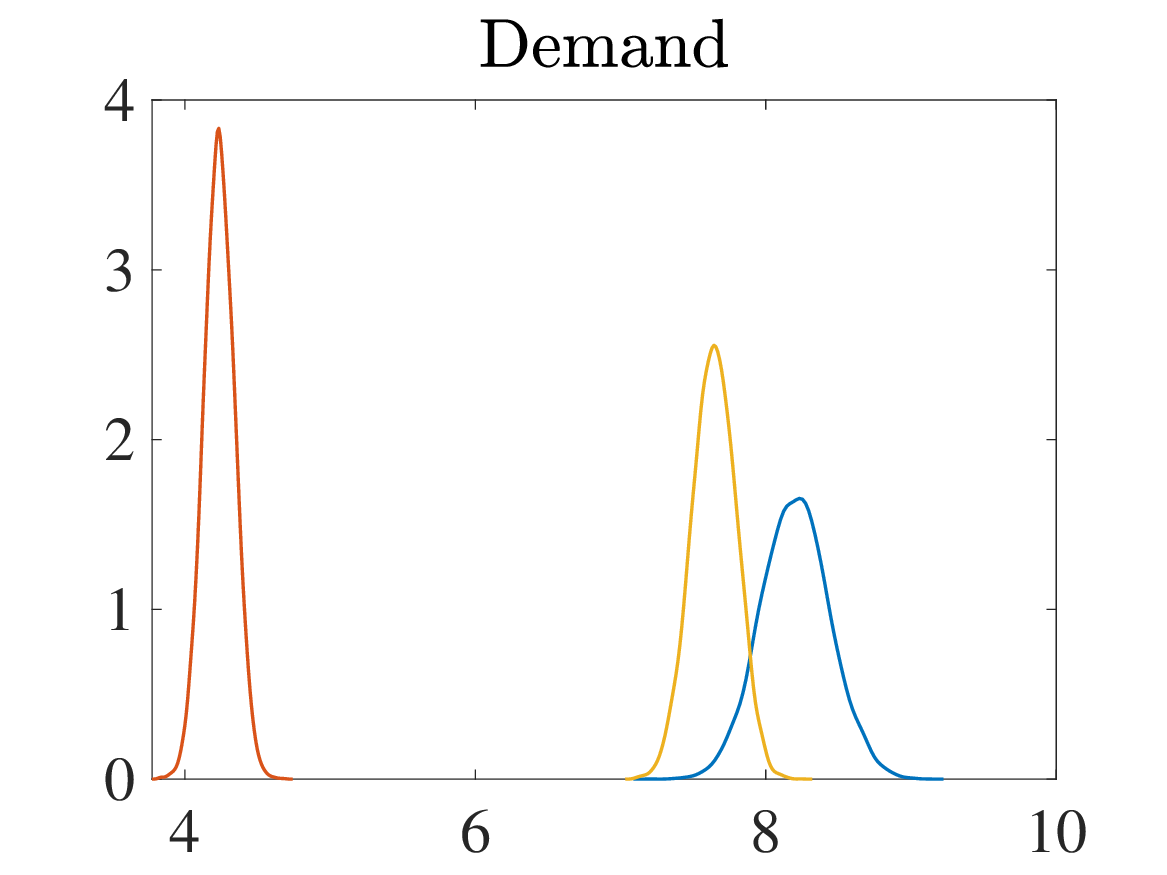} & 
    \includegraphics[width=5.2cm]{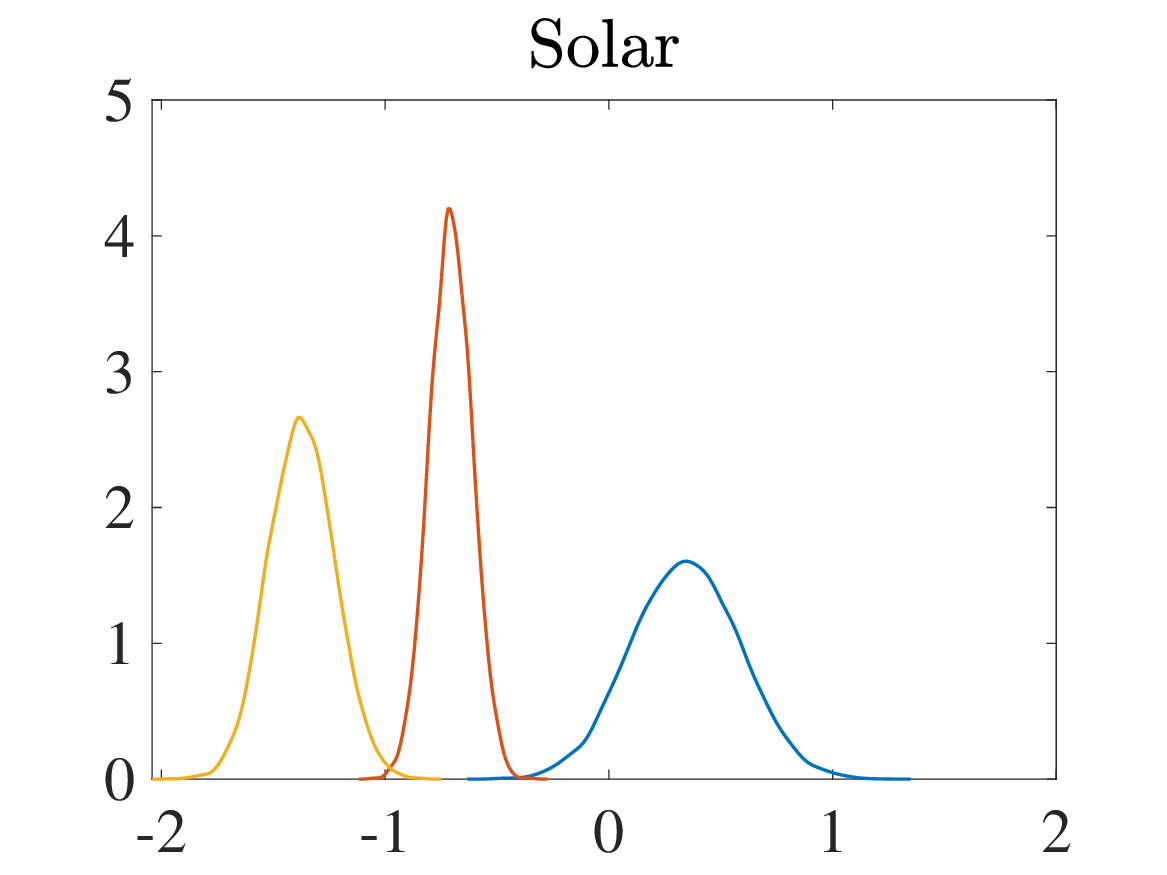} &
    \includegraphics[width=5.2cm]{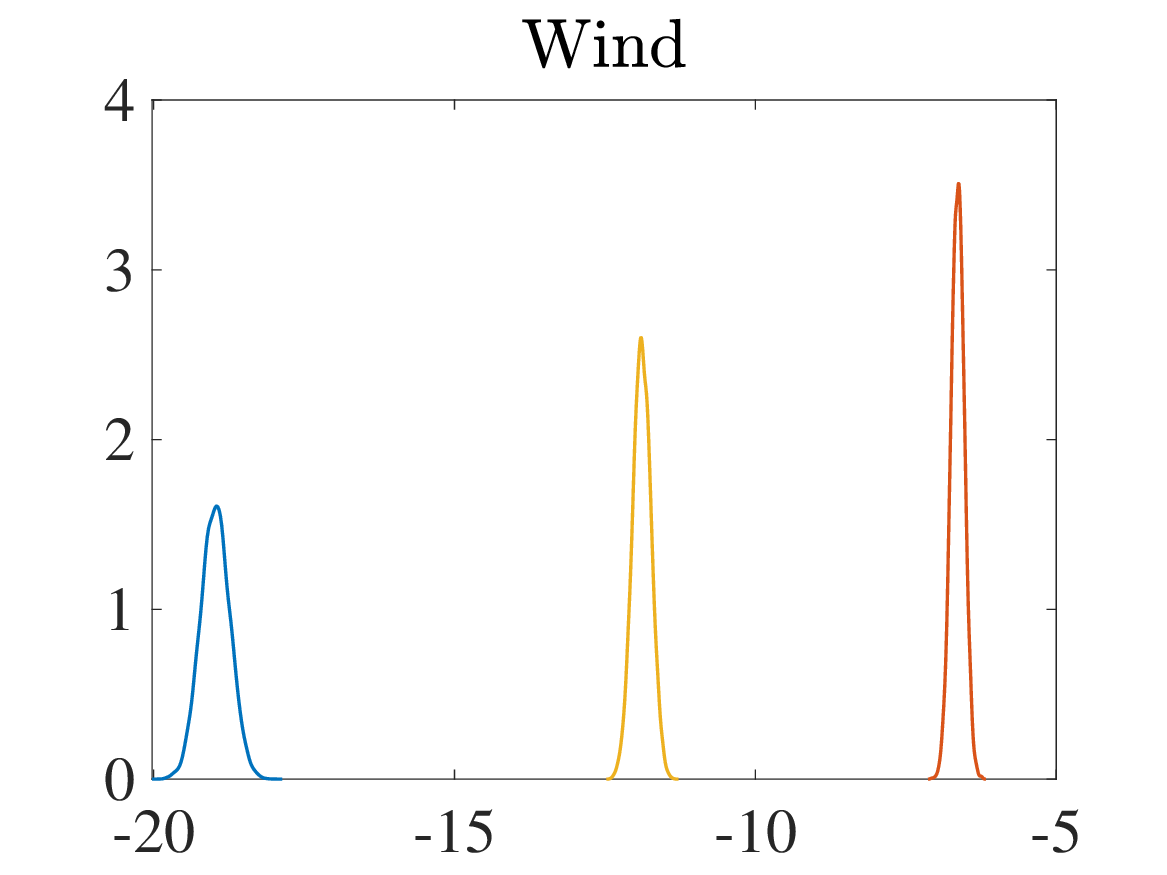} \\
    \includegraphics[width=5.2cm]{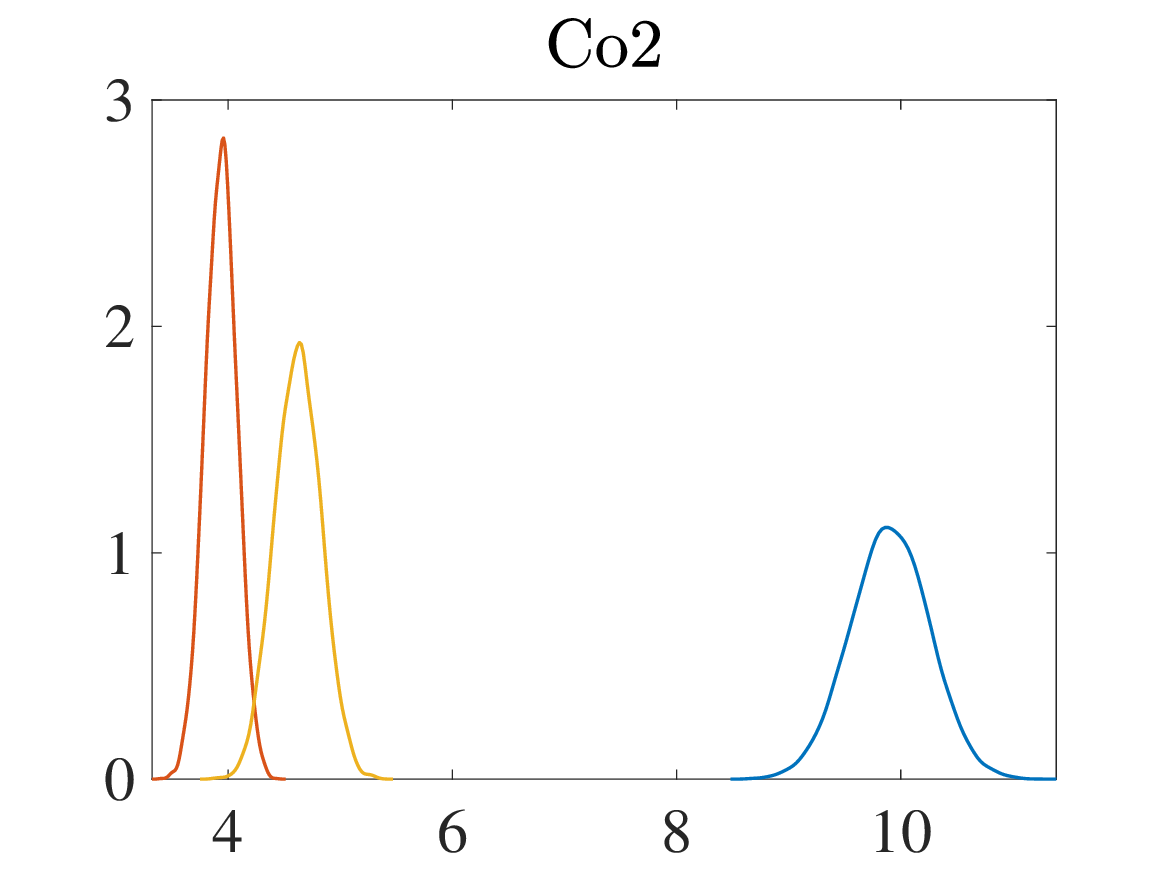} &
    \includegraphics[width=5.2cm]{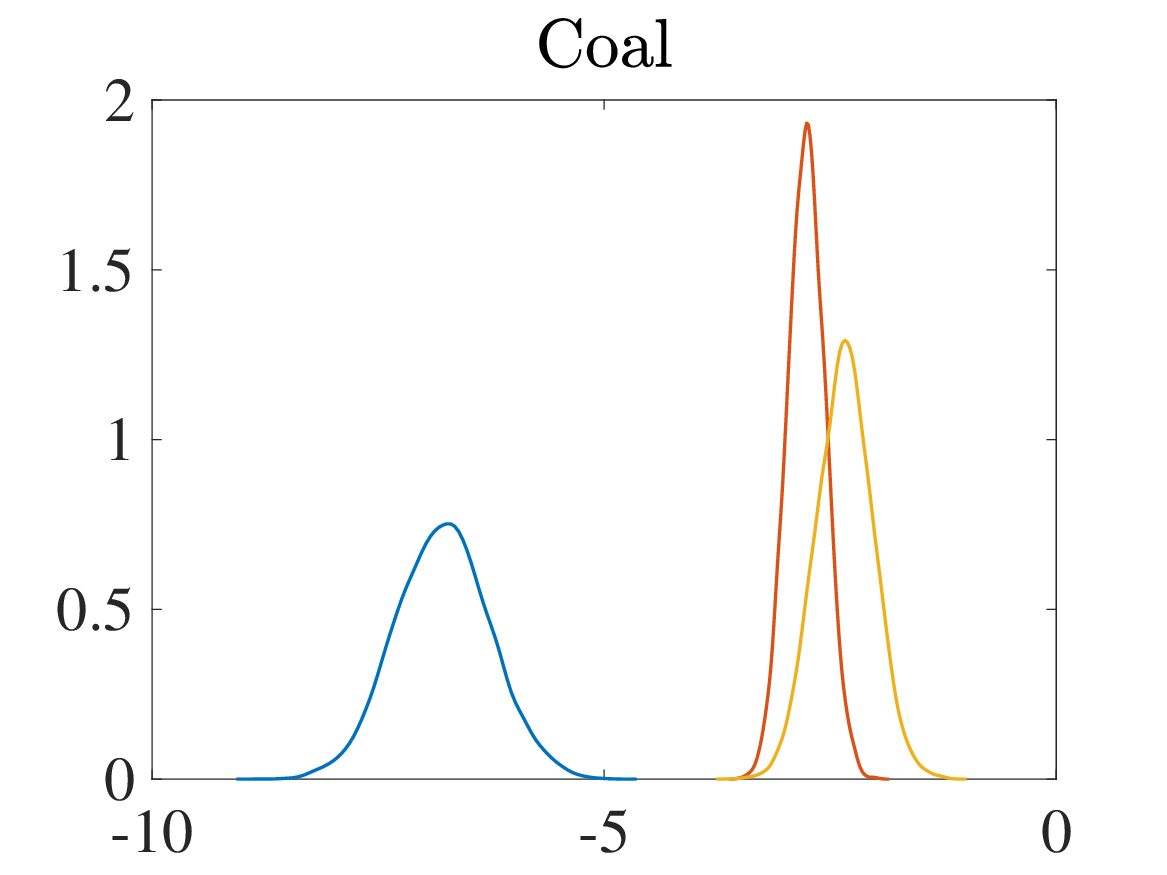} &
    \includegraphics[width=5.2cm]{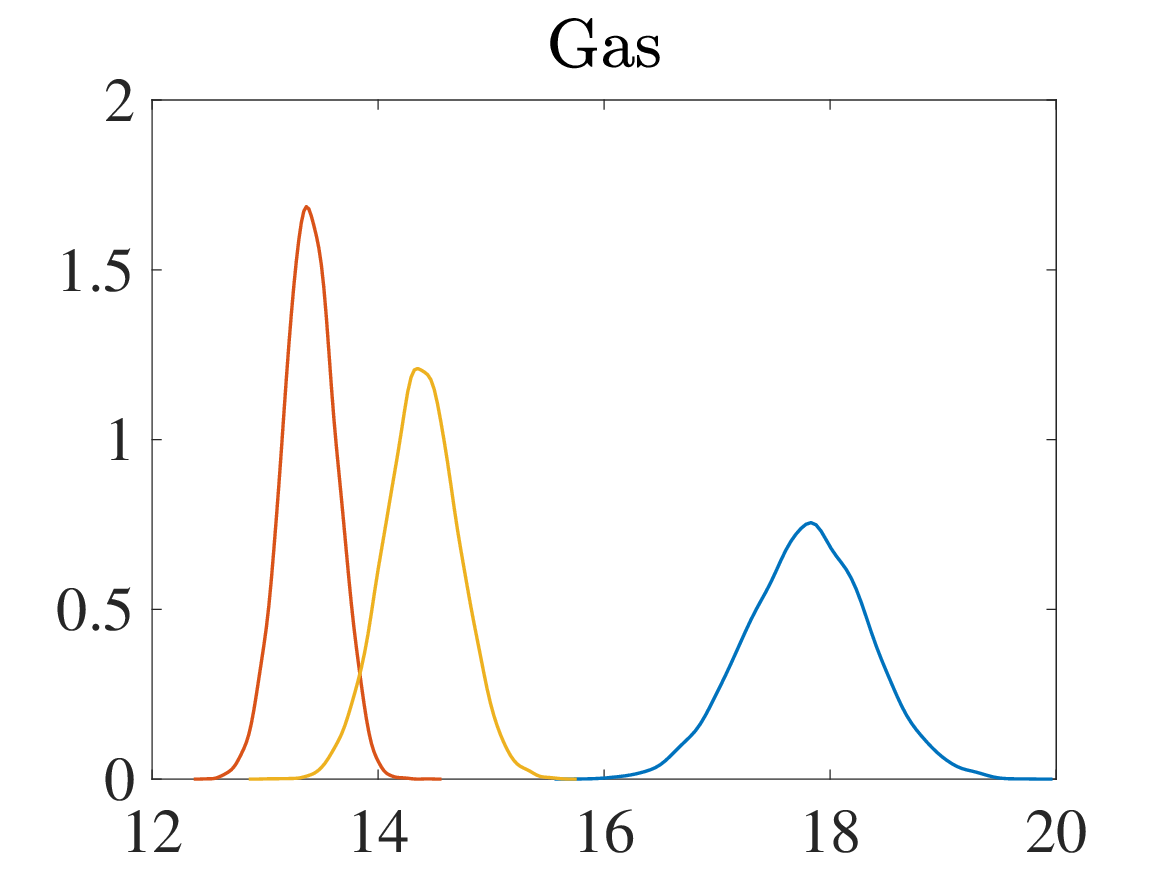} 
    \end{tabular}
 \caption{Full conditional posteriors for the country-effects -- sum of the common and country-specific random effects -- by covariate -- demand, solar, wind,  Co$_2$, coal, and gas -- for Finland (blue), Norway (red), Sweden (yellow). The size of the effect and the density are displayed on the horizontal and vertical axis, respectively.} 
\label{fig:posterior-FI-NO-SW} 
\end{figure}

The analysis of full conditional posterior distributions across the nine European countries provides valuable insights into the uncertainty surrounding the estimated effects and the extent to which countries share similar posterior distributions. Figures \ref{fig:posterior-DE-FE-GE-IT}, \ref{fig:posterior-FI-NO-SW}, and \ref{fig:posterior-PT-SP} display the full conditional posteriors of the country-specific effects, $\beta_{j}+\zeta_{\beta,gj}$, for the three groups of countries identified in the previous discussion.

Figure \ref{fig:posterior-DE-FE-GE-IT} presents results for the first group of countries -- Denmark (blue), France (red), Germany (yellow), and Italy (violet). 
These countries are grouped together not only due to their geographical proximity -- forming a ``continental'' cluster -- but also due to their similar responses in electricity prices following Russia's invasion of Ukraine. 
Notably, the posterior distributions for the forecasted demand and RES are relatively tight, indicating low uncertainty, compared to those for fossil fuel prices.
For the forecasted demand (top left panel), Germany has the largest positive effect suggesting that electricity prices are more sensitive to demand fluctuations than in the other three countries, which show relatively similar effects.

Regarding the RES, forecasted solar generation (top center panel) has a negative effect across all countries, though the magnitude varies.
Germany shows the strongest negative effect, consistent with the findings in the literature, while electricity prices in France appear to be only weakly influenced by solar generation. 
The top right panel, which displays the effect of the forecasted wind generation, reveals similar patterns between Germany and Denmark, while France and Italy show more moderate negative effects.

In contrast, the effects of fossil fuel prices, particularly natural gas, vary markedly across countries. 
The magnitude of the gas price effect is broadly comparable across the four countries, with the smallest impact in Denmark and France and the largest in Germany and Italy.
This pattern reflects the high dependence of Germany and Italy on Russian gas before the invasion of Ukraine, whereas France faced significant outages in its nuclear power plants, becoming a net energy importer during this period (\citealp{gaulier2023energy}). 



For the second group of countries, Finland (blue), Norway (red), and Sweden (yellow), Figure \ref{fig:posterior-FI-NO-SW} provides the posterior distributions. 
A striking feature is the high level of uncertainty in Finland’s estimates, particularly for Co$_2$ and coal.
Finland’s posterior distributions often diverge from those of Norway and Sweden, which display tighter and more aligned distributions. 
This divergence is consistent with Finland’s reliance on Russian gas imports and its need to rapidly diversify its energy mix following the onset of the crisis \citep{vaden2023energy}. 
Norway, by contrast, shows smaller uncertainty and, interestingly, the impact of demand on electricity prices appears to be weaker there than in the other two countries.

\begin{figure}[t!]
\centering
\setlength{\tabcolsep}{0.001pt}
\begin{tabular}{lll}     
    \includegraphics[width=5.2cm]{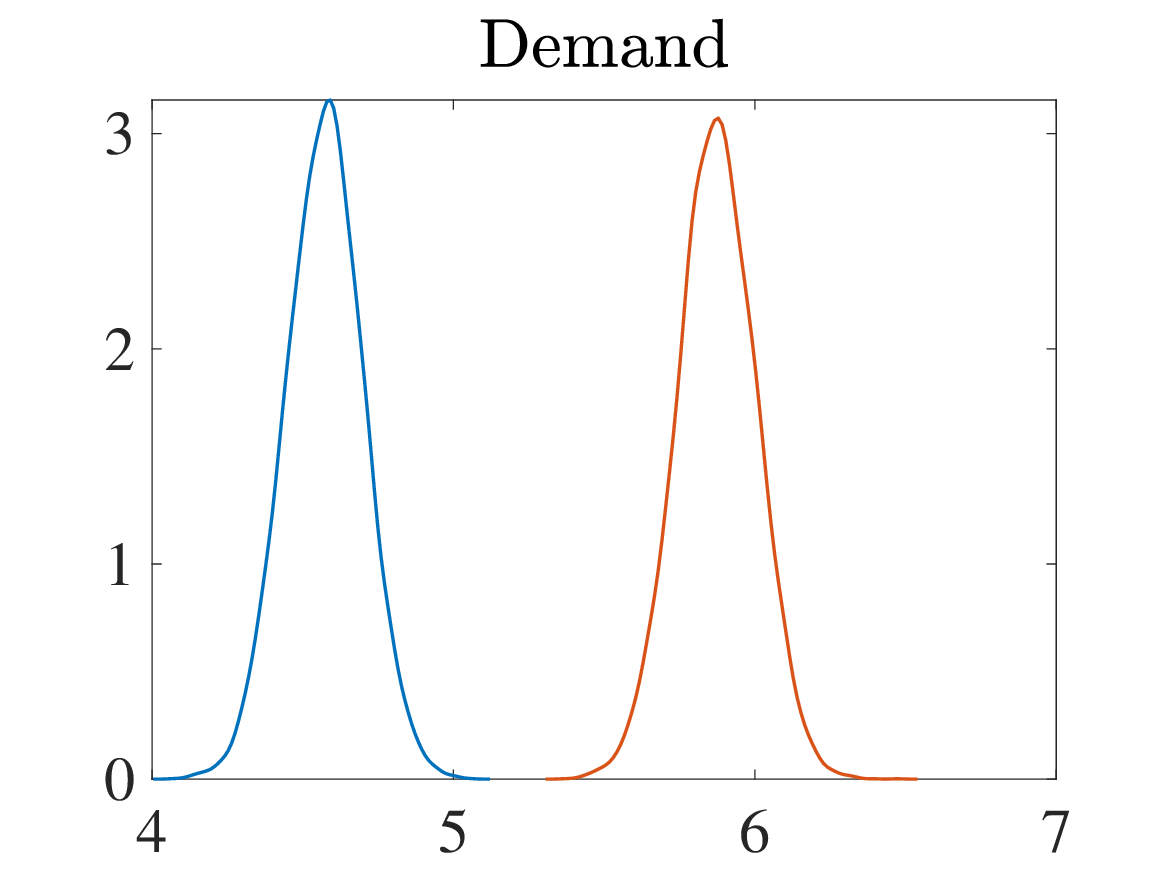} & 
    \includegraphics[width=5.2cm]{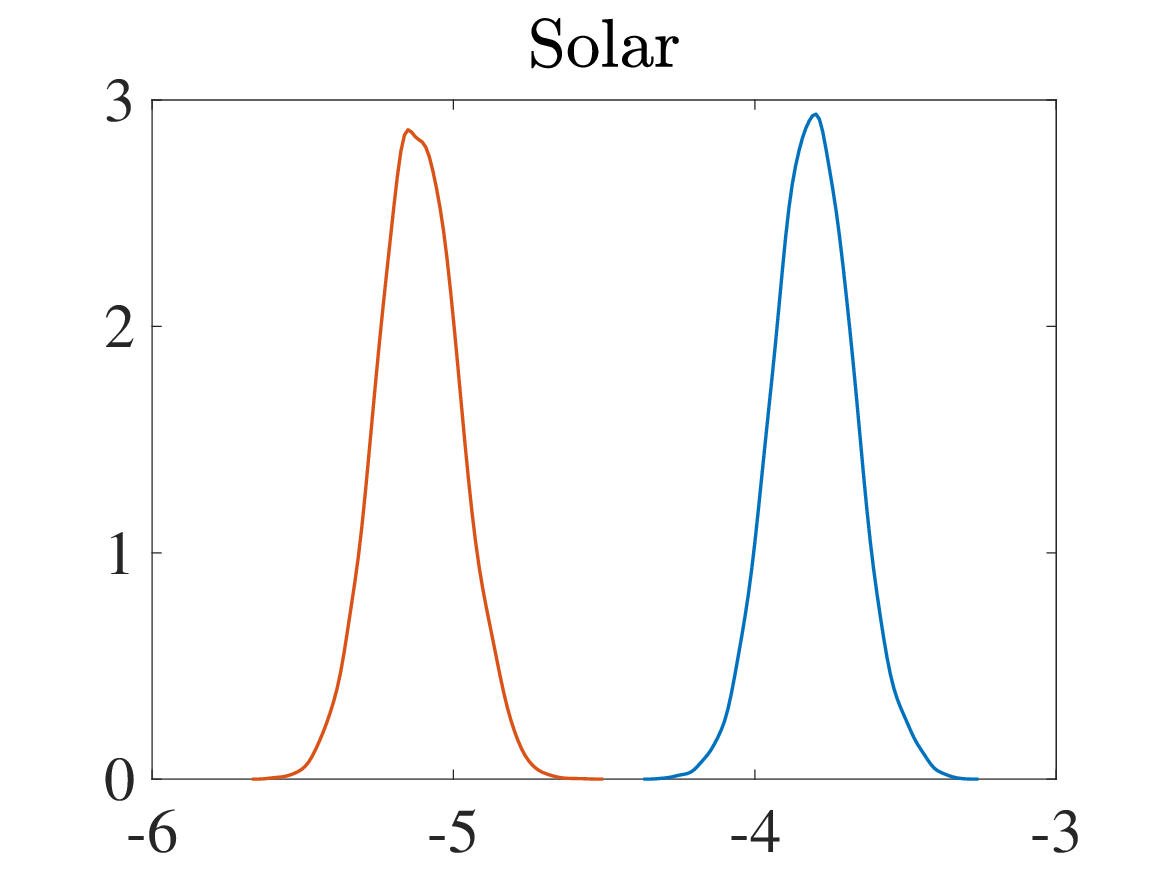} &
    \includegraphics[width=5.2cm]{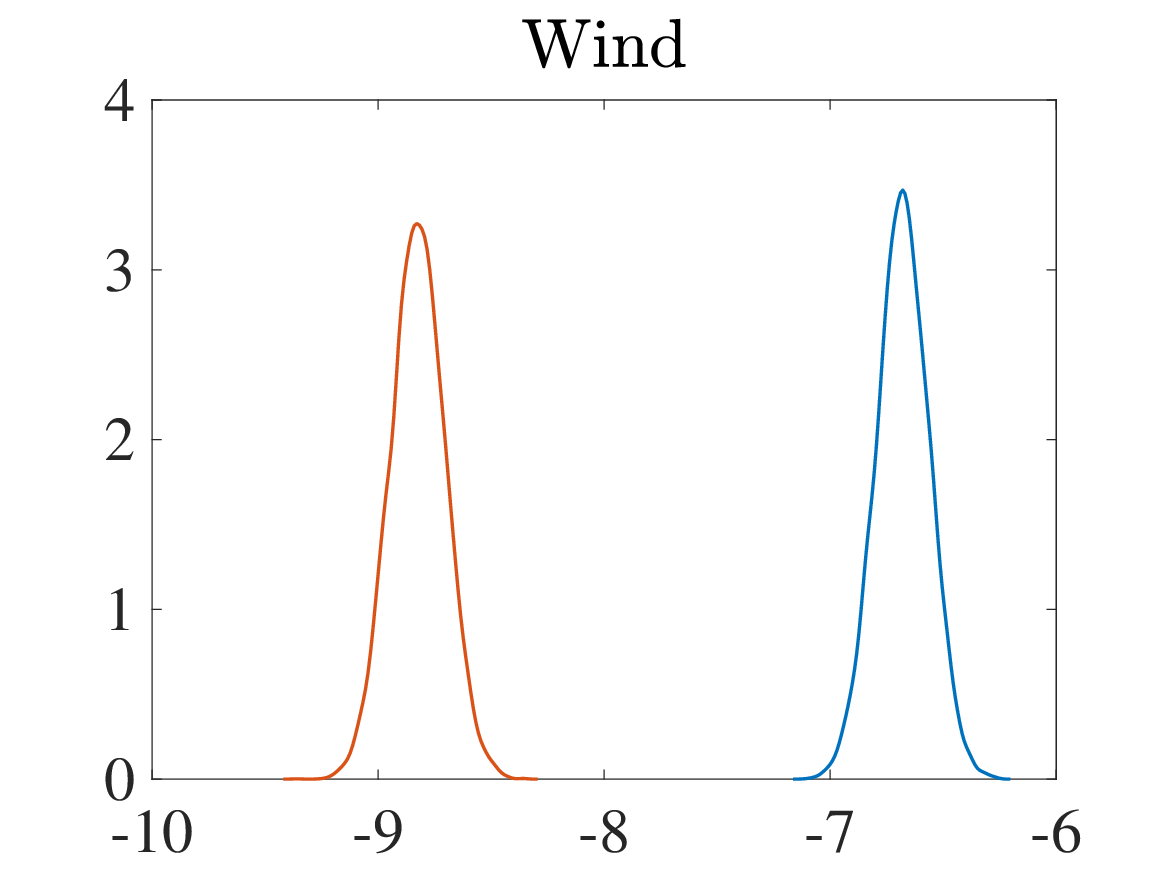} \\
    \includegraphics[width=5.2cm]{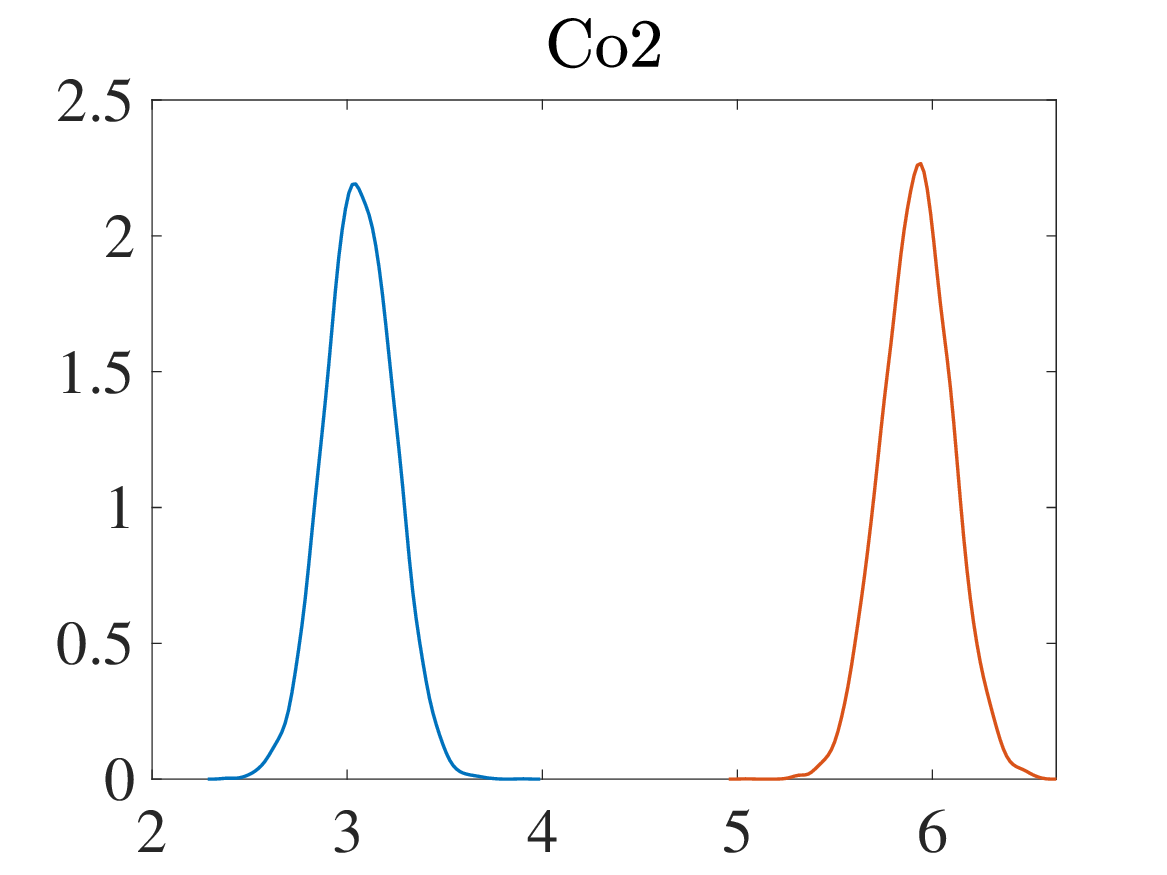} &
    \includegraphics[width=5.2cm]{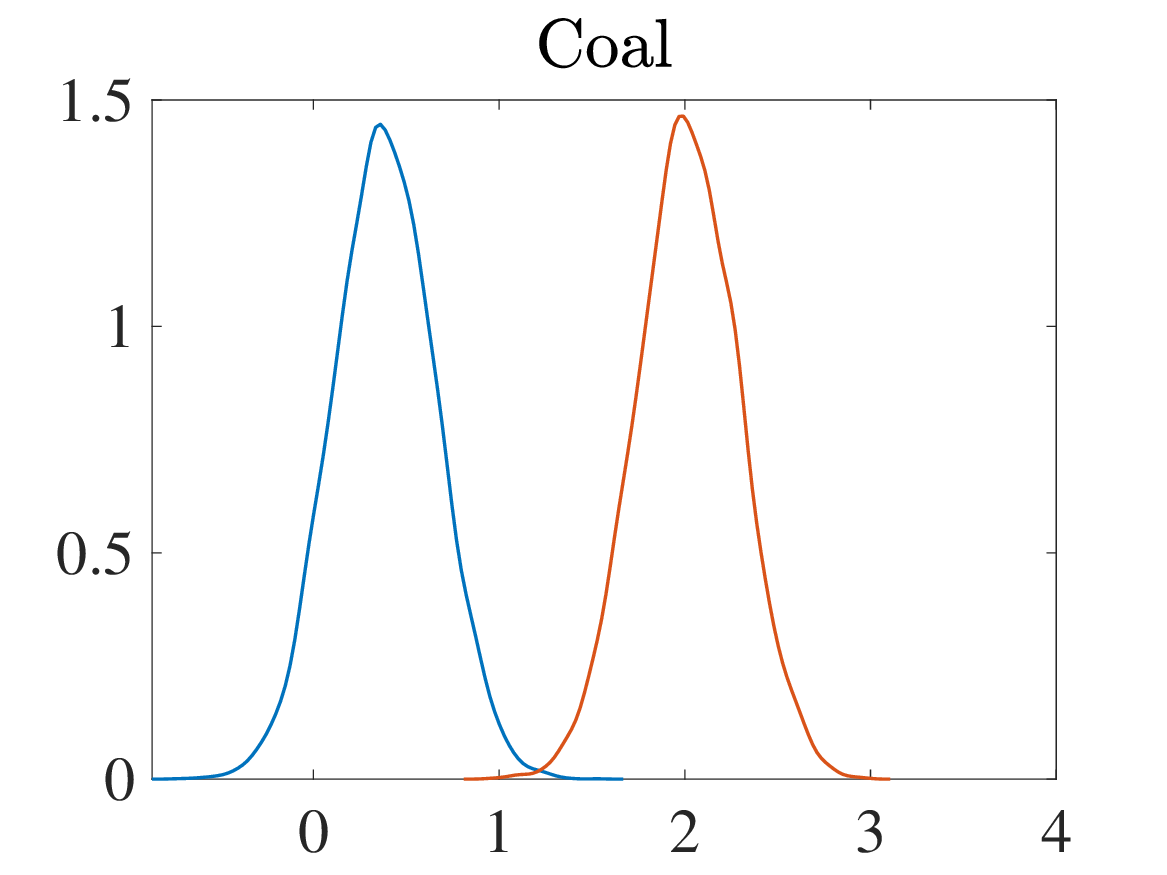} &
    \includegraphics[width=5.2cm]{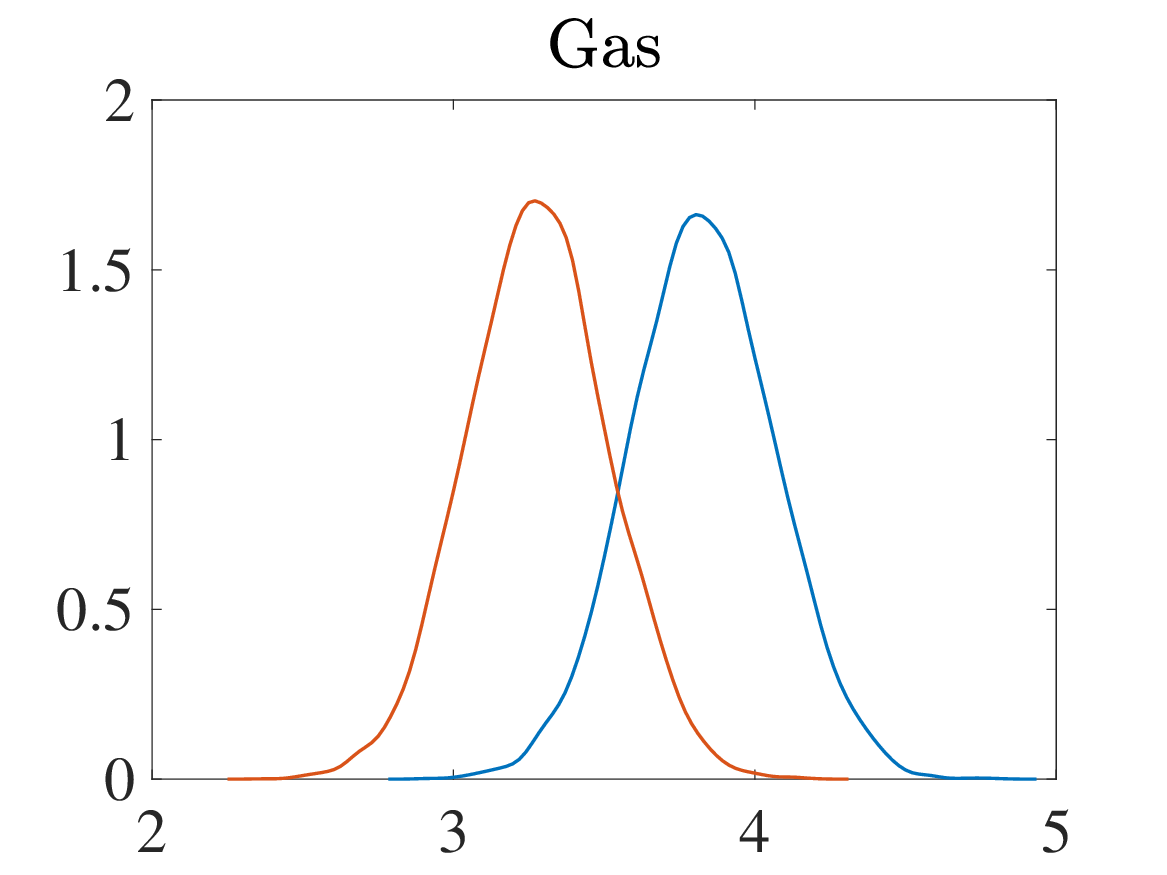} 
    \end{tabular}
\caption{Full conditional posteriors for the country-effects -- sum of the common and country-specific random effects -- by covariate -- demand, solar, wind,  Co$_2$, coal, and gas -- for Portugal (blue) and Spain (red). The size of the effect and the density are displayed on the horizontal and vertical axis, respectively.} 
\label{fig:posterior-PT-SP} 
\end{figure}


Figure \ref{fig:posterior-PT-SP} shows results for the Iberian Peninsula -- Portugal (blue) and Spain (red). 
These countries are distinct in both geographical position and energy policy, particularly their limited exposure to Russian gas and the implementation of the ``Iberian exception'' in 2022, which allowed them to cap gas prices \citep{hidalgo2024iberian}.
As a result, the posterior distributions for Portugal and Spain are relatively tight -- comparable to those of Norway -- and reflect lower overall uncertainty. 
This suggests that policy interventions may have successfully mitigated the impact of rising gas prices and contributed to greater price stability.

\begin{figure}[h!] 
\centering 
\includegraphics[width=\textwidth]{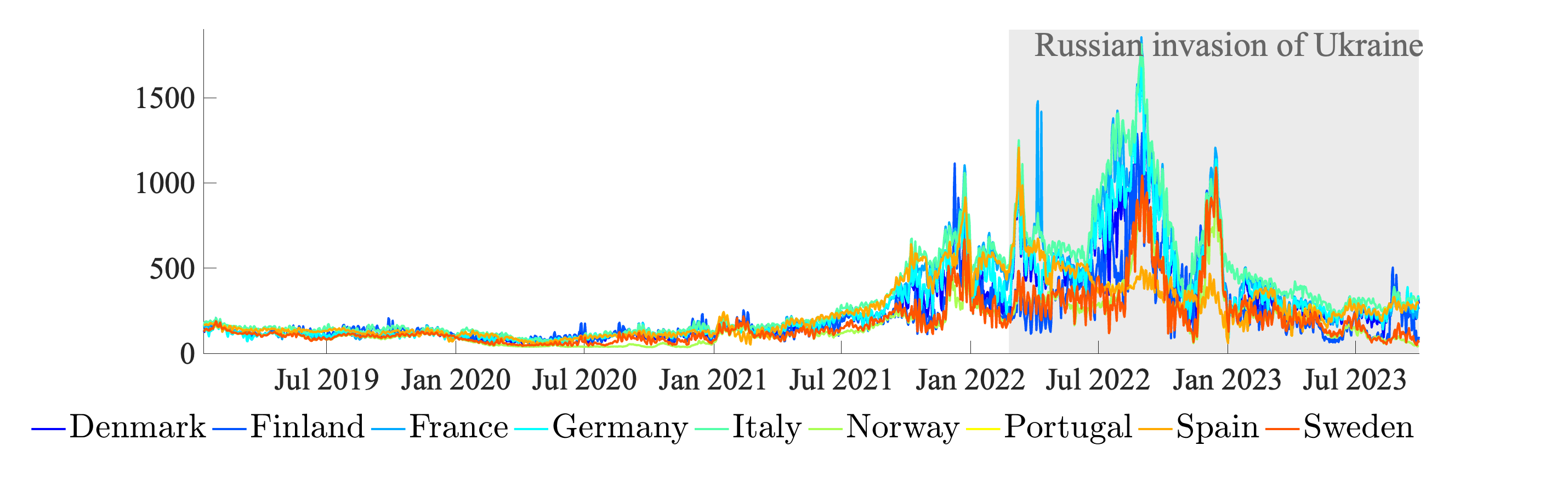} 
\caption{Estimated daily volatility for electricity prices across countries -- Denmark, Finland, France, Germany, Italy, Norway, Portugal, Spain, and Sweden -- from January 2019 to October 2023. The shaded area corresponds to the time from the beginning of the Russia's invasion of Ukraine.} 
\label{fig:joint-volatility} 
\end{figure}


As a further result, leveraging the Rao-Blackwellization presented in Proposition~\ref{sec:method:prop3}, we estimate the time-varying volatility for each country as:
\begin{align*}
    \widehat{\sigma}_{gh,t}^2 = \widehat{\sigma}_{gh}^{2}+ \widehat{\bz}_{g,t+h}^{\prime}(\widehat{R}+ \widehat{Q})\widehat{\bz}_{g,t+h}.
\end{align*}
This measure captures the conditional error variance, which increases substantially following Russia’s invasion of Ukraine, albeit heterogeneously across countries. 
Notably, Portugal and Spain (represented by the yellow and orange lines) exhibit the lowest volatility, consistent with their distinct market structures and limited exposure to Russian gas.

Two prominent peaks are observed in the volatility dynamics after the invasion: one during the summer of 2022 and another in the winter of 2022–2023. 
The first peak corresponds to a combination of severely reduced Russian gas flows to Europe and extreme weather conditions - droughts and heatwaves - which elevated electricity demand and prices. 
The second peak, though less intense, still reflects extreme price fluctuations. 
It is largely attributable to low renewable generation (wind and solar), high gas prices, and heightened uncertainty due to discussions around emergency governmental interventions, such as price caps and windfall taxes. 
Following these peaks, electricity price volatility begins to decline but remains above pre-crisis levels, suggesting ongoing market instability and a possible structural shift.

To sum up, we find considerable cross-country heterogeneity in both the magnitude of country-specific effects and the uncertainty of their posterior distributions. 
While electricity prices demonstrate strong co-movements across countries, our results indicate that national-level factors played a crucial role. 
For example, the ``Iberian exception'' introduced a price cap mechanism; Germany and Italy implemented broad subsidy schemes; and France, facing widespread outages in its nuclear fleet, was forced to meet domestic demand through increased electricity imports. 
These dynamics are effectively captured within the PRUMIDAS framework.

Therefore, while market integration has historically facilitated efficient price convergence and resilience under stable conditions, our findings reveal a critical vulnerability during periods of systemic shocks. 
The European energy crisis illustrates that integration can amplify and transmit external disruptions across interconnected markets, including those less directly exposed to the originating source of the shock. 
These results underscore the importance of building a more diversified and resilient energy infrastructure, alongside implementing targeted shielding mechanisms to mitigate the propagation of future crisis.

\section{Is market integration changing after the Russia's invasion of Ukraine?}
\label{sec:postwarresults}

In the previous section, we examined the effects of renewable energy sources and fossil fuel prices on electricity prices over the full sample period. However, as shown in Figure~\ref{fig:country-electricity-gas}, electricity and gas prices exhibit clear structural changes around the onset of the Russia-Ukraine war, suggesting a possible division into pre- and post-invasion periods.

To capture the specific dynamics of the crisis period, we re-estimate our model using data from January 2022 to October 2023, excluding the more stable phase. This allows us to assess whether the relations between electricity prices and forecasted renewable energy sources and gas prices intensifies under conditions of heightened market stress and whether the initial calm period dampens these effects in the aggregate analysis in Section~\ref{sec:analysis-application}.
We focus on the estimated posterior distributions of country-specific effects associated with forecasted solar and wind generation, as well as natural gas prices. Figure~\ref{fig:posterior-DE-FE-GE-IT_postwar}, \ref{fig:posterior-FI-NO-SW_postwar} and \ref{fig:posterior-PT-SP_postwar} present the results for the country grouped as in Section~\ref{sec:analysis-application}, with top panels corresponding to the full sample period and the bottom ones to the post-invasion\footnote{The results for the other fossil fuel prices and forecasted demand are available in the Supplement.}.

\begin{figure}[h!]
\centering
\setlength{\tabcolsep}{0.001pt}
\begin{tabular}{lll}
    \includegraphics[width=5.2cm]{270525_Figures/ksdensity_full/Solar-DE-FR-GE-IT-ksdensity.eps} &
    \includegraphics[width=5.2cm]{270525_Figures/ksdensity_full/Wind-DE-FR-GE-IT-ksdensity.eps} &
    \includegraphics[width=5.2cm]{270525_Figures/ksdensity_full/Gas-DE-FR-GE-IT-ksdensity.eps} \\
     \includegraphics[width=5.2cm]{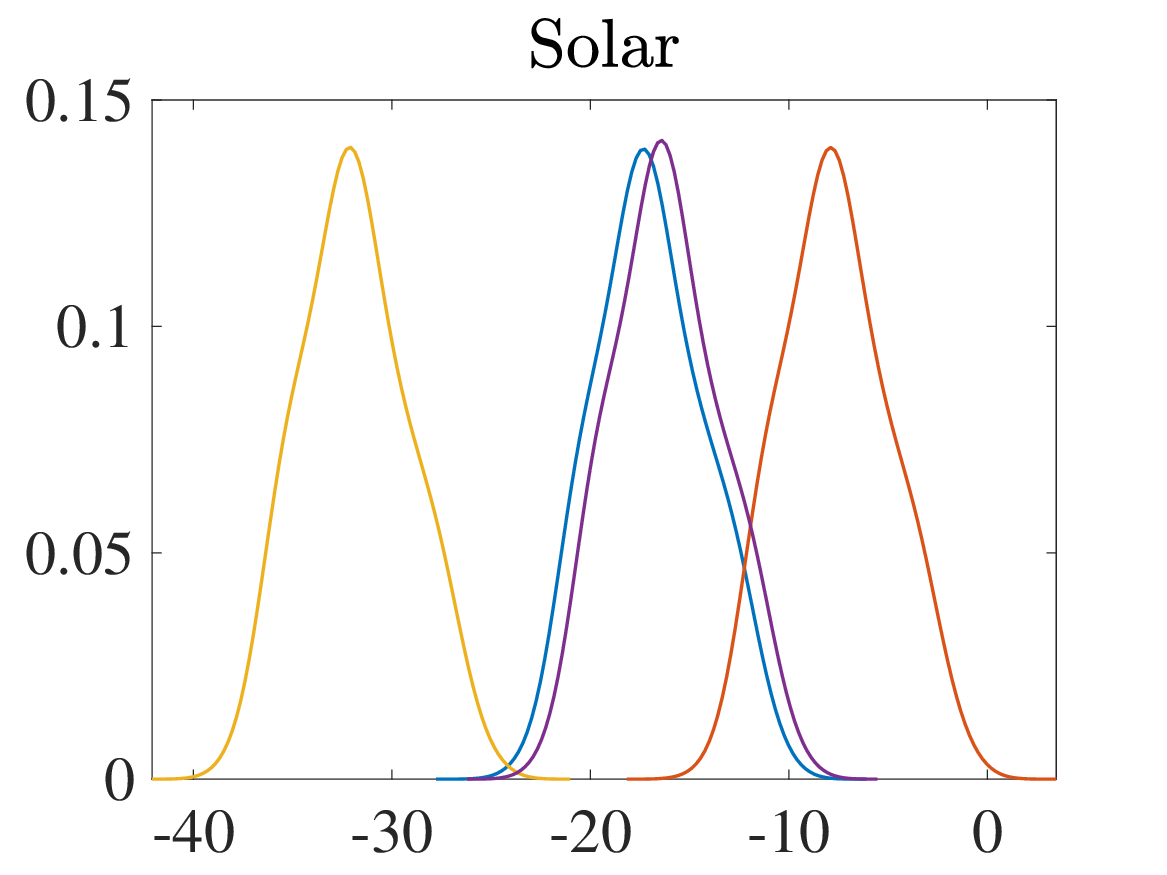} &
    \includegraphics[width=5.2cm]{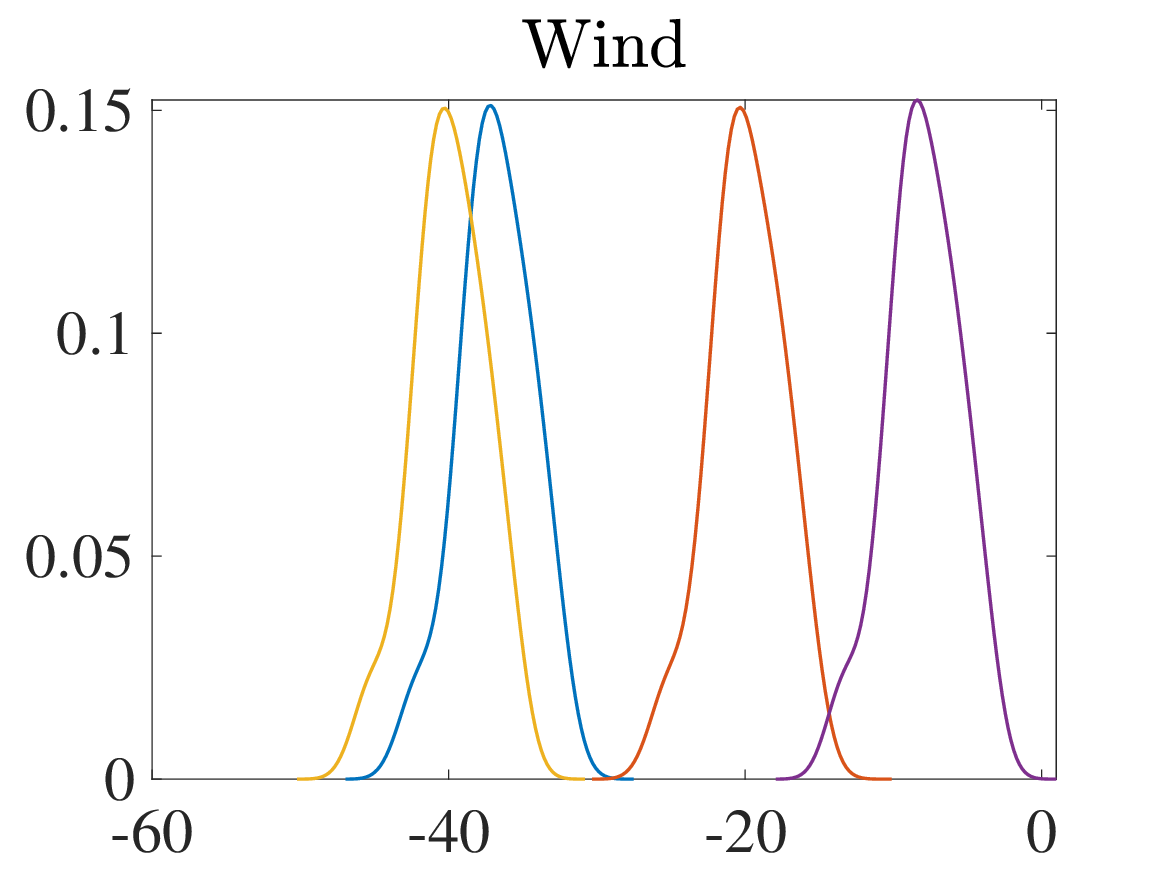} &
    \includegraphics[width=5.2cm]{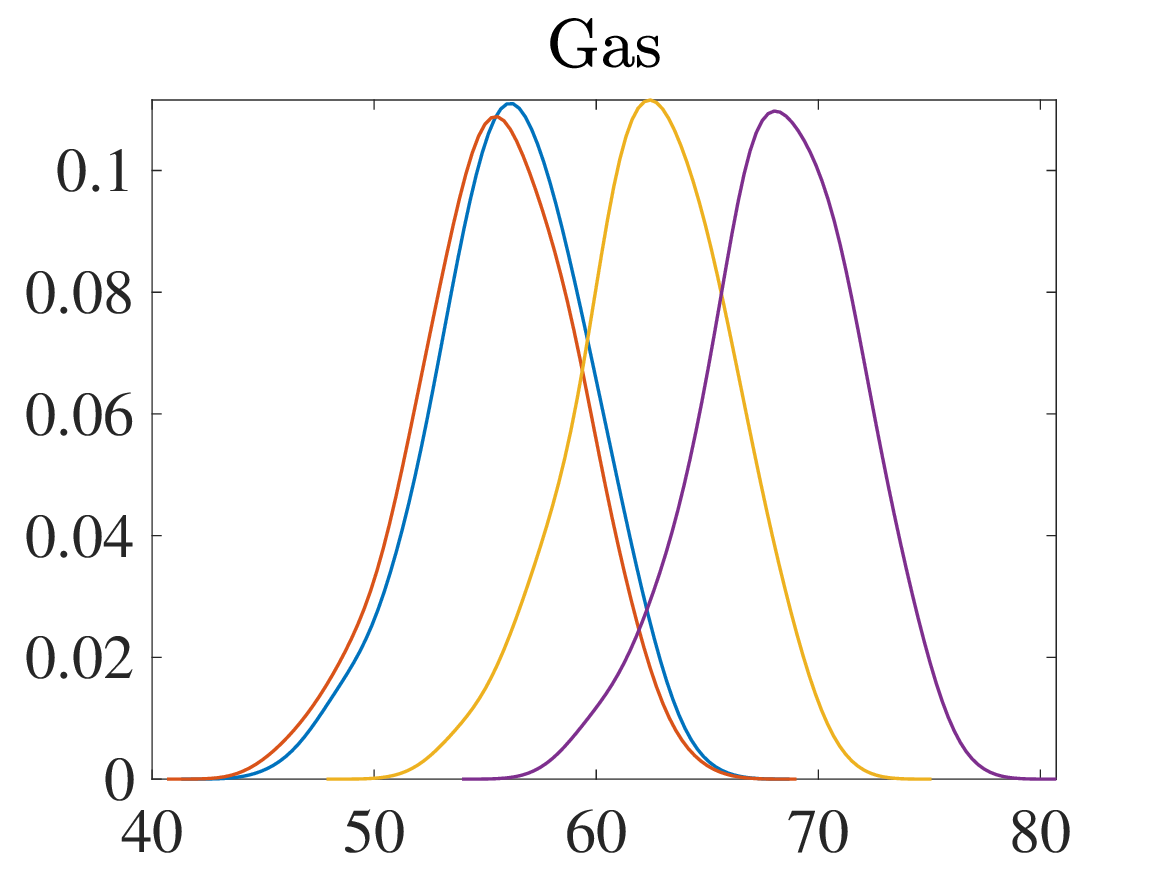}    
    \end{tabular}
    \caption{Full conditional posterior distributions for the country-effects for solar (left), wind (center), and gas (right) -- for Denmark (blue), France (red), Germany (yellow), and Italy (violet). Top panels refer to the full sample data and bottom to the post-invasion period.}
    \label{fig:posterior-DE-FE-GE-IT_postwar}
\end{figure}

Figure~\ref{fig:posterior-DE-FE-GE-IT_postwar}, covering Denmark (blue), France (red), Germany (yellow), and Italy (violet), shows substantial changes in the magnitude of the estimated coefficients across the two periods.
The sign of the effects remains consistent: forecasted solar (left panel) and wind generation (center panel) continue to have a negative effects on electricity prices. However, the magnitudes of these effects increase notably in the post-invasion period. 
The impact of gas prices, in particular, rises substantially from approximately $40$ to $55$ in Denmark and from $50$ to $70$ in Italy. 
Moreover, we observe a stronger alignment in the gas price effects between France and Denmark during the crisis period, while differences between Germany and Italy diminish, suggesting a potential convergence in their price transmission mechanisms under stress.

Turning to the Nordic countries, Figure~\ref{fig:posterior-FI-NO-SW_postwar} reveals mixed results. In Finland, the effects of solar generation shift from a positive to a negative influence on electricity prices.
This change may reflect Finland's reduced reliance on Russian gas following the onset of the war.
Moreover, increased alignment in the gas price effects between Norway and Sweden -- absent in the full sample estimates -- points to a greater degree of integration within the Nordic electricity market during the crisis.

\begin{figure}[h!]
\centering
\setlength{\tabcolsep}{0.001pt}
\begin{tabular}{lll}     
    \includegraphics[width=5.2cm]{270525_Figures/ksdensity_full/Solar-FI-NO-SW-ksdensity.eps} &
    \includegraphics[width=5.2cm]{270525_Figures/ksdensity_full/Wind-FI-NO-SW-ksdensity.eps} &
    \includegraphics[width=5.2cm]{270525_Figures/ksdensity_full/Gas-FI-NO-SW-ksdensity.eps} \\
     \includegraphics[width=5.2cm]{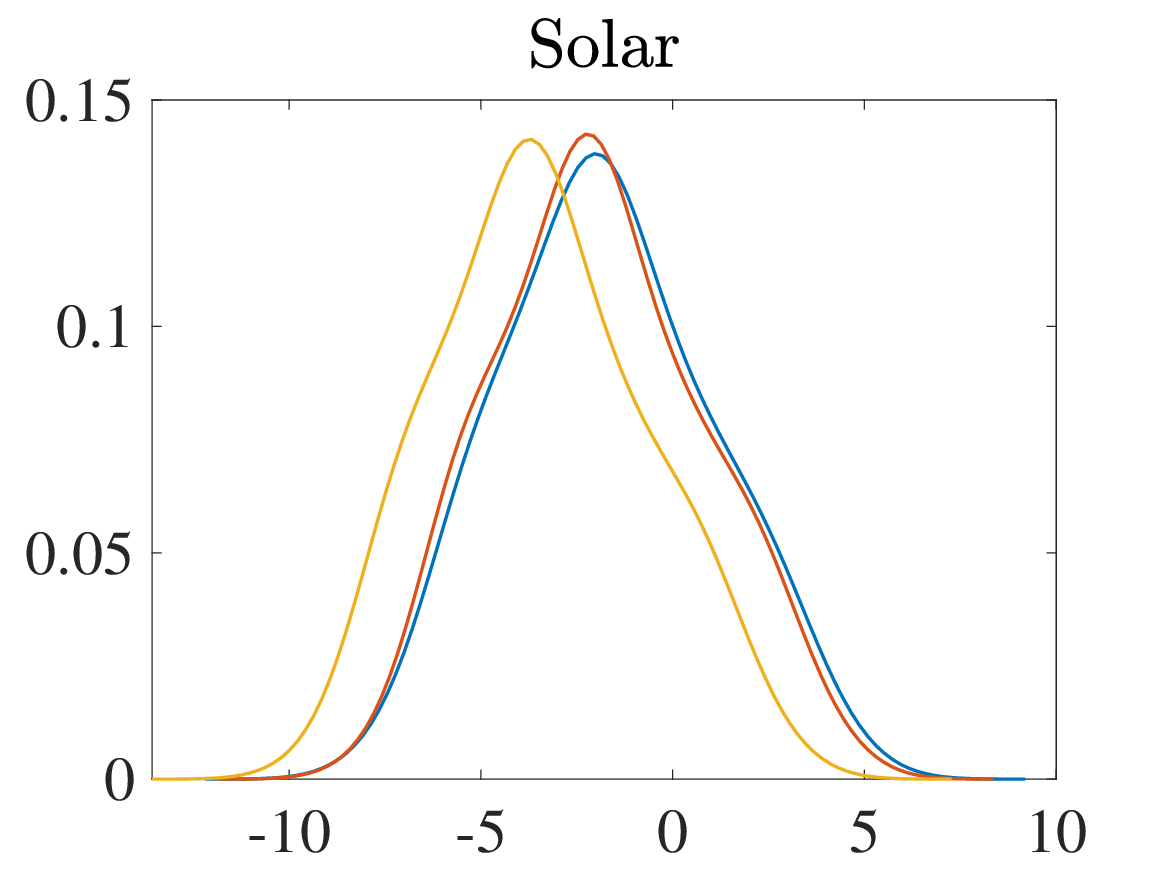} &
    \includegraphics[width=5.2cm]{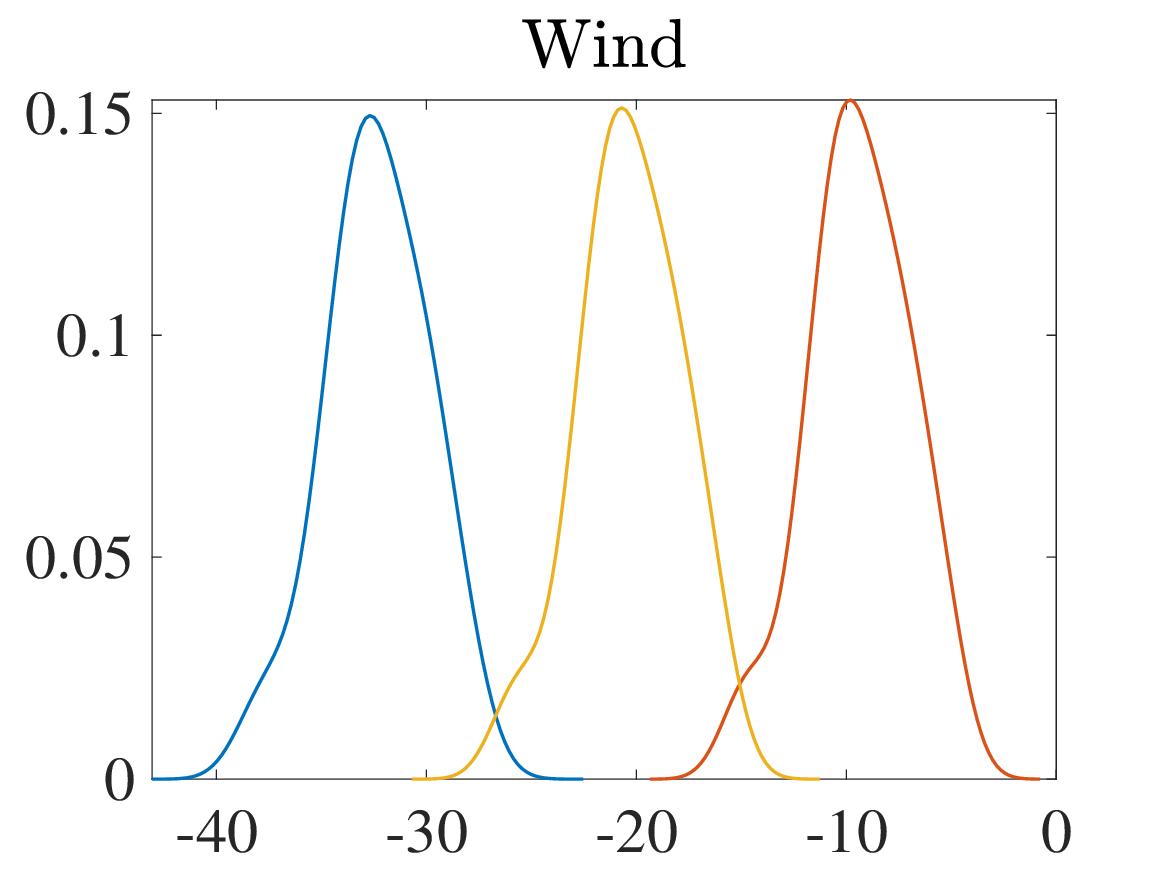} &
    \includegraphics[width=5.2cm]{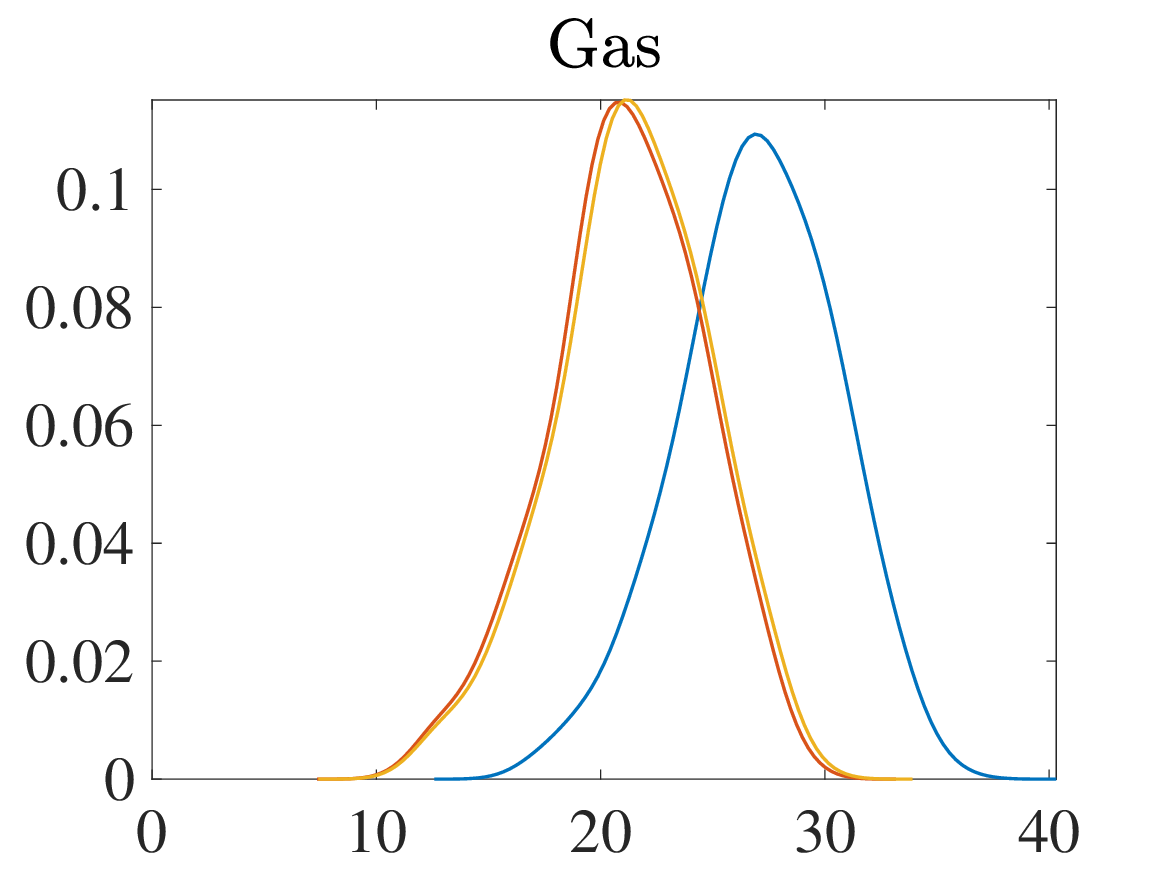}    
    \end{tabular}
    \caption{Full conditional posterior distributions for the country-effects for solar (left), wind (center), and gas (right) -- for Finland (blue), Norway (red), and Sweden (yellow). Top panels refer to the full sample data and bottom to the post-invasion period.}
    \label{fig:posterior-FI-NO-SW_postwar}
\end{figure}

Finally, Figure~\ref{fig:posterior-PT-SP_postwar} displays results for Portugal and Spain and they are consistent with the ``Iberian exception'', whereby the two countries implemented a joint intervention to cap the gas prices. 
In detail, Portugal and Spain show similar negative effects of solar and wind generation on electricity prices in the post-invasion period. 
Importantly, in both countries, the posteriors for the effects of gas prices on electricity prices shifts from being completely positive to allocating more probability to negative effects. 
This reversal appears to reflect the implementation of the gas price cap, which effectively decoupled wholesale electricity prices from gas price spikes.
As shown in Figure~\ref{fig:joint-volatility}, Portugal and Spain were less exposed to the extreme electricity price volatility seen elsewhere in Europe, benefiting from higher share of renewables energy sources in their energy mix.
Finally, while greater reliance on renewables reduces exposure to fossil fuel, it may introduce new vulnerabilities as illustrated by the widespread blackout that affected Spain on 28 May 2025, underscoring the need for grid resilience and flexible backup capacity as renewable penetration increases.

\begin{figure}[h!]
\centering
\setlength{\tabcolsep}{0.001pt}
\begin{tabular}{lll}     
    \includegraphics[width=5.2cm]{270525_Figures/ksdensity_full/Solar-PT-SP-ksdensity.eps} &
    \includegraphics[width=5.2cm]{270525_Figures/ksdensity_full/Wind-PT-SP-ksdensity.eps} &
    \includegraphics[width=5.2cm]{270525_Figures/ksdensity_full/Gas-PT-SP-ksdensity.eps} \\
     \includegraphics[width=5.2cm]{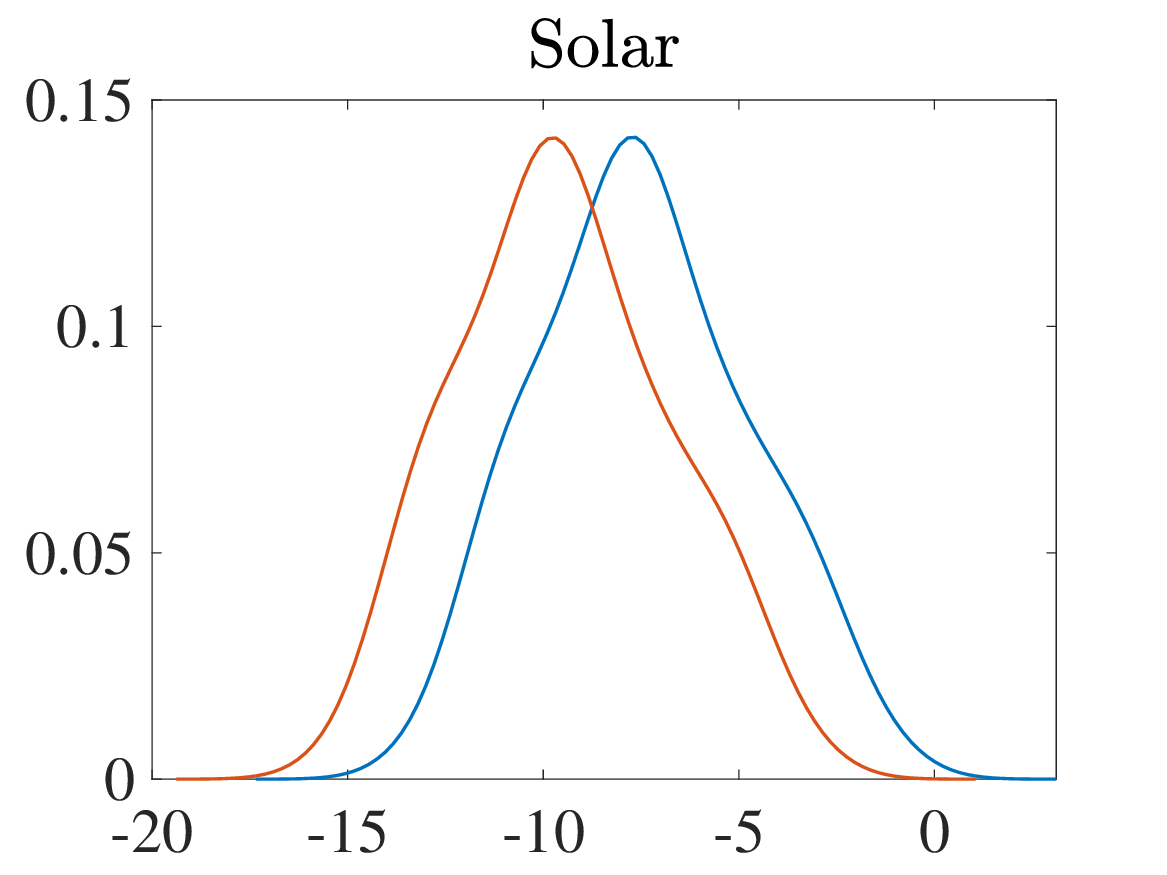} &
    \includegraphics[width=5.2cm]{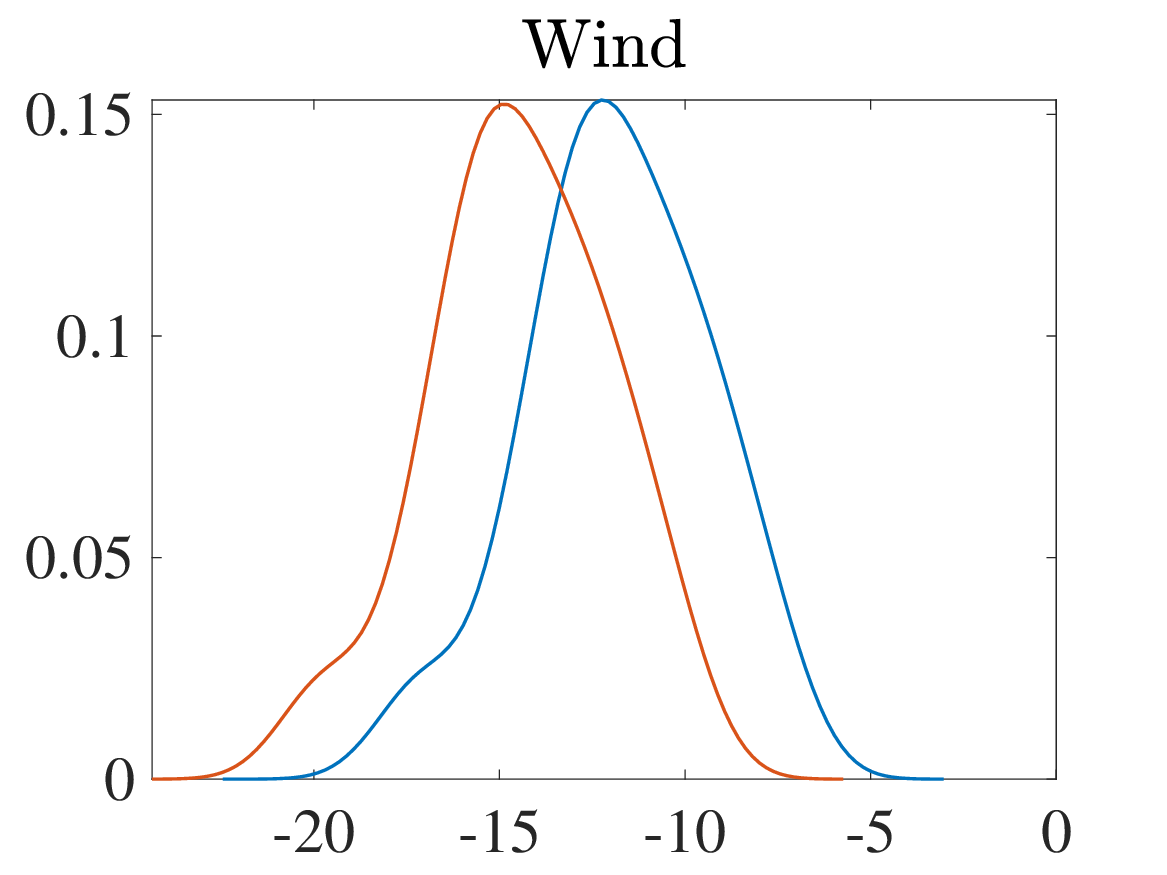} &
    \includegraphics[width=5.2cm]{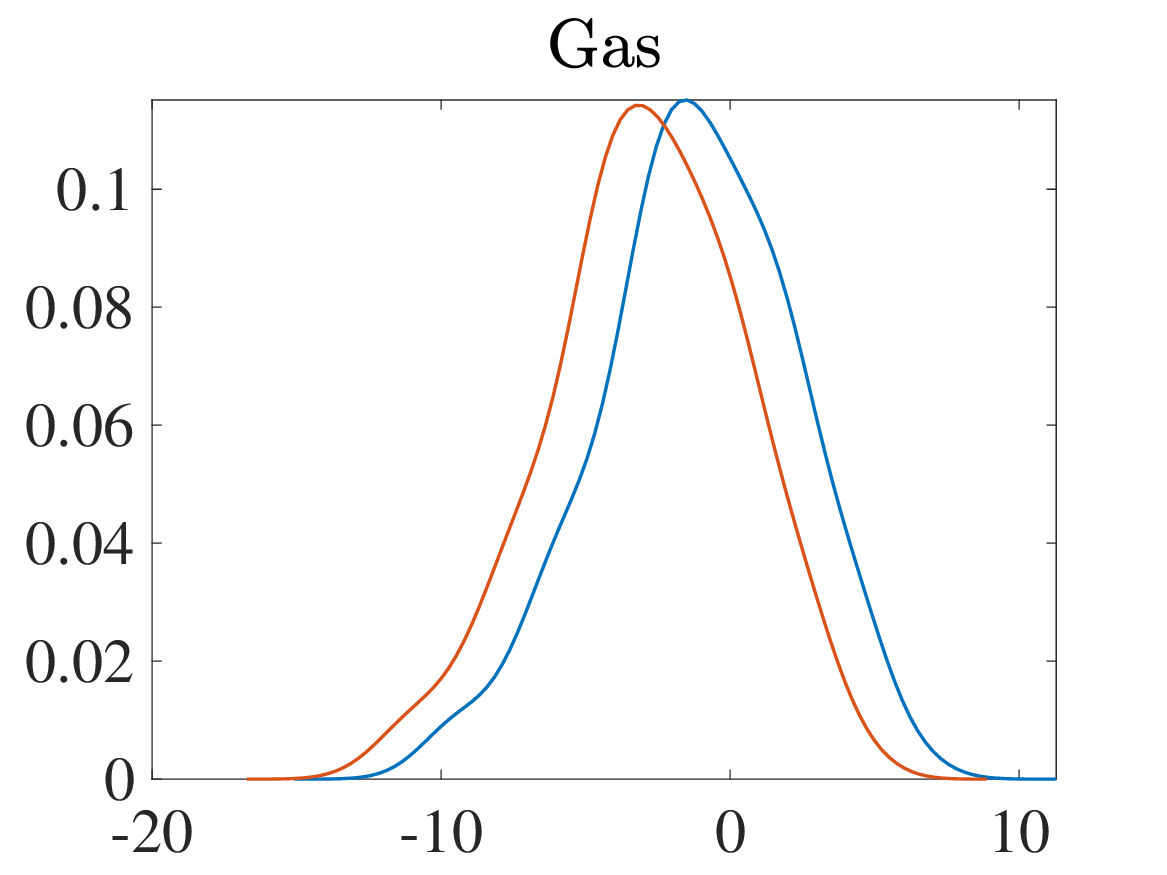}    
    \end{tabular}
    \caption{Full conditional posterior distributions for the country-effects for solar (left), wind (center), and gas (right) -- for Portugal (blue)  and Spain (red). Top panels refer to the full sample data and bottom to the post-invasion period.}
    \label{fig:posterior-PT-SP_postwar}
\end{figure}

\section{Conclusions}
\label{sec:concl}

The recent energy crisis unequivocally underscored the critical role of market integration as both a safeguard for energy security and a fundamental driver of European competitiveness and economic growth. In response, the European Commission has initiated significant efforts to rapidly reduce short-term fossil fuel dependence and accelerate the transition to renewable energy sources.
To address these complex challenges, our study employed a novel panel reverse unrestricted MIDAS (PRUMIDAS) model. This allowed us to examine the impact of daily fossil fuel prices and hourly renewable energy generation on hourly electricity prices, explicitly capturing cross-country interdependencies and incorporating country-specific market characteristics.

Our findings reveal that while renewable energy sources consistently exert downward pressure on electricity prices and contribute to market integration, their impact remains primarily at the national level. 
In contrast, fossil fuel prices are linked to increased price volatility, which can undermine integration by inducing market divergence.

The energy crisis, triggered by Russia's invasion of Ukraine, served as a revealing case study, demonstrating both the inherent strengths of integrated markets in a geopolitical shock and their vulnerabilities to rapid price transmission across national boundaries. These results underscore the paramount need for coordinated policy frameworks and strategic energy diversification to fortify the resilience of European electricity markets against external shocks.

Looking ahead, European countries must collaborate to foster integration and forge a unified strategy to efficiently reshape the EU's energy infrastructure -- spanning generation, storage, transmission, and the role of the merit order criterion -- into a single energy market. Crucially, the EU must strengthen institutional and financing frameworks, and pool resources to accelerate innovation and mitigate risks to the single market (\citetalias{electricity2024analysis}, \citeyear{electricity2024analysis}).

\newpage

\bibliographystyle{chicago}
\bibliography{reference}

\newpage
\appendix

\renewcommand{\theproposition}{\thesection.\arabic{proposition}}
\renewcommand{\theequation}{\thesection.\arabic{equation}}
\setcounter{equation}{0}
\setcounter{proposition}{0}

\section{Proofs}
\label{sec:appendix}

We build on \cite{casarin2018uncertainty} and apply their results to our proposed PRUMIDAS model.

We denote $\by_{g,t+h-1}\equal(y_{g,t+h-1}\comma\ldots\comma y_{g,t+h-A})^{\prime}$ as the vector of lagged values of the high-frequency dependent variable, $\bx_{gj,t+h_j} = (x_{gj,t+h_j},\ldots,x_{gj,t+h_j-B_jH_j}\rightbr^{\prime}$ as the vector of lagged values of either high- or low-frequency covariates, depending on the value of $H_j$ and $h_j$, as defined in Section~\ref{sec:Notation}.
We define $\bz_{g,t+h} = \leftbr1,\by_{g,t+h-1}^{\prime}\comma\bx_{g1,t+h_1}^{\prime}\comma\ldots\comma\bx_{gN,t+h_N}^{\prime}\rightbr^{\prime}$ as the vector including both the autoregressive high-frequency and lagged covariates at both frequency. 

Moreover, let $L = \bigl(1+A+\sum_{j=1}^{N}(1+B_j)\bigr)$ be the size relative to the number of covariates included in the model. Then, denote as $\bgamma = \leftbr \mu,\balpha^{\prime},\bbeta_{1}^{\prime}\comma\ldots\comma\bbeta_{N}^{\prime})^{\prime}$ the $L-$vector of common effects, including the common intercept and common coefficients, $\balpha = (\alpha_{1}\comma\ldots\comma\alpha_{A})^{\prime}$ and $\bbeta_{j} = \leftbr \beta_{j0}, \beta_{j1}\comma\ldots\comma\beta_{jB_j}\rightbr^{\prime}$. 
Let $\bpsi_{h} = (\psi_{\mu,h}\comma\bpsi_{\alpha,h}^{\prime}\comma\bpsi_{\beta,1h}^{\prime}\comma\ldots\comma\bpsi_{\beta,Nh}^{\prime}\rightbr^{\prime}$ be the $L-$vector of hourly random effects, with $\bpsi_{\alpha,h} = \leftbr\psi_{\alpha,h1}\comma\ldots\comma\psi_{\alpha,hA}\rightbr^{\prime}$ be a $A$-vector and $\bpsi_{\beta,jh} = \leftbr\psi_{\beta,jh0}\comma\ldots\comma\psi_{\beta,jhB_j}\rightbr^{\prime}$ a $(1+B_j)$-vector for $j=1,\ldots,N$. 
Then, let $\bzeta_{g} = (\zeta_{\mu,g},\bzeta_{\alpha,g}^{\prime},\bzeta_{\beta,g1}^{\prime},\ldots,\bzeta_{\beta,gN}^{\prime}\rightbr^{\prime}$ the $L-$vector of country random effects, where $\bzeta_{\alpha,g} = \leftbr\zeta_{\alpha,g1},\ldots,\zeta_{\alpha,gA}\rightbr^{\prime}$ is $A$--vector and $\bzeta_{\beta,gj} = \leftbr\zeta_{\beta,gj0}\comma\ldots\comma\zeta_{\beta,gjB_j}\rightbr^{\prime}$ is $(1+B_j)$--vector for $j=1,\ldots,N$.

\begin{proposition}
\label{sec:appendix:propA1}
Based on the previous definitions, the model in Eq.~\eqref{sec:method:second-eq} can be expressed in stacked form as
\begin{equation}
    \label{sec:appendix:eqA1}
    y_{g,t+h} = \bz_{g,t+h}^{\prime}(\bgamma+\bzeta_{g}+\bpsi_{h}) + e_{g,t+h}.
\end{equation}
where $\bgamma$, $\bzeta_{g}$, and $\bpsi_h$ are the vectors of common, country, and hourly random effects, respectively.
\end{proposition}

\begin{proof}
Substitute the coefficients in Eq.~\eqref{sec:method:first-eq} with their specifications in Eq.~\eqref{sec:method:cere}:
\begin{align*}
    y_{g,t+h} 
    &= \bigl(\mu+\psi_{\mu,h}+\zeta_{\mu,g}\bigl)
    +\sum_{a=1}^{A}\bigl(\alpha_{a}+\psi_{\alpha,ha}+\zeta_{\alpha,ga}\bigl)y_{g,t+h-a} \\
    &+\sum_{j=1}^{N}\sum_{b=0}^{B_j}\bigl(\beta_{jb}+\psi_{\beta,jhb}+\zeta_{\beta,gjb}\bigl)x_{gj,t+h_j-bH_j}+e_{g,t+h}.
\end{align*}
Finally, re-write the summations in vector form, as introduced in the preamble:
\begin{align*}
    y_{g,t+h}
    &= \bigl(\mu+\psi_{\mu,h}+\zeta_{\mu,g}\bigl)
    +\bigl(\balpha+\bpsi_{\alpha,h}+\bzeta_{\alpha,g}\bigl)\by_{g,t+h-1}  \\
    &+\sum_{j=1}^{N}\bigl(\bbeta_{j}+\bpsi_{\beta,jh}+\zeta_{\bbeta,gh}\bigl)x_{gj,t+h_j}+e_{g,t+h} \\
    &= \bz_{g,t+h}^{\prime}\bigl(\bgamma+\bpsi_{h}+\bzeta_{g}\bigl)+e_{g,t+h}.
\end{align*}
\end{proof}

\begin{proposition}
\label{sec:appendix:propA2}
The stack-form model in Eq.~\eqref{sec:appendix:eqA1} admits a ``low-frequency'' panel vector autoregressive representation as
\begin{equation}
\label{sec:appendix:eqA2}
    \by_{gt} = Z_{gt}\bigl(\bgamma+\bzeta_{g}\bigl)+\diag\bigl(Z_{gt}\Psi^{\prime}\bigl)+\be_{gt},
\end{equation}
where $Z_{gt} = \leftbr\bz_{g,t+0}\comma,\ldots,\comma\bz_{g,t+H-1}\rightbr^{\prime}$ and $\Psi = \leftbr\bpsi_0\comma,\ldots\comma\bpsi_{H-1}\rightbr^{\prime}$ are matrices of size $H\times L$, $\by_{gt} = \leftbr y_{g,t+0},\ldots,y_{g,t+H-1}\rightbr$, and $\be_{gt} \sim \mathcal{N}_H(0,\Sigma_g)$ with $\Sigma_g = \diag(\sigma_{g,0}^2,\ldots, \sigma_{g,H-1}^2)$.
\end{proposition}

\begin{proof}
Let us stack Eq.~\eqref{sec:appendix:eqA1} for $h=1,\ldots,H-1$ to obtain a system of equations. 
Then, let $\Psi$ be such that $( \bz_{g,t+0}^{\prime}\bpsi_0,\ldots,\bz_{g,t+H-1}^{\prime}\bpsi_{H-1}\rightbr = \leftbr Z_{gt}\Psi^{\prime}\rightbr\odot I_{H} = \diag\leftbr Z_{gt}\Psi^{\prime}\rightbr$ is a $H\times 1$ vector.
Please refer to \cite{magnus2019matrix}, page 52, for the properties of the Hadamard product, denoted by $\odot$.
Therefore, from Eq.~\eqref{sec:appendix:eqA1} it follows that
\begin{equation*}
    \by_{gt} = Z_{gt}\bigl(\bgamma+\bzeta_{g}\bigl)+\diag\bigl(Z_{gt}\Psi^{\prime}\bigl)+\be_{gt},
\end{equation*}
where $\be_{gt}=\leftbr e_{g,t+0},\ldots,e_{g,t+H-1})^{\prime}$ and $\be_{gt}\sim\mathcal{N}_{H}\leftbr 0,\Sigma_{g})$ with $\Sigma_{g} = \diag(\sigma_{g,0}^2\comma\ldots,\sigma_{g,H-1}^2)$, where $\mathcal{N}_{H}$ denotes the $H$-dimensional multivariate normal distribution.
\end{proof}

Following Section~\ref{sec:model:common-random-effects}, we assume that $\bpsi_{h}\sim\mathcal{N}_{L}\leftbr \mathbf{0}_{L},Q)$ and $\bzeta_{g}\sim\mathcal{N}_{L}\leftbr \mathbf{0}_{L},R)$, where $Q = \diag\{\bigl\leftbr q_{\mu},\biota_{A}^{\prime}\otimes q_{\alpha},\biota_{N\sum_j(1+B_j)}^{\prime}\otimes q_{\beta}\bigr\rightbr^{\prime}\}$, 
$R = \diag\{\bigl\leftbr r_{\mu},\biota_{A}^{\prime}\otimes r_{\alpha},\biota_{N\sum_j(1+B_j)}^{\prime}\otimes r_{\beta}\bigr\rightbr^{\prime}\}$, and $\biota_A^{\prime}$ and $\biota_N^{\prime}$ are the vector of ones of size $A$ and $N$, respectively.
\begin{proposition}
\label{sec:appendix:propA3}
By marginalizing the random effects in Eq.~\eqref{sec:appendix:eqA1}, the model simplifies to
\begin{equation}
\label{sec:appendix:eqA3}
    y_{g,t+h}= \bz_{g,t+h}^{\prime}\bgamma+\Tilde{e}_{g,t+h},
\end{equation}
where $\Tilde{e}_{g,t+h}\sim\mathcal{N}(0,\sigma_{ght}^{2})$, $\sigma_{ght}^{2} = \sigma_{gh}^{2}+ \bz_{g,t+h}^{\prime}(R+Q)\bz_{g,t+h}$.
\end{proposition}

\begin{proof}
Since $e_{g,t+h}$ is normally distributed and conditionally on the coefficients in Eq.~\eqref{sec:appendix:eqA1} it follows that $( y_{g,t+h}|\bgamma,\bpsi_{h},\bzeta_{g}\rightbr\sim\mathcal{N}\bigl\leftbr\bz_{g,t+h}^{\prime}(\bgamma+\bzeta_{g}+\bpsi_{h}),\sigma_{gh}^{2}\bigr\rightbr$. Thus, using the moment generating function, we can integrate out the random effects and obtain the marginal density of $y_{g,t+h}|\bgamma$ as follows:
\begin{align*}
    \Ev\bigl(\exp\{\theta y\}\bigl) &= \int \exp\{\theta y\} f(y|\bgamma)\,dy \\
    &= \int \exp\{\theta y\} f\bigl(y|\bgamma,\bpsi_{h},\bzeta_{g}\bigl)f(\bzeta_{g})f(\bpsi_{h})\,dy\,d\bpsi_{h}\,d\bzeta_{g} \\    
    &= \int \exp\Bigl\{\theta\bigl[\bz_{g,t+h}^{\prime}\bigl(\bgamma+\bpsi_{h}+\bzeta_{g}\bigl)\bigl]+\frac{1}{2}\theta^{2}\sigma_{gh}^{2}\Bigl\} f(\bpsi_{h})\,f(\bzeta_{g})\,d\bpsi_{h}\,d\bzeta_{g} \\
    &= \exp\Bigl\{\theta\bz_{g,t+h}^{\prime}\bgamma\Bigl\}\exp\Bigl\{\theta\bigl[\bz_{g,t+h}^{\prime}(\bpsi_{h}+\bzeta_{g}\bigl]+\frac{1}{2}\theta^2\sigma_{gh}^{2}\Bigl\} \\
    &= \exp\Bigl\{\theta\bz_{g,t+h}^{\prime}\bgamma\Bigl\}\exp\Bigl\{\frac{1}{2}\theta^{2}\bigl(\sigma_{gh}^{2}+\bz_{g,t+h}^{\prime}(Q+R)\bz_{g,t+h}\bigl)\Bigl\}
\end{align*}
which is the moment generating function of normal random variable with mean $\bz_{g,t+h}^{\prime}\bgamma$ and variance $\sigma_{ght}^{2} = \sigma_{gh}^{2}+ \bz_{g,t+h}^{\prime}(R+Q)\bz_{g,t+h}$. 

Additionally, this result can be applied to Eq.~\eqref{sec:appendix:eqA3} since it is a stacked version of Eq.~\eqref{sec:appendix:eqA2}. Hence, it follows that $\by_{gt}|\bgamma$ is normally distributed with mean $\bm{z}'_{g,t+h}\bgamma$ and variance $\Sigma_{gt} = \diag(\sigma_{g0t}^2\comma\ldots,\sigma_{g(H-1)t}^2)$.
\end{proof}

\end{document}